\documentclass[12pt]{article}

\usepackage{amssymb}
\usepackage{amsmath}
\usepackage{amsfonts}
\usepackage{eurosym}
\usepackage{dsfont}
\usepackage{mathrsfs}
\usepackage{bbm}
\usepackage{bm}
\usepackage{fix-cm}
\usepackage{accents}

\usepackage{geometry}
\usepackage{setspace}
\usepackage{color}
\usepackage{enumitem}
\usepackage{soul}

\usepackage{float}
\usepackage{graphicx}
\usepackage{caption}
\usepackage{subcaption}
\usepackage{footmisc}
\usepackage{tablefootnote}
\usepackage{booktabs}

\usepackage{appendix}

\usepackage{tikz}
\usepackage{tikz-3dplot}
\usepackage{tikz-network}
\usetikzlibrary{backgrounds,automata,arrows,positioning,patterns,calc,shapes}
\usetikzlibrary{decorations.pathreplacing}
\usetikzlibrary{positioning}

\usepackage[colorlinks=true, urlcolor=purple, linkcolor=blue,citecolor=purple]{hyperref}
\usepackage{natbib}

\usepackage{datetime}
\newdateformat{monthyeardate}{\monthname[\THEMONTH], \THEYEAR}

\newtheorem{theorem}{Theorem}
\newtheorem{proposition}{Proposition}
\newtheorem{lemma}{Lemma}
\newtheorem{corollary}{Corollary}

\newtheorem{assumption}{Assumption}
\newtheorem{definition}{Definition}
\newtheorem{remark}{Remark}
\newtheorem{example}{Example}

\newenvironment{proof}[1][Proof]{\noindent\textbf{#1.} }{\ \rule{0.5em}{0.5em}}

\allowdisplaybreaks

\newcommand{\bG}{\mathbf{G}}
\newcommand{\bM}{\mathbf{M}}
\newcommand{\bH}{\mathbf{H}}
\newcommand{\mN}{\mathcal{N}}
\newcommand{\mS}{\mathcal{S}}
\newcommand{\se}{\boldsymbol{\mu}^*}
\newcommand{\w}{\boldsymbol{\omega}}
\newcommand{\btheta}{\boldsymbol{\theta}}
\newcommand{\bx}{\mathbf{x}}
\newcommand{\bX}{\mathbf{x}}
\newcommand{\bp}{\mathbf{p}}
\newcommand{\bt}{\mathbf{t}}
\newcommand{\etab}{\boldsymbol{\eta}}
\newcommand{\ubar}[1]{\underaccent{\bar}{#1}}
\newcommand{\btau}{\boldsymbol{\tau}}
\newcommand{\bmu}{\boldsymbol{\mu}}
\newcommand{\parallelsum}{\mathbin{\!/\mkern-5mu/\!}}
\newcommand{\lp}{\mathbbm{p}}
\newcommand{\brho}{\boldsymbol{\rho}}
\newcommand{\bB}{\mathbf{B}}

\begin{document}

\title{\textsc{When Do Markets Work? Multiplex Networks and Efficiency}\thanks{We thank the participants at the workshop ``Advances in Economic Theory and Beyond'' at Tsinghua University (September 2026), the Economics of Networks Workshop at Stony Brook (July 2026), the 11th Annual Conference on Network Science in Economics at the University of Miami  (April 2026), the Workshop on Networks at Tsinghua University (April 2026), and the  economic seminar at the Capital University of Economics and Business (April 2026). We are especially grateful to Robert M. Anderson, Francis Bloch, Krishna Dasaratha,  Haosui Duanmu, Matt Elliott, Andrea Galeotti, Christian Ghiglino, Sanjeev Goyal, Paul Heidhues,  Alastair Langtry, Jian Li, Van-Quy Nguyen, John Quah, Sudipta Sarangi, and Mu Zhang for their helpful comments and suggestions.}}
\author{
	Chengqing Li\thanks{School of Economics and Management, Tsinghua University. lcq22@mails.tsinghua.edu.cn} \quad 
    Yves Zenou\thanks{Department of Economics, Monash University, Australia.  yves.zenou@monash.edu.}\quad
	Junjie Zhou\thanks{School of Economics and Management, Tsinghua University. zhoujj03001@gmail.com}
    }

\date{October 1, 2026}
\maketitle
	
\begin{abstract}
We study an Arrow--Debreu economy with externalities generated by multiplex
networks. Market equilibrium prices reflect not only preferences and the
scarcity of goods, but also consumers' network positions induced by
goods-specific externalities and the linkages across goods arising from the
endogenous budget constraint. We quantify the departures of equilibrium
allocations from Pareto efficiency using both a utility-based measure of
efficiency loss and a coefficient of resource utilization, thereby linking
network multiplexity quantitatively to equilibrium performance. These measures
also identify zero-efficiency-loss cases: competitive equilibria are efficient
when the relevant network-centrality vectors are parallel, a condition
satisfied, in particular, when all networks are regular or when all layers share
the same network structure. When competitive markets allocate goods
inefficiently, a Lindahl equilibrium implemented through personalized prices
can restore efficiency, although it may leave some consumers worse off.




\end{abstract}
	
\clearpage


    



\section{Introduction}


Externalities arise in many economic activities that are embedded in multiple social relationships. For example, an individual’s investment in education may be influenced by that of classmates or peers, while decisions related to housing purchases are more likely to be affected by the choices of family members or close friends. Furthermore, different activities exhibit distinct externality properties. Public goods provided by others generate positive spillovers, yet they simultaneously diminish an agent's incentive to contribute. In contrast, visible consumption, such as clothing, often induces social comparison; observing friends with high-quality attire may make an individual feel out of fashion, thereby stimulating their desire to purchase new clothes. Although externalities have been extensively studied in general equilibrium literature, relatively few works specifically focus on environments with multiple social relationships and heterogeneity of externality types across these relationships. Consequently, individual behavior and the equilibrium outcomes remain underexplored in the presence of multiplex networks.

Our main contribution is to quantify and measure the efficiency losses generated by externalities in multiplex networks and to characterize when these losses disappear. To this end, we introduce both a utility-based measure of efficiency loss and a coefficient of resource utilization. These measures allow us to assess quantitatively how network multiplexity shapes equilibrium performance and the extent to which equilibrium allocations depart from Pareto efficiency. They also identify the conditions under which the efficiency loss vanishes. In particular, competitive equilibria are efficient when the relevant network-centrality vectors are parallel, a condition satisfied when all networks are regular or when all layers share the same network structure. Thus, rather than simply concluding that competitive markets fail in the presence of externalities, our analysis quantifies the magnitude of this failure and identifies the network structures under which it disappears altogether.


Specifically, we consider a pure exchange economy in which externalities are structured by multiplex networks. The economy features multiple goods, each corresponding to a layer in the multiplex network, and network structures may differ across goods. Within each layer, consumers are subject to (positive or negative) externalities that aggregate linearly  and need not be identical across goods. Agents are heterogeneous in their endowments and in their network positions across layers. Each agent chooses the consumption of each good subject to a budget constraint, taking prices as given. The market then determines the equilibrium prices and the allocation of goods across consumers.

We first establish the existence of a competitive equilibrium. Provided that the intensity of negative externalities is sufficiently small, Proposition~\ref{prop.exist} demonstrates existence using the social equilibrium technique of \cite{debreuSocialEquilibriumExistence1952}. Under stronger conditions, we show that the competitive equilibrium is unique and interior. 
In Theorem~\ref{thm.interior}, we characterize equilibrium allocations, prices, and utilities, which are closely related to the vector $\se$, the \emph{effective endowment}.  An agent's consumption of a good is proportional to the Katz--Bonacich (K--B) centrality associated with that good, weighted by the effective endowment. Prices depend not only on preference weights and scarcity (total endowments), but are also proportional to the aggregate K--B centrality weighted by effective endowments. Furthermore, utilities are ordered according to the effective endowment. 
The effective endowment is the reciprocal of the multiplier of each consumer's budget constraint---that is, the shadow value of income---and is endogenously determined in equilibrium. In particular, $\se$ captures both an agent's share of the initial endowment and the pattern of externalities she provides and receives across these multiplex networks.

We then examine the welfare properties of the equilibrium. We introduce two efficiency measures: a utility-based efficiency loss, defined as the gap between the weighted utility of the competitive equilibrium and the Pareto optimum under a given weight, and a coefficient of resource utilization (CRU), following \cite{debreu1951coefficient}. Proposition~\ref{prop.efficiencyloss} shows that the utility-based loss is characterized by the Kullback--Leibler divergence between the Pareto weights and an endogenous layer-specific weight vector. The two measures are dual to each other. Moreover, the CRU admits bounds that depend only on the dispersion among these endogenous layer-specific vectors. This provides a quantitative link between network multiplexity and equilibrium efficiency.

Using the efficiency measures, we identify an exact zero-loss condition: for any pair of goods, the vectors of unweighted (transposed) Katz--Bonacich (K--B) centralities are proportional to each other. In Theorem~\ref{thm.efficiency}, we show that the interior competitive equilibrium is Pareto efficient, and thus the First Welfare Theorem holds if and only if the \emph{centrality parallel condition} holds. Observe that the centrality-parallel condition is an exact zero-loss benchmark under which competitive equilibria remain Pareto efficient despite network externalities. It includes economically relevant cases, such as regular networks,   common network structures across layers, and the classical Arrow--Debreu economy without externalities. When the condition fails, efficiency losses vary continuously with the underlying primitives.

Specifically, in a competitive equilibrium, each consumer maximizes her own utility, taking other consumers' allocations and market prices as given. Consequently, the marginal rates of substitution (MRS) between any pair of goods must equal their relative prices, and are therefore common across consumers.  
By contrast, a social planner internalizes externalities across all consumers. Hence, at a Pareto-efficient allocation, the MRS between two goods is adjusted by the unweighted (transposed) K--B centralities, which capture the extent to which each consumer's consumption of a good generates externalities for others. The centrality parallel condition implies when the aggregation of individual choices is aligned with the social planner's allocation at the efficient outcome.

In contrast, when this condition fails, we explicitly construct a reallocation of consumption bundles that strictly Pareto dominates the equilibrium allocation. The construction exploits the wedge between consumers' MRS and the social planner's MRS. Using a duality argument and an application of Farkas' lemma, we show that the existence of such a Pareto-improving redistribution is equivalent to the failure of the centrality parallel condition. Thus, the condition provides a necessary and sufficient characterization of efficiency. 


Moreover, under the centrality parallel condition, the Second Welfare Theorem also holds. In other words, any Pareto-efficient allocation can be implemented as a competitive equilibrium after a suitable redistribution of endowments. 
Taken together, these results highlight that the precise structure of externalities---in particular, who influences whom and the overall network topology---is critical for the efficiency of market equilibrium.


Our welfare results extend beyond the baseline utility specification to general utility functions that need not aggregate externalities linearly and can feature income effects, thus allowing for distributional effects. We define each entry of the externality network as the marginal utility generated by another consumer's consumption, normalized by the marginal utility of own consumption. Unlike in the baseline model, this network generally depends on the allocation. We derive a generalized centrality parallel condition under which the competitive equilibrium is efficient. We then illustrate this condition using an other-regarding preference example and identify circumstances in which income effects generate inefficiency and allow a Pareto improvement through endowment redistribution. Beyond the efficiency condition, the link between our two efficiency measures continues to hold under general utilities. With additional structure, such as additive separability across goods, the efficiency loss admits a Bregman-divergence representation, which nests the Kullback--Leibler divergence obtained in the baseline model.

Next, we examine the effects of endowments on equilibrium prices and consumer utilities in Proposition~\ref{prop.comparative}. Both the price effect and the welfare effect admit a decomposition into a redistribution component and an aggregate component. The latter vanishes when the aggregate endowment is fixed, i.e., for pure transfers of endowments across consumers. The former also vanishes in the absence of externalities, implying that pure transfers do not affect equilibrium prices. 
In our setting, however, the redistribution component induces nontrivial changes in equilibrium prices due to interactions across multiple layers (goods). The multiplex structure plays a novel role through the effective endowment channel, which captures both shares of endowments and cross-good interactions. In a model with one private good and one conspicuous network good, \cite{ghiglinoKeepingNeighborsSocial2010} show that transferring endowment from a more central agent to a less central one reduces the equilibrium price of the network good relative to the private good. We demonstrate that this result does not necessarily hold in a multiplex network environment, where cross-layer interactions must be taken into account (Proposition~ \ref{prop.priceeffect}). 
With respect to welfare effects, we establish that there exists a strictly positive weight vector such that, for any marginal transfer of endowments, the corresponding marginal change in consumers' utilities—when weighted by this vector—is equal to zero. Consequently, no local redistribution of endowments can be strictly Pareto improving.\footnote{Under a general utility specification, we show that endowment redistribution can induce Pareto improvements.}

Inefficiencies can arise due to missing markets and the absence of prices for externalities. In Proposition~\ref{prop.lindahl}, we characterize a Lindahl equilibrium in the spirit of completing these missing markets, following \cite{arrow1969organization} and \cite{foley70}. In Theorem~\ref{thm.lindahl}, we show that the Lindahl equilibrium restores efficiency through personalized prices that internalize externalities among consumers. However, despite its efficiency, the Lindahl equilibrium need not Pareto dominate the competitive equilibrium. Highly central agents, who derive substantial benefits from unpriced externalities in the competitive equilibrium, may become worse off under a Lindahl equilibrium in which these externalities are fully priced.

\paragraph*{Related literature}
This paper bridges the literature on network games and general equilibrium theory. The classic network game literature typically considers single actions under exogenous costs or budgets. Some recent papers extend this framework to multiple activities \citep{chen2018multiple, demange2025dual}. In particular, while \cite{chen2018multiple} allow for multiple activities, the cost structure remains exogenous. Consequently, the general equilibrium implications of network externalities—where prices and budgets are endogenously determined—remain underexplored. 

Two notable exceptions are \cite{ghiglinoKeepingNeighborsSocial2010} and \cite{elliottNetworkApproachPublic2019}, who study network games with market forces. \cite{ghiglinoKeepingNeighborsSocial2010} examine a two-good exchange economy in which one good is subject to social comparison under market prices. \cite{elliottNetworkApproachPublic2019} study a model with a single public good and characterize Lindahl outcomes as a particular class of efficient allocations, showing that agents' contributions must coincide with their eigenvector centralities in the benefits network.
We advance this literature by building on \cite{zenouGamesMultiplexNetworks2025}, who model games on multiplex networks.\footnote{Recent contributions to the literature on multiplex networks include  \cite{joshiNetworkFormationMultigraphs2020}, \cite{chengTheoryMultiplexitySustaining2021}, \cite{billandModelFormationMultilayer2023} and \cite{chandrasekhar2024multiplexing}.} Within this framework, we develop a general equilibrium approach in which prices are endogenously determined by markets.  This allows us to generalize the framework of \cite{ghiglinoKeepingNeighborsSocial2010} to multiple socially interacting goods structured by multiplex networks, accommodating different types of externalities—such as social comparison and public goods—and generating cross-layer price effects. We also investigate welfare efficiency and policy interventions, which are not explored in \cite{ghiglinoKeepingNeighborsSocial2010}. Compared to \cite{elliottNetworkApproachPublic2019}, we obtain a markedly different characterization of prices and consumptions in Lindahl equilibrium due to the multiplex network structure.

From the perspective of general equilibrium theory, we contribute to the extensive literature on the welfare theorems in the presence of externalities. It is well established that the standard welfare theorems typically fail in such environments. Early contributions identify conditions on information structures \citep{ledyardRelationOptimaMarket1971, osanaExternalitiesBasicTheorems1972} or restrictions on preferences—such as non-benevolence or non-malevolence \citep{WINTER196999,parksParetoIrrelevantExternalities1991}, symmetry across goods \citep{arrow2009conspicuous} or specific other-regarding preferences \citep{dufwenbergOtherRegardingPreferencesGeneral2011}—under which the welfare theorems can be restored.\footnote{Recent work has identified alternative conditions under which the Second Welfare Theorem can be restored in the presence of externalities. \cite{delmercatoSufficientConditionsSimple2023} establish the result under a ``Social Redistribution'' condition, while, with both consumption and production, \cite{anderson-duanmu-2WT} introduce a ``reciprocal externalities'' condition under which the Second Welfare Theorem holds even though the First Welfare Theorem fails. Their framework models externalities through interdependent consumption sets and production correspondences, whereas we consider a pure exchange economy in which consumption externalities are structured by multiplex networks. Thus, while these papers share our interest in the validity of the welfare theorems in the presence of externalities, they approach the question from different but complementary perspectives and derive distinct conditions for their validity.} 
Beyond asking whether an equilibrium is efficient, we quantify the degree of equilibrium (in)efficiency using two measures: a utility-based welfare loss and the coefficient of resource utilization (CRU), an approach that has not been explored in the prior literature. Our analysis yields the centrality parallel condition that characterizes when the competitive equilibrium is efficient. Complementing previous preference- or information-based approaches, this condition provides a topological characterization in terms of the network structure of externalities.\footnote{Network structure is also critical in other contexts; see, for example, \cite{ollar2023network} for full implementation under relaxed common-prior assumptions.} It restores efficiency by aligning private and social marginal rates of substitution at the competitive equilibrium. Our general-utility extension further allows for income effects and encompasses  other-regarding and Bergson-type preferences studied in \cite{parksParetoIrrelevantExternalities1991} and \cite{dufwenbergOtherRegardingPreferencesGeneral2011}.

 When the welfare theorems fail, another strand of the literature has investigated how  policies can restore efficiency or achieve Pareto improvements. For instance, \cite{anderson2025cap} propose a ``quota equilibrium'' with aggregate emission constraints and an ``emission tax equilibrium'' in production economies.\footnote{See also \cite{anderson-duanmu2026}.} We instead focus on a decentralized approach that establishes markets for externalities through personalized pricing, known as the Lindahl equilibrium. Early work such as \cite{BERGSTROM1970383} establishes the existence and efficiency of Lindahl equilibria under individually separable preferences.\footnote{\cite{bonnisseau2023existence} establish existence of Lindahl equilibria under non-separable utility and free disposal. We complement their result by proving existence without  free disposal, but with bounded externalities.} 
Exploiting the network structure, we provide a closed-form characterization of the Lindahl equilibrium. Despite restoring efficiency, however, the Lindahl equilibrium need not Pareto dominate the competitive equilibrium and may leave some consumers worse off.

\section{Motivating examples}\label{sec_Examples}

Consider a pure exchange economy with two consumers, 1 and 2, and two goods, $x$ and $y$. 
Consumer $i$ is endowed with $(w_i, v_i)$. Aggregate endowments satisfy
$w_1 + w_2 = \bar w$ and $ v_1 + v_2 = \bar v,$
and we normalize $\bar w = \bar v = 2$.

\paragraph*{Example I (asymmetric consumers)}
First, consider the case where the utility of consumers 1 and 2 consuming goods $x$ and $y$ is given by:
\begin{equation*}
    u_1 = \alpha \ln(x_1) + (1-\alpha) \ln(y_1 + \phi y_2),
    \qquad
    u_2 = \alpha \ln(x_2) + (1-\alpha) \ln(y_2).
\end{equation*}

In the benchmark case, when $\phi=0$, the competitive equilibrium is Pareto efficient, since the marginal rate of substitution between goods $x$ and $y$ is equalized across the two consumers, i.e., $MRS^{xy}_1 = MRS^{xy}_2$. This is equivalent to $\frac{y_1^*}{x_1^*} = \frac{y_2^*}{x_2^*}$, implying that the set of Pareto-efficient allocations (i.e., the contract curve) coincides with the diagonal of the Edgeworth box, as shown in Figure~\ref{fig.benchmark}.

Let us now set $\phi = 0.7$, so that the consumption of good $y$ from consumer 2 creates a positive externality on consumer 1. The \emph{private} marginal rate of substitution of consumer 1 can be easily calculated as
\begin{equation*}
    MRS^{xy}_1 = \left. \frac{\partial u_1(x_1,y_1,y_2^*)/\partial x_1}{\partial u_1(x_1,y_1,y_2^*)/\partial y_1} \right|_{(x_1,y_1)=(x^*_1,y^*_1)} = \frac{\alpha (y_1^* + \phi y_2^*)}{(1-\alpha)x_1^*},
\end{equation*}
where consumer 1 takes consumer 2’s consumption as given. 


The \emph{social} marginal rate of substitution for player 1 is
\footnote{It follows that $$
\widehat{MRS}_1^{xy}
\equiv
\left.
\frac{\partial \tilde u_1(x_1,y_1)/\partial x_1}{\partial \tilde u_1(x_1,y_1)/\partial y_1}
\right|_{(x_1,y_1)=(x_1^*,y_1^*)},
$$
where feasibility implies \(y_2=\bar v-y_1\) and thus \(\tilde u_1(x_1,y_1)\equiv u_1(x_1,y_1,\bar v-y_1)\).}
$$
\widehat{MRS}_1^{xy}
=
\frac{\alpha\bigl(y_1^*+\phi y_2^*\bigr)}{(1-\phi)(1-\alpha)x_1^*}
=
\frac{MRS^{xy}_1}{1-\phi}.
$$
Since \(0<\phi<1\), we have \(\widehat{MRS}_1^{xy}>MRS^{xy}_1\). For consumer 2, as there are no externalities received, the private and social marginal rates of substitution coincide: $   MRS^{xy}_2 =  \widehat{ MRS}^{xy}_2=\frac{\alpha y_2^*}{(1-\alpha) x_2^*}$.

As shown in Figure~\ref{fig.box}, which depicts the Edgeworth box, the social indifference curve (red dashed line) is therefore steeper than the private indifference curve (red solid line). It follows that, starting from the competitive equilibrium, a reduction in consumer 1’s consumption of $y$ together with a corresponding increase in her consumption of $x$, can raise the utilities of  both consumers. That is, such a reallocation constitutes a Pareto improvement, as illustrated by the shaded gray region in Figure~\ref{fig.box}. When $\phi=0$, the shaded gray region vanishes, as the competitive equilibrium is efficient in standard general equilibrium theory  and the set of Pareto-efficient allocations (i.e., the contract curve) coincides with the diagonal of the Edgeworth box.

In equilibrium, the allocations satisfy
$MRS^{xy}_1 = MRS^{xy}_2 = \frac{p^x}{p^y},$
that is, each consumer's marginal rate of substitution equals the common relative price ratio, where \(p^x\) and \(p^y\) denote the prices of goods \(x\) and \(y\), respectively. The set of all equilibrium allocations is
depicted by the purple curve in Figure~\ref{fig.box}. On the other hand,    the set of Pareto-efficient allocations is determined by $ \widehat{\mathit{MRS}}_1^{xy}= \widehat{\mathit{MRS}}_2^{xy}$, that is, the contract curve, which is shown by the green curve. It is immediate that no interior equilibrium is efficient due to the wedge between $ \mathit{MRS}_1^{xy}$ and $ \widehat{\mathit{MRS}}_1^{xy}$ for consumer~$1$.

\begin{figure}[ht]
    \begin{subfigure}{0.45\textwidth}
        \centering
        \includegraphics[width=1\textwidth]{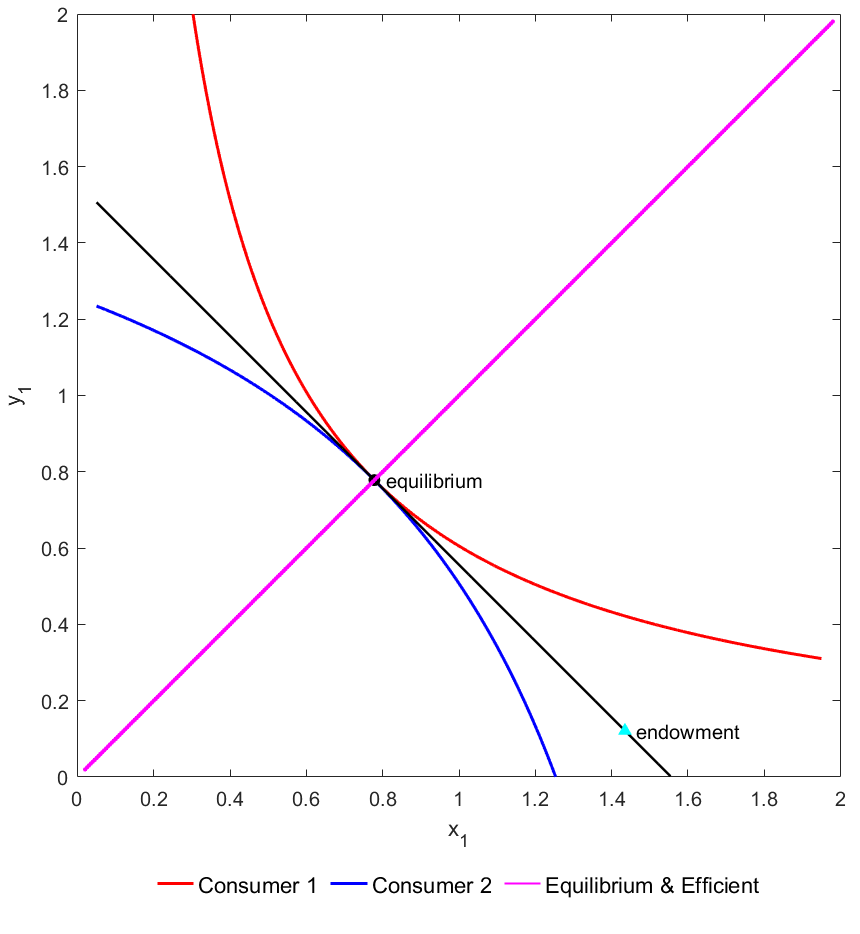}
        \caption{Benchmark: $\phi=0$}
        \label{fig.benchmark}
    \end{subfigure}
    \hfill
    \begin{subfigure}{0.45\textwidth}
        \centering
        \includegraphics[width=1\textwidth]{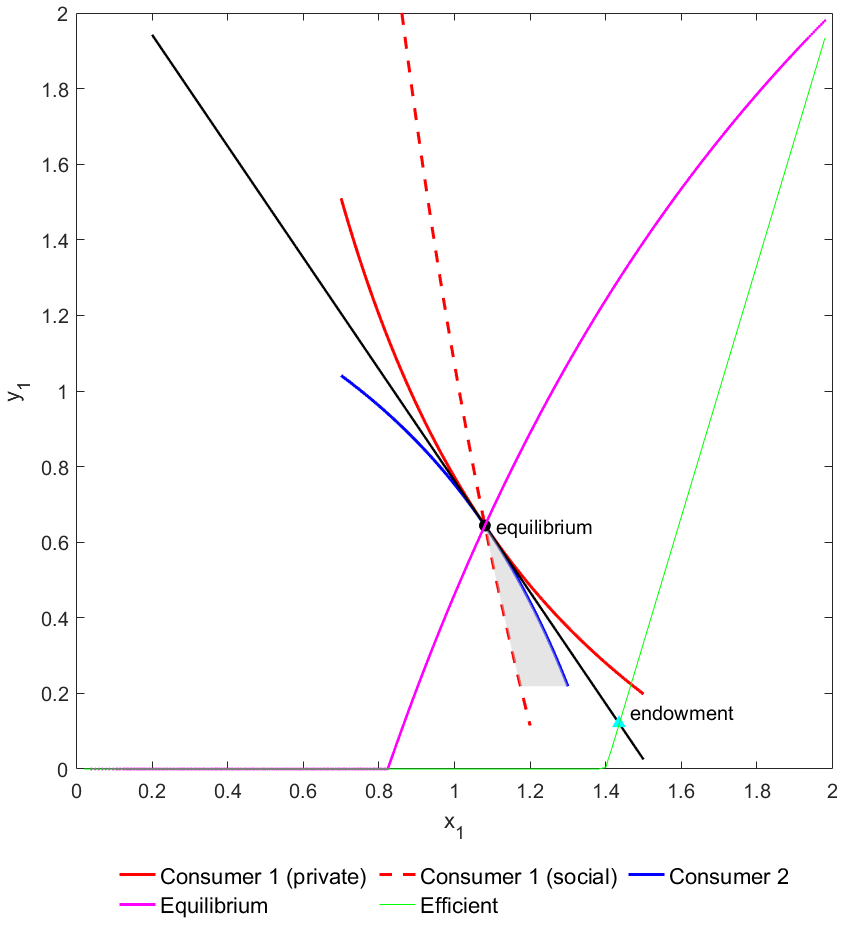}
        \caption{Inefficient equilibrium: $\phi=0.7$}
        \label{fig.box}
    \end{subfigure}
    \caption{Edgeworth boxes for the benchmark and Example I}
    \label{fig.inefficientexample}
\end{figure}



This inefficiency reflects a broken symmetry in network positions:
consumer~1 receives spillovers from consumer~2, but not vice versa.
As we will show, this asymmetry corresponds precisely to the failure of
network regularity, and our main result  (Theorem~\ref{thm.efficiency})
establishes  that such asymmetry leads to market failure.

\paragraph*{Example II (symmetric consumers)}
We now impose symmetry across consumers, so that consumer 2 also receives spillovers from consumer 1. Utilities are now  given by
\begin{equation*}
    u_1 = \alpha \ln(x_1) + (1-\alpha) \ln(y_1 + \phi y_2),
    \qquad
    u_2 = \alpha \ln(x_2) + (1-\alpha) \ln(y_2 + \phi y_1),
\end{equation*}
with $\phi = 0.7$. A similar calculation yields $\widehat{\mathit{MRS}}_i^{xy} = MRS^{xy}_i/(1-\phi)$ for $i=1,2$.
Hence, while private and social MRS do not coincide at the individual level, the equilibrium condition $MRS^{xy}_1 = MRS^{xy}_2$ is equivalent to the efficiency condition $\widehat{\mathit{MRS}}_1^{xy}= \widehat{\mathit{MRS}}_2^{xy}$ since the common factor $1/(1-\phi)$ cancels out under symmetry of consumers. It follows that interior competitive equilibrium is efficient. Figure~\ref{fig.box2} provides an illustration in the Edgeworth box. 


\begin{figure}[ht]
    \begin{subfigure}{0.42\textwidth}
        \centering
        \includegraphics[width=1\textwidth]{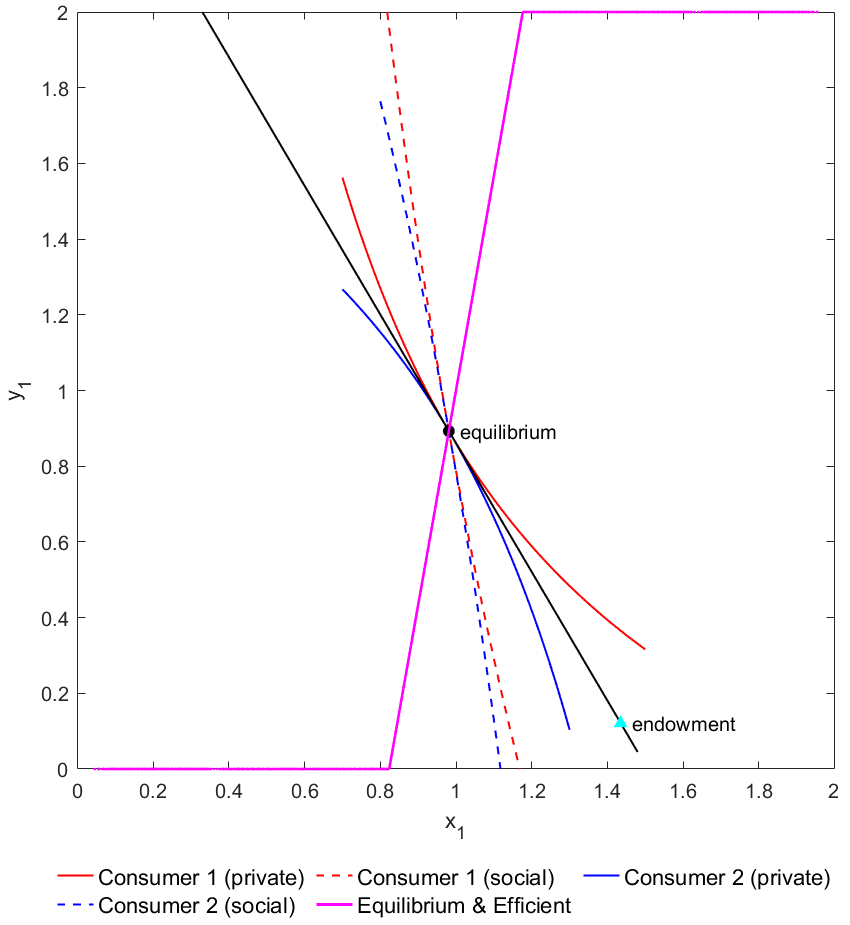}
        \caption{Symmetric consumers}
        \label{fig.box2}
    \end{subfigure}
    \hfill
    \begin{subfigure}{0.42\textwidth}
        \centering
        \includegraphics[width=1\textwidth]{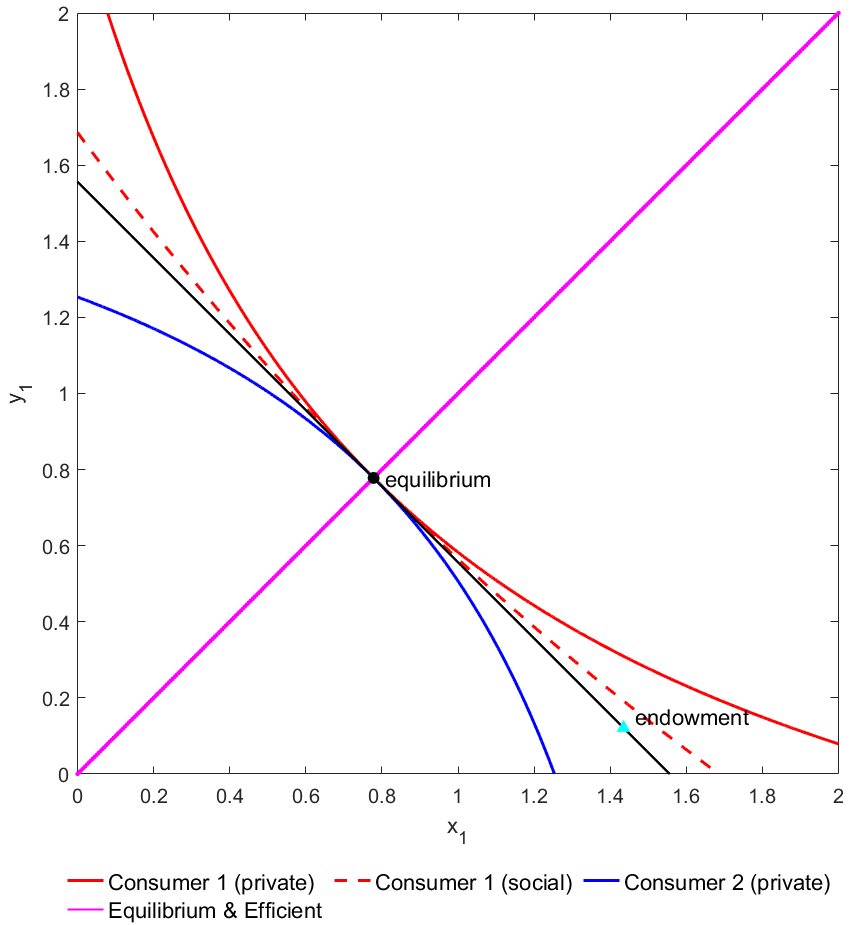}
        \caption{Symmetric goods}
        \label{fig.box3}
    \end{subfigure}
    \caption{Edgeworth boxes for Examples II and III, where the equilibria are efficient.}
    \label{fig.efficientexample}
\end{figure}

Symmetry across consumers here corresponds to the network being regular
(each consumer has the same degree). Theorem~\ref{thm.efficiency} and Lemma~\ref{lemma.symmetry} part $(i)$
generalize  this observation: competitive equilibria are efficient
when all  networks are regular.

\paragraph*{Example III (symmetric goods)}
Finally, we impose symmetry across goods by assuming that consumer 1 receives spillovers from consumer 2’s consumption of both goods with the same intensity. Utilities are given by
\begin{equation*}
    u_1 = \alpha \ln(x_1 + \phi x_2) + (1-\alpha) \ln(y_1 + \phi y_2),
    \qquad
    u_2 = \alpha \ln(x_2) + (1-\alpha) \ln(y_2).
\end{equation*}
Under this symmetry across goods, we have $\widehat{\mathit{MRS}}_i^{xy} = MRS^{xy}_i$ for $i=1,2$.
Hence, as illustrated in Figure~\ref{fig.box3}, the private and social indifference curves are tangent at the competitive equilibrium, so the wedge disappears. It follows that the competitive equilibrium is efficient. Moreover, the set of Pareto-efficient allocations coincides with the diagonal of the Edgeworth box, which is identical to the benchmark case without externalities, that is, when $\phi=0$.

Symmetry across goods here corresponds to the network being the same for the two goods. Theorem~\ref{thm.efficiency}  and Lemma~\ref{lemma.symmetry} part $(ii)$
generalize  this observation: competitive equilibria are efficient
when  all goods/layers have the same network structure. 

\medskip

In summary, these examples highlight several key insights about competitive equilibrium in networked environments. We focus here on the simplest network structure—the dyad. In the sequel, we show that our main result—symmetry restores efficiency—extends to more general network structures. In particular, considering very general networks in which goods are embedded in multilayer structures, with each good corresponding to a layer, we quantify the departure of equilibrium from Pareto efficiency using two efficiency measures, which link the magnitude of equilibrium inefficiency to the underlying multiplex network structure (Proposition~\ref{prop.efficiencyloss} in Section~\ref{sec_Welfare}). These measures imply a zero-loss condition, which shows that network regularity  or same structure between  layers    is a  sufficient condition for the efficiency of the competitive equilibrium (Theorem~\ref{thm.efficiency} in Section~\ref{sec_Welfare}). In other words, although network spillovers and externalities typically distort incentives, the First and Second Welfare Theorems continue to hold under some conditions on the networks. When these conditions fail, a natural question is what policy intervention can restore (constrained) efficiency. We address this question in Section~\ref{sec_policy}
by introducing a Lindahl equilibrium.


\section{Model and equilibrium analysis}\label{sec_model}

\subsection{The model}

We consider a pure exchange economy denoted by $\mathcal{E} := (\mathbb{R}^{|\mS|}_+, u_i, \w_i)_{i \in \mN}$, where $\mN = \{1, \ldots, n\}$ is the set of consumers, $\mS=\{1,\cdots, \bar s\}$, the set of consumption goods, and  $\w_i \in \mathbb{R}^{|\mS|}_{++}$ is the initial (strictly positive) endowment vector for consumer $i$. Each consumer~$i$ can consume a non-negative bundle of goods, thus the consumption set is $\mathbb{R}^{|\mS|}_+$. The preference over consumption bundles is represented by the utility function $u_i: \mathbb{R}^{|\mS|\times |\mN|}_+ \rightarrow \mathbb{R}$. 
We assume that each consumer~$i$ has the following (log form) Cobb-Douglas utility function:\footnote{For completeness,  we assume that the utility $u_i$ in \eqref{eq.utility} is equal to $-\infty$ when $q^s_i = x^{s}_i + \phi^s \sum_{j \in \mN} g^{s}_{ij} x^{s}_{j} \leq 0$ for some $s$.}
\begin{equation}\label{eq.utility}
	u_i(\bx_{i},\bx_{-i}) = \sum_{s\in \mS} \alpha^{s} \ln
	\left( 
	x^{s}_i + \phi^s \sum_{j \in \mN} g^{s}_{ij} x^{s}_{j}
	\right),
\end{equation}
where the parameter $\alpha^{s}>0$ represents the preference weight of good $s$\footnote{In this benchmark model with utility \eqref{eq.utility}, we assume that $\alpha^s$ is common across agents in each layer/good $s$. Remark~\ref{app:remark.alpha} in the Online Appendix shows that our welfare result is robust to allowing for agent-specific weights $\alpha_i^s$, a case that can be accommodated by our general utility extension.} and $\phi^s \in \mathbb{R}$
denotes the network spillover parameter on good $s$. For each good $s \in \mS$, cross-consumer spillovers are captured by a non-negative adjacency matrix $\bG^s$ without self-loops, i.e., $g^s_{ij} \geq 0$ and $g^s_{ii}=0$ for all $i,j \in \mN$. The network is fully general and may be directed or undirected, as well as weighted or unweighted, and may differ across goods. The parameter $\phi^s$ can be positive or negative  depending on the economic context. 
Consumption externalities are embedded in a multiplex network, so utility is defined over the product of all consumers’ consumption sets and depends on both own and others’ consumption.  Define effective consumption   as
\begin{equation}\label{eq_eff_cons}
q_i^s:=  x^{s}_i + \phi^s \sum_{j \in \mN} g^{s}_{ij} x^{s}_{j}. 
\end{equation}
It 
enters preferences via a Cobb--Douglas aggregator with $\sum_{s \in \mS} \alpha^s = 1$.


\begin{definition}\label{def.ce}
    A competitive equilibrium of the economy $\mathcal{E}$ is a tuple $(\bX^*, \bp^*)$ satisfying
    \begin{enumerate}
        \item[(i)] Utility maximization: for each consumer $i$, given the price vector $\bp^*>\boldsymbol{0} $ and the consumption of others $\bx^*_{-i} := (\bx^*_1,\ldots,\bx^*_{i-1},\bx^*_{i+1},\ldots,\bx^*_n)$, self consumption $\bx^{*}_i$ solves
        \begin{equation*}
            \begin{gathered}
                \max\limits_{\bx_i \in \mathbb{R}^{|\mS|}_+} u_{i}(\bx_i,\bx^*_{-i}) \\
        		s.t. \ \bp^* \cdot \bx_i \leq \bp^* \cdot \w_i.
            \end{gathered}
        \end{equation*}

        \item[(ii)] Market clearing: $\bX^*$ is attainable, that is 
        \begin{equation*}
            \bX^* \in \mathcal{A} := \left \{\bX \in \mathbb{R}^{|\mS|\times |\mN|}_{+}: \sum_{i \in \mN} \bx_i =  \sum_{i \in \mN} \w_i  \right \}.
        \end{equation*}
    \end{enumerate}
\end{definition}

\paragraph*{Modeling discussion}
The multiplex network framework captures the heterogeneity of social influence, enabling us to analyze how structural variations across layers\textemdash goods in our setting\textemdash interact to determine equilibrium outcomes. The framework also accommodates different types of externalities through the parameter $\phi^s$.   
Direct computation reveals that, for all distinct $i$ and $j$ in $\mN$,
\begin{equation*}
sign \left\lbrace 
\frac{\partial u_i}{\partial  x_j^s} \right\rbrace =  sign\left\lbrace  \phi^s \times g_{ij}^s\right\rbrace. 
\end{equation*} 
Specifically, $\phi^{s} > 0$ corresponds to local public goods, characterized by positive spillovers and strategic substitutes. Conversely, $\phi^{s} < 0$ captures social comparison, exhibiting negative externalities and strategic complements.\footnote{See, for instance,  \cite{ghiglinoKeepingNeighborsSocial2010, immorlicaSocialStatusNetworks2017,langtryKeepingJonesesReferenceDependent2023}.} 
The distinguishing feature of our analysis relative to \cite{zenouGamesMultiplexNetworks2025} is the treatment of income. While they assume exogenous budget constraints, our general equilibrium framework endogenizes prices and, therefore, income. This approach  not only allows us to characterize equilibrium prices and evaluate the impact of endowment transfers but also yields a novel network condition that ensures the efficiency of the equilibrium allocation (see Remark~\ref{remark.zz}).

Our  specification features non-separable externalities, which enter the utility function as a linear aggregation. This specification treats others' consumption as contributing to effective consumption at constant network-dependent rates. Accordingly, the marginal external effect is proportional to the marginal utility of own consumption: for any allocation~$\bX$,
\begin{equation*}
    \frac{\partial u_i(\bX)/\partial x^s_j}{\partial u_i(\bX)/\partial x^s_i} \equiv \phi^sg^s_{ij}.
\end{equation*}
This provides a tractable benchmark that isolates the role of network structure from allocation-dependent variation and allows us to perform comparative statics exercises and study policy implications.

However, the linear aggregation assumption rules out direct cross-good spillovers and nonlinear external effects. To address these limitations, in Section~\ref{sec_gen_ut}, we extend our benchmark Cobb--Douglas utility function \eqref{eq.utility} to a general utility specification that allows for nonlinear externalities, cross-good spillovers, and income effects, and hence for distributional effects. We show that all our welfare results continue to hold under this more general specification.

\subsection{Equilibrium analysis}\label{sec_Equilibrium}
\subsubsection{Existence}

For each good $s \in \mathcal S$, let $\bar g^s := \max_{i,j \in \mathcal N} g^s_{ij}\geq 0$ and $\bar g := \max_{s \in \mathcal S} \bar g^s \geq 0$. Let $\bar{\omega}^s = \sum_{k \in \mathcal N} \omega_k^s$ denote the aggregate endowment of good $s$, and let $\eta_i^s := \omega_i^s/\bar{\omega}^s>0$ be consumer $i$’s endowment share. Write $\etab^s=(\eta_1^s,\cdots,\eta_n^s)^{\intercal}$ and $\underline{\eta} := \min_{i \in \mathcal N,\; s \in \mathcal S} \eta_i^s > 0$.\footnote{We use the superscript $\intercal$ to denote the transpose of a vector or matrix.}  
We impose the following assumptions.

\begingroup
    \addtocounter{assumption}{1}
    
    \edef\myprefix{\arabic{assumption}}
    
    \setcounter{assumption}{0}
    
    \renewcommand{\theassumption}{\myprefix\Alph{assumption}}

    \begin{assumption}\label{assump.nonempty}
        $\phi^s > -\frac{\min_{i \in \mN} \eta_i^s}{(n-1)\bar{g}^s}$ for any $s \in \mS$.
    \end{assumption}

    \begin{assumption}\label{assump.interior}
        $|\phi^s| < \frac{\ubar{\eta}}{(n+1) \bar{g}}$ for any $s \in \mS$.
    \end{assumption}

    \setcounter{assumption}{\myprefix}
\endgroup




Assumption~\ref{assump.nonempty} ensures that negative externalities are not too  strong to reduce effective consumption below zero, which would otherwise jeopardize the existence of an equilibrium. In particular, Assumption~\ref{assump.nonempty} is automatically satisfied when $\phi^s\geq 0$. 
Assumption~\ref{assump.interior} strengthens Assumption~\ref{assump.nonempty} by requiring that all externalities, in absolute value, be sufficiently small.

To show the existence of the competitive equilibrium, we follow  \cite{debreuSocialEquilibriumExistence1952} to define an abstract economy  $\hat{\mathcal{E}}$ as the tuple 
\begin{equation*}
    \hat{\mathcal{E}} := \left\{ (\mathcal{X}_i,u_i,A_i)_{i \in \mN}, \left(\mathcal{P},\bp^\intercal\left[\sum_{i \in \mN}\bx_i - \sum_{i \in N} \w_i\right],A_p \right) \right\}.
\end{equation*}
We explain each component of the abstract economy below.

First, define the choice set for each consumer:
\begin{equation*}
    \mathcal{X}_i := \left\{ \bx_i \in \mathbb{R}^{|\mS|}_+: \bx_i \leq C\sum_{i \in \mN} \w_i \right\},
\end{equation*}
where  $ C >1$ is a constant. Also, define the set of  normalized prices as $\mathcal{P} := \{ \bp \in \mathbb{R}^{|\mS|}_+: \mathbf{1}^\intercal \bp = 1 \}$.

Denote by $\mathcal{X}_{-i}$ the product of choice sets $\mathcal{X}_{1} \times \ldots \times \mathcal{X}_{i-1} \times \mathcal{X}_{i+1} \times \ldots \times \mathcal{X}_{n}$. Then for each individual $i$, define a correspondence $A_i: \mathcal{X}_{-i} \times \mathcal{P} \rightrightarrows \mathcal{X}_i$ such that
\begin{equation*}
    A_i(\bx_{-i},\bp) := \left\{ \bx_i \in \mathcal{X}_i: \bp^\intercal \bx_i \leq \bp^\intercal \w_i, x^{s}_i + \phi^s \sum_{j \in \mN} g^{s}_{ij} x^{s}_{j} \geq 0, \forall s \in \mS  \right\},
\end{equation*}
which can be interpreted as the constraints for consumer $i$. Moreover, the correspondence with respect to consumptions $A_p: \mathcal{X} \rightrightarrows \mathcal{P}$ is defined as $A_p (\bX) \equiv \mathcal{P}$.

In the proof, we show that, for a suitable choice of $C$, for any $(\bx_{-i},\bp) \in \Pi_{k \neq i} \mathcal{X}_k \times \mathcal{P}$, $A_i(\bx_{-i},\bp)$ always contains $\w_i$  and, hence, is non-empty. Furthermore, it can be shown that $A_i$ is continuous and has a compact and convex graph. Finally, since $\mathcal{X}_i$ and $\mathcal{P}$ are compact and $u_i$ is continuous and quasi-concave in $\bx_i$, it follows from \cite{debreuSocialEquilibriumExistence1952} that an equilibrium exists for the constructed abstract economy $\hat{\mathcal{E}}$.

\begin{lemma}[Existence for the abstract economy]\label{lemma.exist}
    If Assumption \ref{assump.nonempty} holds, then there exists an equilibrium $(\bX^*, \bp^*)$ of the abstract economy $\hat{\mathcal{E}}$ such that 
    \begin{align*}
        &\bx^*_i \in \arg\max_{\bx_i \in A_i(\bx^*_{-i},\bp^*)} u_i(\bx_i,\bx^*_{-i}),\\
        &\bp^* \in \arg\max_{\bp \in \mathcal{P}} \bp^\intercal \left[ \sum_{i \in \mN} \bx^*_i - \sum_{i \in N} \w_i \right]. 
    \end{align*}
\end{lemma}


Based on this existence result, we can next show the existence of the competitive equilibrium of the original pure exchange economy $\mathcal{E}$.

\begin{proposition}[Existence and Uniqueness]\label{prop.exist}\quad
\begin{enumerate}
    \item[(i)]
    If Assumption \ref{assump.nonempty} holds, 
    then a competitive equilibrium of economy $\mathcal{E}$ exists.
    
    \item[(ii)]
  If Assumption \ref{assump.interior} holds, 
then the equilibrium is interior, the allocation $\bX^*$ is unique, 
and the price vector $\bp^*$ is unique up to a positive scalar multiple.
\end{enumerate}
\end{proposition}

To show the existence, it suffices to verify that the equilibrium of the abstract economy~$\hat{\mathcal{E}}$ is a competitive equilibrium of $\mathcal{E}$. There are two gaps between these two equilibria. First, the equilibrium allocation of $\hat{\mathcal{E}}$ may not be attainable in $\mathcal{E}$. Moreover, although consumers face the same budget constraint, the feasible consumption set in $\hat{\mathcal{E}}$, $\mathcal{X}_i$, is smaller than in $\mathcal{E}$, namely $\mathbb{R}^{|\mS|}_+$.  In the proof of Proposition~\ref{prop.exist}, we close these gaps to show existence, and establish the interiority and uniqueness of the equilibrium under Assumption~\ref{assump.interior}.


\subsubsection{Characterization of  interior equilibria}

In this section, we give a further characterization of the interior equilibrium. Under Assumption \ref{assump.interior}, define the matrix:\footnote{The spectral radius of $\phi^s\bG^s$ is smaller than $\|\phi^sG^s\|_\infty \leq (n-1)|\phi^s|\bar g^s<1$, where the last inequality follows from Assumption~\ref{assump.interior}. Therefore, $\bM^s$ is well-defined.} 
\begin{equation*}
	\bM^{s} := [\mathbf{I}_n + \phi^{s} \bG^{s}]^{-1} 
	= \sum_{k=0}^{\infty} (-\phi^s \bG^s)^k,
\end{equation*}
which captures consumers’ influence in layer/good~$s$. The Katz-Bonacich (K-B) centrality  \citep{ballesterWhosWhoNetworks2006}\footnote{See also \cite{sadler2022ordinal}, who relates K--B-type centrality measures to equilibria in network games.}  for good/layer $s$ with weight vector~$\mathbf{z}$ is defined as $\mathbf{b}^s(\mathbf{z}) := \bM^{s} \mathbf{z}$    and  the aggregate   K-B centrality as $b^s(\mathbf{z}) := \mathbf{1}^\intercal \mathbf{b}^s(\mathbf{z}) = \mathbf{z}^\intercal \widetilde{\mathbf{b}}^s$, where $\widetilde{\mathbf{b}}^s := (\bM^s)^\intercal \mathbf{1}$ denotes the (transposed)  unweighted K-B centrality for good $s$.\footnote{If the network is undirected, i.e, $g^s_{ij}=g^s_{ji}$, then $\widetilde{\mathbf{b}}^s=\mathbf{M}^s{\mathbf{1}}=\mathbf{b}^s(\mathbf{1})$, which is the unweighted K-B centrality for good $s$ \citep{ballesterWhosWhoNetworks2006}.}

\begin{theorem}[Characterization of interior equilibrium]\label{thm.interior}
    If Assumption \ref{assump.interior} holds, the interior competitive equilibrium $(\bX^*, \bp^*)$  has the following characterization. For each good $s\in\mS$, the consumption profile  is  given by
        \begin{equation}\label{eq.equix}
            \bx^{s*}  
            = \frac{\bar{\omega}^{s}}{ b^s(\se)}  \mathbf{b}^s (\se),
        \end{equation}
    the equilibrium price is equal to
        \begin{equation}\label{eq.equip}
            p^{s*} 
            =\frac{\alpha^{s}}{\bar{\omega}^{s}} b^s(\se),
        \end{equation}
                and, for each consumer $i \in \mN$, the utility  is given by
        \begin{equation}\label{eq.equiu}
            u_i^* =  \sum_{s \in \mS} \alpha^s \ln\left( \frac{\bar{\omega}^{s}\mu_i^*}{b^s(\se)}\right),
        \end{equation}
where $\se> \mathbf{0}$ solves the following homogeneous linear system 
\begin{equation}\label{eq.mudetermine}
    \underbrace{ \left[\sum_{s \in \mS}
    \alpha^s \left(\mathbf{I}_n - \boldsymbol{\eta}^{s}\mathbf{1}^\intercal  \right)\bM^{s}   \right]}_{\overline{\bM}}  \bmu = \mathbf{0}.
\end{equation} 
\end{theorem}

Given that the equilibrium allocation is interior, the first-order condition of utility maximization problem for  consumer~$i$ is given by
\begin{equation*}
    \frac{\alpha^s}{q^s_i} - \lambda_i p^s = 0,
\end{equation*}
where $q_i^s$ is defined in \eqref{eq_eff_cons} and $\lambda_i>0$ is the Lagrange multiplier of $i$'s budget constraint; hence, it  represents the shadow price. Defining $\mu_i \equiv 1/\lambda_i$, we can rewrite the first-order condition as
\begin{equation}\label{eq.foc}
    \frac{\alpha^s}{p^s}\mu_i = q^s_i = x^s_i+\phi^s\sum_{j \in \mN} g_{ij}^s x^s_j.
\end{equation}
In matrix form, we obtain
\[
 \frac{\alpha^s}{p^s}\bmu= \bx^s+\phi^s\mathbf{G}^s\bx^s \implies  \bx^s = \frac{\alpha^s}{p^s} \bM^s \bmu = \frac{\alpha^s}{p^s}  \mathbf{b}^s (\bmu).
\] 
By further imposing the market clearing condition $\mathbf{1}^\intercal \bx^s = \bar{\omega}^s$, we can express prices as in equation~\eqref{eq.equip}. Substituting this price expression into the equation above yields equation~\eqref{eq.equix}. Hence, both consumption quantities and prices can be expressed as functions of $\bmu$. Moreover, the effective consumption for each $s\in\mS$ is given by
\begin{equation}\label{eq.q}
    \mathbf{q}^{s} = (\bM^s)^{-1} \bx^{s} =\frac{\alpha^s}{p^s} \bmu= \frac{\bar{\omega}^{s}}{ b^s(\bmu)} \bmu
\end{equation}
Under Cobb–Douglas aggregation, this directly yields the equilibrium utility of each consumer  $i$, as given in \eqref{eq.equiu}.

What remains is to determine $\bmu$.
The binding budget constraints for all consumers imply 
\begin{equation*} 
	\sum_{s \in \mS} \alpha^{s} 
    \etab^s \mathbf{1}^\intercal \bM^{s} \bmu
	= \sum_{s \in \mS} p^{s} \w^{s}
	= \sum_{s \in \mS} p^{s} \bx^{s} 
	= 
	\left(
	\sum_{s \in \mS}
	\alpha^s \bM^{s}
	\right)
	\bmu
\end{equation*}
where we substitute $\bx^s$  and $p^s$ from equations \eqref{eq.equix} and \eqref{eq.equip}, respectively, and use $\etab^s=\w^s/\bar\omega^s$. 
This system determining $\bmu$  is equivalent to \eqref{eq.mudetermine} and admits at least one nontrivial solution because $\overline{\bM}$ is singular. Indeed, $\mathbf{1}^\intercal \overline{\bM} = \mathbf{0}^\intercal$ since $\mathbf{1}^\intercal \boldsymbol{\eta}^s = 1$ for all~$s$. In the proof, we further show that  $\operatorname{rank}(\overline{\bM})=n-1$ under Assumption~\ref{assump.interior}. Hence the solution for $\bmu$ is unique up to scale, which implies uniqueness of the interior equilibrium.\footnote{Any scalar multiple of a solution $\bmu$ to (\ref{eq.mudetermine}) induces the same proportional scaling of the equilibrium prices $\bp$, while leaving the relative prices and consumption allocation $\bx^{s}$ unchanged.}

We now turn to the properties of the equilibrium. First, all equilibrium objects can be expressed as functions of $\se$, which we refer to as the effective endowment. In particular, equation~\eqref{eq.mudetermine} shows that $\se$ is determined by individual endowment shares adjusted by the multiplex network. As we show in Remark~\ref{remark.benchmark}, without externalities, $\se$ is simply proportional to a preference-weighted average of endowment shares.

Equilibrium consumption for good $s$, given in  \eqref{eq.equix}, is proportional to Katz–Bonacich (K-B) centrality. The ratio of equilibrium consumption between consumers $i$ and $j$ for good $s$  satisfies:
\begin{equation*}
    \frac{x^{s*}_i}{x^{s*}_j} = \frac{b^{s}_i(\se)}{b^{s}_j(\se)}.
\end{equation*}
Equivalently, by equation~\eqref{eq.q}
\begin{equation*}
    \frac{q^{s*}_i}{q^{s*}_j} = \frac{\mu^*_i}{\mu^*_j}.
\end{equation*}
Since this holds for all $s \in \mS$, effective consumption for every good is proportional to the effective endowment $\bmu^*$. By equation~\eqref{eq.equiu}, $u_i^*-u_j^* = \ln \mu^*_i-\ln \mu^*_j$. An immediate implication is that if $\mu^*_i > \mu^*_j$, then $u^*_i > u^*_j$.

With externalities, equilibrium prices,  given by \eqref{eq.equip}, depend not only on preference weights and scarcity (total endowments) but also increase with $b^s(\se)$,  the aggregate K-B centrality weighted by effective endowments.   Consequently, ceteris paribus, \emph{a ``more central'' layer/good characterized by a larger aggregate K-B centrality   has a higher price}. For instance, consider two goods $s$ and $t$ with identical preference weights, total endowments, and spillover parameters $\phi^s=\phi^t<0$ (indicating social comparison). If good $s$ possesses a denser network structure (i.e., $\bG^s \ge \bG^t$ elementwise), it will exhibit a higher equilibrium price. 

However, the effect of network centrality on prices depends on the sign of the spillover parameter. To see this, consider two goods $s$ and $t$ that are identical in all primitives except for the sign of spillovers, with $-\phi^s = \phi^t > 0$. Thus, good $s$ exhibits social comparison, whereas good $t$ is a local public good. Social comparison induces excessive private consumption, while local public-good induces free riding and hence insufficient consumption. Equilibrium prices adjust in opposite directions: the price of good $s$ is higher than the price of good $t$, as prices adjust to clear markets. This observation shows that equilibrium prices absorb network externalities and respond to the distortions they generate. Whether such price adjustments are sufficient to restore efficiency is the question we address in Section~\ref{sec_Welfare}.

\begin{remark}[Multiplexing with exogenous  prices]\label{remark.zz}
In the multiplex network setting of \cite{zenouGamesMultiplexNetworks2025}, agents maximize utility subject to \emph{exogenous} budget constraints $\mathbf{1}^\intercal \bx \leq T_i$, implying fixed prices $p^s = 1$ for all $s$. Consequently, $\bmu^{exo}$ is explicitly given by
\begin{equation*}
    \bmu^{exo} = \left( \sum_{s \in \mS} \alpha^s \bM^{s} \right)^{-1} \mathbf{T},
\end{equation*}
where $\mathbf{T} = (T_1, \cdots, T_n)^\intercal$ denotes the exogenous income vector. 

By contrast, in our general equilibrium setting, prices are \emph{endogenously} determined, implying that income is also endogenous and given by $T_i = \bp^\intercal \w_i$. As a result, $\se$ is characterized by equation~\eqref{eq.mudetermine}, which links income to the detailed distribution of endowments across goods, as captured by $\etab^s$. Moreover, endogenous prices incorporate the effects of network externalities and adjust in response to the distortions they create. This endogenous price adjustment leads to welfare properties that differ from those in the exogenous-price environment. We discuss this point further in Online Appendix~\ref{appsec.regular}.
\end{remark}

\begin{remark}[Model without externalities]\label{remark.benchmark}
Suppose that $\phi^s = 0$ for all $s \in \mS$, so that $\bM^s = \mathbf{I}_n$. In this case, the economy reduces to a standard competitive equilibrium with Cobb--Douglas preferences and no spillovers. Consumption then coincides with effective consumption, since there are no externalities, and is proportional to the effective endowment vector $\se$, determined by  \eqref{eq.mudetermine}. We have
\begin{equation*}
    \frac{\se}{\mathbf{1}^\intercal \se}
    =
    \sum_{t \in \mS} \alpha^t \boldsymbol{\eta}^t = \frac{\bx^{s*}}{\bar{\omega}^s},
    \quad \text{for all} \ s \in \mS.
\end{equation*}
That is, normalized effective endowments and equilibrium consumptions coincide with a weighted average of individual endowment shares, where the weights are given by the preference parameters $\{\alpha^t\}_{t \in \mS}$.
Moreover, equilibrium relative prices are given by
\begin{equation*}
    \frac{p^{s*}}{p^{t*}}
    =
    \frac{\alpha^s}{\alpha^t}
    \frac{\bar{\omega}^t}{\bar{\omega}^s},
\end{equation*}
which depend only on preference weights and aggregate endowments. In particular, prices are invariant to redistributions of endowments across consumers that leave aggregate endowments unchanged. As shown in the comparative statics analysis in Section~\ref{sec_CS}, this neutrality property breaks down in the presence of spillovers.
\end{remark}

\section{Welfare properties of equilibrium}\label{sec_Welfare}

Having characterized the equilibrium, we next study its welfare properties. In the presence of externalities, are competitive equilibria efficient? We first introduce two efficiency measures that quantify the welfare loss associated with the equilibrium allocation. We then characterize the conditions under which the competitive equilibrium is efficient and identify the role of network structure in shaping the alignment between private and social incentives.

\subsection{Efficiency measures}

An attainable allocation $\hat{\bX} \in \mathcal{A}$ is (Pareto) efficient if there is no attainable allocation $\bX \in \mathcal{A}$ such that $u_i(\bX) \geq u_i (\hat{\bX})$ for all $i \in \mathcal{N}$, with at least one inequality strict. By the convexity of $\mathcal{A}$ and the concavity of $u_i$, this definition can be represented through a planner's weighted utility maximization problem.

\begin{lemma}[Pareto frontier]\label{lemma.efficienttoplanner}
    An interior attainable allocation $\hat{\bX} \in \mathcal{A}$ is (Pareto) efficient if and only if it solves
    \begin{equation}\label{eq.pareto}
        \max\limits_{\bX \in \mathcal{A}} \sum_{i\in \mN} \theta_i u_i(\bX)=\btheta^\intercal \mathbf{u}(\bX) 
    \end{equation}
    for some positive weight vector $\btheta \in \mathbb{R}^n_{++}$. Moreover, the interior solution to \eqref{eq.pareto} takes the form:
    \begin{equation}\label{eq.efficientx}
        \hat{\bx}^s(\btheta) = \bar{\omega}^s \bM^s \left\{ (\widetilde{\mathbf{b}}^s)^{-1}
        \odot \frac{\btheta}{\mathbf{1}^\intercal \btheta} \right\},
        \quad s \in \mS,
    \end{equation}
    where $(\widetilde{\mathbf{b}}^s)^{-1} =(1/\tilde{b}^s_1,\ldots,1/\tilde{b}^s_n)^\intercal$ and $\odot$ denotes the Hadamard product.
\end{lemma}

Lemma~\ref{lemma.efficienttoplanner} implies that every efficient allocation can be characterized as the solution to a planner's problem for some Pareto weight vector $\btheta$. Motivated by this characterization, we construct an efficiency measure based on the welfare gap between the equilibrium allocation and the corresponding planner optimum. Specifically, fix the equilibrium allocation $\bX^*$ in Theorem~\ref{thm.interior} and a Pareto weight vector $\btheta$, and define the efficiency loss $ L(\btheta)$ at $\btheta$ as
\begin{equation*}
    L(\btheta):= \max_{\bX \in \mathcal{A}} \btheta^\intercal \mathbf{u}(\bX) - \btheta^\intercal \mathbf{u}(\bX^*).
\end{equation*}
For any $\btheta$, $L(\btheta)$ is nonnegative, and  it equals zero only if $\bX^*$ solves \eqref{eq.pareto} for weight $\btheta$.

The second measure is resource-based and adopts the coefficient of resource utilization (CRU) introduced by \cite{debreu1951coefficient}. Formally, define
\begin{equation}\label{eq.efficiencycoef}
    \begin{gathered}
        \mathcal{CRU} := \min_{\bX \geq \mathbf{0},\gamma \in [0,1]} \gamma\\
        s.t. \ \sum_{i \in \mN} x^s_i \leq \gamma \bar{\omega}^s, \ \forall s \in \mS,\\
        u_i(\bX) \geq u_i(\bX^*), \forall i \in \mN.
    \end{gathered}
\end{equation}
While $L(\btheta)$ evaluates the utility gap, the coefficient $\mathcal{CRU}$ asks the dual question, that is, what fraction of aggregate resources is sufficient to reproduce the equilibrium utility profile.\footnote{The measure $L$ depends on the cardinal representation of utility and
is therefore not invariant to arbitrary increasing transformations of individual utility functions, whereas $\mathcal{CRU}$ is invariant.} A value closer to one indicates that less proportional resource saving is possible while maintaining these utility levels. 

We now characterize these two measures. Define the simplex $\Delta := \{\mathbf{z} \in \mathbb{R}^n_{+} : \mathbf{1}^\intercal \mathbf{z} = 1\}$, and denote its interior by $\Delta^o$. For each good $s \in \mS$, define the layer-specific vector
\begin{equation}\label{eq.rho}
    \brho^s := \frac{\se \odot \widetilde{\mathbf{b}}^s}{(\se)^\intercal \widetilde{\mathbf{b}}^s} \in \Delta^o,
\end{equation}
where $\se$ is defined in Theorem~\ref{thm.interior}. 
For any $\mathbf{y}, \mathbf{z} \in \Delta^o$, the Kullback-Leibler (KL) divergence $D_{KL} (\mathbf{z} \,\|\, \mathbf{y}) := \sum_i z_i \ln \frac{z_i}{y_i}$, is nonnegative and equals zero if and only if $\mathbf{y} = \mathbf{z}$.

\begin{proposition}[Characterization of the efficiency measures]\label{prop.efficiencyloss}
Suppose externalities are sufficiently small, given an interior equilibrium:
\begin{enumerate}
    \item[(i)] Let $\btheta \in \Delta^o$ be a Pareto weight vector that supports an interior efficient allocation. Then,
    \begin{equation*}
        L(\btheta) 
        = 
        \sum_{s \in \mS} \alpha^s  
        D_{KL} \left( \btheta \,\|\, \boldsymbol{\rho}^s \right).
    \end{equation*}
    
    \item[(ii)] The CRU satisfies
    \begin{equation}\label{eq.minlosscru}
        \ln(\mathcal{CRU})
        = - \min_{\btheta \in \Delta}\sum_{s \in \mS} \alpha^s  
        D_{KL} \left( \btheta \,\|\, \brho^s \right).
    \end{equation}
    Moreover, the CRU has the following bounds:
    \begin{equation} \label{eq-L-bounds}
        - \max_{s,t \in \mS} D_{KL} \left( \brho^s \,\|\, \brho^t \right) 
        \leq
        \ln(\mathcal{CRU})
        \leq
        -\frac{1}{2} \sum_{s, t \in \mS} 
        \alpha^s \alpha^t H(\brho^s, \brho^t )^2 .
    \end{equation}
    where $H$ denotes the Hellinger distance.\footnote{$H(\mathbf{y}, \mathbf{z} ):=  \sqrt{\sum_{i \in \mN}\left(\sqrt{y_i}-\sqrt{z_i} \right)^2}$ for $\mathbf{y},\mathbf{z} \in \Delta$ (See, for example, \citealp{tsybakov2008introduction}).}
\end{enumerate}    
\end{proposition}

Part~(i) indicates that the efficiency loss $L(\btheta)$ measures the aggregate discrepancy between the Pareto weight $\btheta$ and the layer-specific weights $\{\brho^s\}_{s\in\mS}$. In fact, $\brho^s$ has a natural interpretation as the Pareto weight vector associated with layer $s$. Consider a planner's problem associated with weight $\brho^s$. Applying Lemma~\ref{lemma.efficienttoplanner}, the optimal allocation of good $s$ is
\begin{equation*}
    \bar{\omega}^s \bM^s \left\{(\widetilde{\mathbf{b}}^s)^{-1} \odot \brho^s \right\} = \frac{\bar{\omega}^s}{(\se)^\intercal \widetilde{\mathbf{b}}^s} \bM^s \se,
\end{equation*}
which coincides with the equilibrium allocation of good $s$. This may not necessarily hold for other goods. However, the planner must use a common Pareto weight across all layers. Inefficiency therefore arises when different layers imply different relative social weights on the same set of consumers. Observe that $L$ uses the Kullback--Leibler (KL) divergence to measure the distance between the Pareto weight $\btheta$ and the
layer-specific weight $\boldsymbol{\rho}^s$. The KL-divergence form arises
from the log form of Cobb--Douglas specification; under more general utility functions,
such as CRRA utilities, the measure takes a different but qualitatively similar
 form (i.e., Bregman-divergence; see Section \ref{sec_gen_ut} and Online Appendix~\ref{appsec.general}).

Part~(ii) establishes a novel link between the welfare-based and the resource-based measure. While $L(\btheta)$ evaluates the equilibrium relative to the planner optimum associated with a given Pareto weight vector, the CRU evaluates the equilibrium under the Pareto weights that make it the closest to efficiency. In this sense, $-\ln\mathcal{CRU}$ is the smallest welfare loss that can be attained by varying the social weights. This connection is particularly useful because the two measures have different invariance properties. The value of $L(\btheta)$ depends on the cardinal representation of utility and is generally not invariant to monotone transformations, whereas $\mathcal{CRU}$ depends only on the underlying preferences and is therefore invariant. Part~(ii) shows that, the ordinal resource-based measure can nevertheless be characterized exactly through the minimum of the cardinal welfare loss.

Let $\btheta^{\min}$ denote the minimizing Pareto weight in~\eqref{eq.minlosscru}. The explicit solution derived in the proof  yields the following decomposition of $L$ for Pareto weights satisfying the interiority condition in Proposition~\ref{prop.efficiencyloss}:
\begin{equation}\label{eq.ltheta}
    L(\btheta)
    =
   \underbrace{ D_{KL}(\btheta \,\Vert\, \btheta^{\min})}_{\geq 0}
  + \underbrace{(-\ln \mathcal{CRU})}_{\geq 0}.
\end{equation}
Hence, $L(\btheta^{\min})=-\ln\mathcal{CRU}$, and the additional welfare loss associated with any other such Pareto weight is exactly its KL divergence from $\btheta^{\min}$. This decomposition separates the minimum inefficiency inherent in the equilibrium from the additional loss generated by choosing social weights away from the best-fitting weights.

The bounds in Proposition~\ref{prop.efficiencyloss} provide a complementary pairwise interpretation. They relate the $\mathcal{CRU}$ to the dispersion among the layer-specific Pareto weights without requiring the explicit computation of $\btheta^{\min}$. In particular, a nonzero cross-layer dispersion makes the $\mathcal{CRU}$ strictly below one.

The characterization also identifies the condition under which $\mathcal{CRU}=1$. This means that the equilibrium allocation does not admit any resource-saving representation while preserving all equilibrium utility levels. Proposition~\ref{prop.efficiencyloss} further implies that $\mathcal{CRU}=1$ if and only if $\brho^s=\brho^t$ for all $s,t\in\mS$. In the next subsection, we  translate this condition into the primitives of the network structure and show that it is precisely the condition under which the competitive equilibrium is Pareto efficient.

\subsection{Multiplex network conditions for efficiency}

The preceding section expresses equilibrium inefficiency in terms of statistical distances among the layer-specific Pareto weights $\{\brho^s\}_{s\in\mS}$. We now translate this condition into restrictions on the underlying network structure. Since the effective-endowment vector $\se$ in \eqref{eq.rho} is common across layers, we have
\begin{equation*}
    \brho^s=\brho^t
    \iff
    \widetilde{\mathbf{b}}^s
    \parallelsum
    \widetilde{\mathbf{b}}^t,
\end{equation*}
where two vectors $\bx$ and $\mathbf{y}$ are said to be parallel, denoted by $\bx\parallelsum\mathbf{y}$, if one is a nonnegative scalar multiple of the other. We define the \emph{centrality parallel condition} as
\begin{equation}\label{eq.parallel}
    \widetilde{\mathbf{b}}^s
    \parallelsum
    \widetilde{\mathbf{b}}^t,
    \qquad
    \forall s,t\in\mS.
\end{equation}
This condition depends only on the externality structure across network layers and not on the distribution of initial endowments.

\begin{theorem}[Welfare theorems]\label{thm.efficiency}
    Suppose the equilibrium allocation $\bX^*$ is interior.
    \begin{itemize}
        \item[(i)] If the centrality parallel condition~\eqref{eq.parallel} holds:
        
        \begin{enumerate}
            \renewcommand{\labelenumi}{(\roman{enumi})}
            \item[(i1)]  (First Welfare Theorem) The equilibrium is efficient. Moreover, the Pareto weight in \eqref{eq.efficientx} is $\btheta^* = \mathbf{b} \odot \se $, where $\se$ is the effective endowment, and $\mathbf{b}$ is a vector satisfying $\mathbf{b} \parallelsum \widetilde{\mathbf{b}}^s, \forall s \in \mS$.
    
            \item[(i2)] (Second Welfare Theorem) Every interior efficient allocation can be supported by a competitive equilibrium for an appropriate choice of initial endowments.
             \end{enumerate}

        \item[(ii)]     If condition \eqref{eq.parallel} does not hold, then both welfare theorems fail.
    \end{itemize}
   
\end{theorem}

The intuition is obtained by comparing the incentives faced by individual consumers with those of a social planner. Define the marginal rate of substitution (MRS) between goods $s$ and $t$ for consumer $i$ by
\begin{equation*}
    \mathit{MRS}^{st}_i(\bX) := \frac{\partial u_i (\bX)/ \partial x^s_i}{\partial u_i (\bX)/ \partial x^t_i}.
\end{equation*}
At a competitive equilibrium $\bX^*$, each consumer’s MRS is equal to the relative price:
\begin{equation*}
    \mathit{MRS}^{st}_i(\bX^*) = \frac{p^s}{p^t} = \mathit{MRS}^{st}_j(\bX^*).
\end{equation*}

From the planner’s perspective, an increase in $x_i^s$ affects not only consumer $i$’s own utility, but also the utilities of other consumers through network spillovers. In particular,
\begin{equation*}
    \theta_i \frac{\partial u_i}{\partial x^s_i} +
    \sum_{j\neq i} \theta_j\frac{\partial u_j}{\partial x^s_i}
    =
    \theta_i \frac{\partial u_i}{\partial x^s_i}   +
    \phi^s\sum_{j\neq i}g^s_{ji} \theta_j\frac{\partial u_j}{\partial x^s_j}
    = \beta^s, \quad \forall i \in \mN
\end{equation*}
where $\beta^s$ is the Lagrangian multiplier of the resource constraint, i.e., the shadow value associated with the attainability constraint for good $s$. Rewrite this in matrix form:
\begin{equation}\label{eq.plannerfoc}
    \beta^s \mathbf{1} = [\mathbf{I}_n + \phi^s (\bG^s)^\intercal] 
    \left(\theta_1\frac{\partial u_1}{\partial x^s_1},\ldots, \theta_n\frac{\partial u_n}{\partial x^s_n}  \right)^\intercal
    \implies
    \theta_i\frac{\partial u_i}{\partial x^s_i} = \beta^s \widetilde{b}^s_i,
\end{equation}
where $\widetilde{b}^s_i$ is the (transposed) unweighted K-B centrality. Then,  at an efficient allocation~$\hat{\bX}$,
\begin{equation*}
    \mathit{MRS}^{st}_i(\hat{\bX}) = \frac{\beta^s\widetilde{b}^s_i}{\beta^t\widetilde{b}^t_i}.
\end{equation*}

For the First Welfare Theorem to hold, the planner's MRS must also be equalized across consumers, which requires $\widetilde{b}^s_i / \widetilde{b}^t_i  = \widetilde{b}^{s}_j / \widetilde{b}^{t}_j$.
This is equivalent to $\widetilde{\mathbf{b}}^s \parallelsum \widetilde{\mathbf{b}}^t$ for any $s,t \in \mS$ (see \eqref{eq.parallel}). Moreover, recall that the equilibrium first-order condition~\eqref{eq.foc} is given by
\begin{equation}\label{eq.foceqm}
    \mu_i \frac{\partial u_i}{\partial x^s_i} = p^s.
\end{equation}
Comparing this condition with~\eqref{eq.plannerfoc}, and using the parallel condition, we can choose multipliers and Pareto weights so that $\bX^*$ satisfies the planner’s KKT conditions. It follows that the competitive equilibrium is efficient. In particular, the associated Pareto weight vector is given by $\btheta^* = \mathbf{b} \odot \bmu^*$, where $\mathbf{b}$ is a representative vector proportional to $\widetilde{\mathbf{b}}^s$ for any $s \in \mS$.

If  the condition in \eqref{eq.parallel} is violated, a centrality-based wedge emerges between the private MRS and the planner’s MRS. In particular, the planner’s MRS cannot be equalized across individuals at the efficient allocation. This wedge implies that the equilibrium is not Pareto efficient, and can therefore be exploited to construct a Pareto improvement.

\begin{corollary}[Inefficiency]\label{cor.notparallel}
    Given an interior equilibrium $(\bX^*, \bp^*)$, if the centrality parallel condition~\eqref{eq.parallel} fails,  there exists an attainable allocation that strictly Pareto dominates~$\bX^*$. 
\end{corollary}

The proof proceeds by construction. We explicitly construct a small reallocation of consumption for goods \(s\) and \(t\) in the directions of \(\boldsymbol{\tau}^s\) and \(\boldsymbol{\tau}^t\), respectively. 
The existence of such directions of redistributions, \(\boldsymbol{\tau}^s\) and \(\boldsymbol{\tau}^t\), that make every consumer strictly better off is guaranteed precisely when the centrality parallel condition \eqref{eq.parallel} fails for goods \(s\) and \(t\). This follows from an application of Farkas’ Lemma. The detailed argument is provided in the proof of Corollary~\ref{cor.notparallel}.

\medskip 

Let us now study the Second Welfare Theorem. The construction is the converse of that for the First Welfare Theorem. Consider any interior efficient allocation $\hat{\bX}$. Setting initial endowments equal to $\hat{\bX}$, we show that this allocation can be supported by a competitive equilibrium under suitable supporting prices. When the parallel condition \eqref{eq.parallel} holds, equations \eqref{eq.plannerfoc} and \eqref{eq.foceqm} can be used to construct $\mu_i$ and $p^s$ in reverse, as in the proof of the First Welfare Theorem. Hence, the Second Welfare Theorem holds. If the parallel condition fails, then the Second Welfare Theorem fails as well, since no competitive equilibrium is efficient.

All three arguments in Theorem~\ref{thm.efficiency} show that the centrality parallel condition~\eqref{eq.parallel} is a necessary and sufficient condition for the efficiency of the competitive equilibrium.


\begin{lemma}[Two special cases]\label{lemma.symmetry}
    The centrality parallel condition holds if either
    \begin{enumerate}
        \item[(i)] $(\bG^s)^\intercal$ is regular for all $s \in \mS$;\footnote{That is, all agents have the same indegrees on layer $s$. Formally, $(\bG^s)^\intercal \mathbf{1} \parallelsum \mathbf{1}$.} or
        
        \item[(ii)] $\phi^s \bG^s$ is identical for all $s \in \mS$ (so that $\widetilde{\mathbf{b}}^s = \widetilde{\mathbf{b}}^t$).
    \end{enumerate}
\end{lemma}

Case~(i) generalizes Example~II in Section~\ref{sec_Examples}. Regularity on
layer \(s\) implies that all agents have the same indegree on that layer, that
is, \((\bG^s)^\intercal \mathbf{1} \parallelsum \mathbf{1}\), which in turn
implies \(\widetilde{\mathbf{b}}^s \parallelsum \mathbf{1}\). When regularity
holds on every layer, it follows that
\(\widetilde{\mathbf{b}}^s \parallelsum \widetilde{\mathbf{b}}^t
\parallelsum \mathbf{1}\) for all \(s,t\). Importantly, this condition still
allows for heterogeneity across layers, both in \(\phi^s\) and in the common
indegrees.

Case~(ii) captures a different form of symmetry across goods. Goods are symmetric in the sense that the externalities induce the same Leontief inverse matrix in all layers, that is, $\bM^s=\bM^t \equiv \bM$ for all $s,t \in \mS$. Hence, while the underlying networks may be irregular, they must coincide across layers. In this case, equation~\eqref{eq.mudetermine} yields
\begin{equation*}
    \frac{\bM \se}{\mathbf{1}^\intercal \bM \se} = \sum_{s \in \mS} \alpha^s \boldsymbol{\eta}^s = \frac{\bx^s}{\bar{\omega}^s},
\end{equation*}
Therefore, equilibrium consumption shares coincide with those in Remark~\ref{remark.benchmark}, namely the benchmark economy without externalities. This symmetry condition is in the same spirit as the analysis in \cite{arrow2009conspicuous} of conspicuous consumption and leisure, where competitive equilibrium can be Pareto optimal when relative-consumption and leisure comparison enter the utility function symmetrically.

Observe that the conditions in parts $(i)$ and $(ii)$ of Lemma~\ref{lemma.symmetry} are sufficient but not necessary. Indeed, some network structures satisfy neither condition in Lemma~\ref{lemma.symmetry} yet still lead to an efficient equilibrium. For example, consider the following two undirected networks:\footnote{Other  examples can be constructed  using core-periphery networks. }
    \begin{figure}[ht]
        \centering
        \begin{tikzpicture}[scale=0.8]
            \coordinate (a1) at (0,0.5);
            \coordinate (b1) at (-1.5,1.5);
            \coordinate (c1) at (1,1.5);
            \coordinate (d1) at (2.5,0.5);
            
            \def\layergap{2.2}
            
            \coordinate (a2) at ($(a1)+(2.5*\layergap,0)$);
            \coordinate (b2) at ($(b1)+(2.5*\layergap,0)$);
            \coordinate (c2) at ($(c1)+(2.5*\layergap,0)$);
            \coordinate (d2) at ($(d1)+(2.5*\layergap,0)$);

            \draw[line width=1.4pt] (b1) -- (a1);
            \draw[line width=1.4pt] (a1) -- (d1);
            \draw[line width=1.4pt] (d1) -- (c1);
            \foreach \p/\lab in {a1/1,b1/2,c1/3,d1/4}{
              \draw[fill=white, line width=0.9pt] (\p) circle (5pt);
              \node at (\p) {\small\bfseries $\lab$};
            }
            \node at (0.5,-0.2) {good 1: $\phi^1>0$};
            
            \draw[line width=1.4pt] (a2) -- (b2);
            \draw[line width=1.4pt] (b2) -- (c2);
            \draw[line width=1.4pt] (c2) -- (d2);
            \foreach \p/\lab in {a2/1,b2/2,c2/3,d2/4}{
              \draw[fill=white, line width=0.9pt] (\p) circle (5pt);
              \node at (\p) {\small\bfseries $\lab$};
            }
            \node at (6,-0.2) {good 2: $\phi^2<0$};

        \end{tikzpicture}
    \end{figure}

 For good~1, we obtain $\widetilde{b}^1_1 / \widetilde{b}^1_2 = 1 - \phi^1$, while for good~2, $\widetilde{b}^2_1 / \widetilde{b}^2_2 = \frac{1}{1 - \phi^2}$. Hence, the centrality parallel condition holds if $(1 - \phi^1)(1 - \phi^2) = 1$; for instance, this is satisfied when $\phi^1 = 0.2$ and $\phi^2 = -0.25$.

\begin{remark} 
The centrality-parallel condition should be interpreted as an exact zero-loss
condition  rather than as a generic restriction. Its role is to identify   sharp
benchmark cases in which competitive equilibria remain  Pareto efficient
despite the presence of network externalities. These benchmarks include
economically relevant environments, such as economies with regular networks or
with a common network structure across layers; they also nest the classical
Arrow--Debreu framework without externalities. When the centrality-parallel condition fails,
efficiency loss is not merely a binary outcome: the measures $L$ and
$\mathcal{CRU}$ vary continuously with the underlying primitives and provide a
quantitative assessment of departures from the zero-loss benchmarks, as illustrated by the three numerical examples in the next subsection.
\end{remark}

\subsection{Numerical examples}

We use three numerical examples to illustrate the welfare properties of equilibrium and the relationship between network multiplexity and efficiency.  Example IV illustrates the geometry of the welfare-based measure $L(\btheta)$ and its decomposition in~\eqref{eq.ltheta}. Example V focuses on how the departure from regular networks affects the  $\mathcal{CRU}$ measure. Example VI instead varies the degree of structural alignment across layers.

\paragraph*{Example IV}

We first illustrate $L(\btheta)$ in an economy with three consumers and three goods. Endowments and preferences are symmetric: each consumer is endowed with $1/3$ units of every good, and all goods receive equal preference weight, $\alpha^s\equiv 1/3$. We set $\phi^s\equiv0.4$ for all goods. The network associated with good~$1$ is empty, that associated with good~$2$ is a star centered at consumer~$1$, and that associated with good~$3$ is complete.  Figure~\ref{fig.ltheta} plots $L(\btheta)$ over the Pareto-weight simplex, using $\theta_2$ and $\theta_3$ as coordinates and setting $\theta_1=1-\theta_2-\theta_3$.\footnote{Since consumers~$2$ and~$3$ are symmetric in every layer, the contours are symmetric around the line $\theta_2=\theta_3$. } 
The three layer-specific weights are equal to:
\begin{equation*}
    \brho^1 = (0.36, 0.32, 0.32),
    \quad
    \brho^2 = (0.16, 0.42, 0.42),
    \quad
    \brho^3 = (0.36, 0.32, 0.32).
\end{equation*}
Since these weights $\brho$ are not the same across layers ($\brho^1=\brho^3\neq \brho^2$),  the centrality parallel condition is violated, so the equilibrium is inefficient.  

\begin{figure}[ht]
    \centering
    \includegraphics[width=0.5\linewidth]{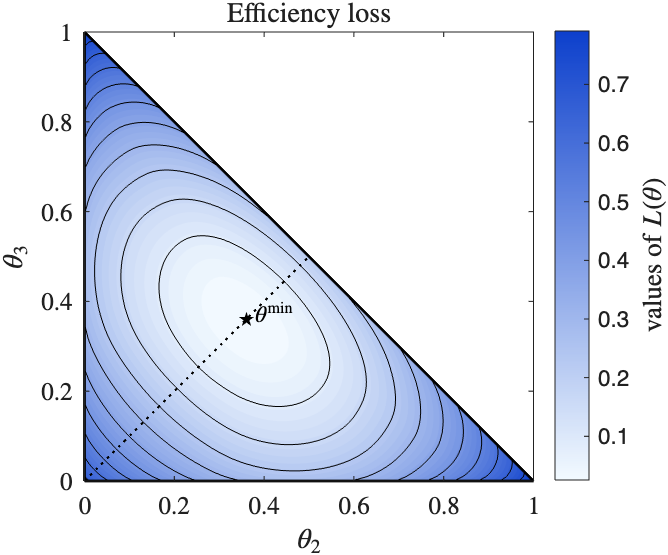}
    \caption{Efficiency loss $L(\btheta)$ over the Pareto-weight simplex.}
    \label{fig.ltheta}
\end{figure}

Solving the minimization problem \eqref{eq.minlosscru} yields $\btheta^{\min}=(0.28,0.36,0.36)$, with  minimal value   $L(\btheta^{\min})=0.025>0$. Therefore, by Proposition~\ref{prop.efficiencyloss}, $\mathcal{CRU} =\exp\{{-L(\btheta^{\min})}\}= 0.975<1$.  Moreover,
equation~\eqref{eq.ltheta} implies that any deviation from
\(\btheta^{\min}\) generates an additional loss equal to
\(D_{KL}(\btheta\Vert\btheta^{\min})\), as illustrated by the contour sets in
Figure~\ref{fig.ltheta}.

\paragraph*{Example V}

We next analyze how the departure from network regularity affects efficiency loss. Consider an economy with four consumers and two goods. Endowments and preferences are symmetric: each consumer is endowed with $0.25$ units of each good, and $\alpha^1=\alpha^2=1/2$. We consider $\phi^1=\phi^2=\phi$, with $\phi=0.2$ or $0.3$. The network associated with good~$1$ is a cycle and remains fixed, whereas the network associated with good~$2$ is
\begin{equation*}
    \bG^2(\beta)
    =
    (1-\beta)\mathbf{K}_4+\beta\mathbf{S}_4,
    \qquad
    \beta\in[0,1],
\end{equation*}
where $\mathbf{K}_4$ denotes the complete network and $\mathbf{S}_4$ a star centered at consumer~$1$. At $\beta=0$, the two layers have different topologies but are both regular; thus, the centrality parallel condition \eqref{eq.parallel} holds. As $\beta$ increases, the second layer gradually departs from regularity and approaches the asymmetric star network.

\begin{figure}[ht]
    \begin{subfigure}{0.43\textwidth}
        \centering
        \includegraphics[width=1\textwidth]{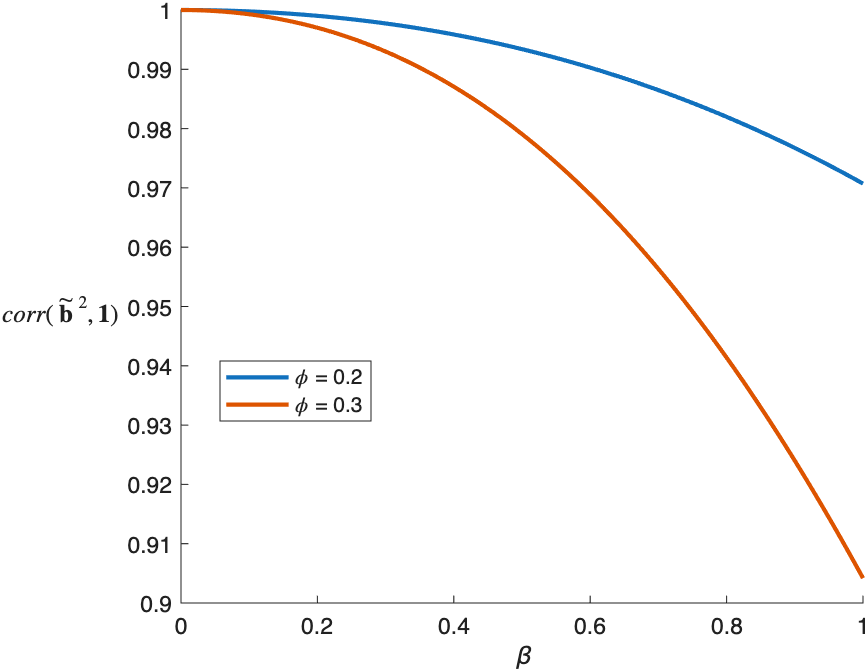}
        \caption{Measure of departure from regularity}
        \label{fig.lessregularb}
    \end{subfigure}
    \hfill
    \begin{subfigure}{0.43\textwidth}
        \centering
        \includegraphics[width=1\textwidth]{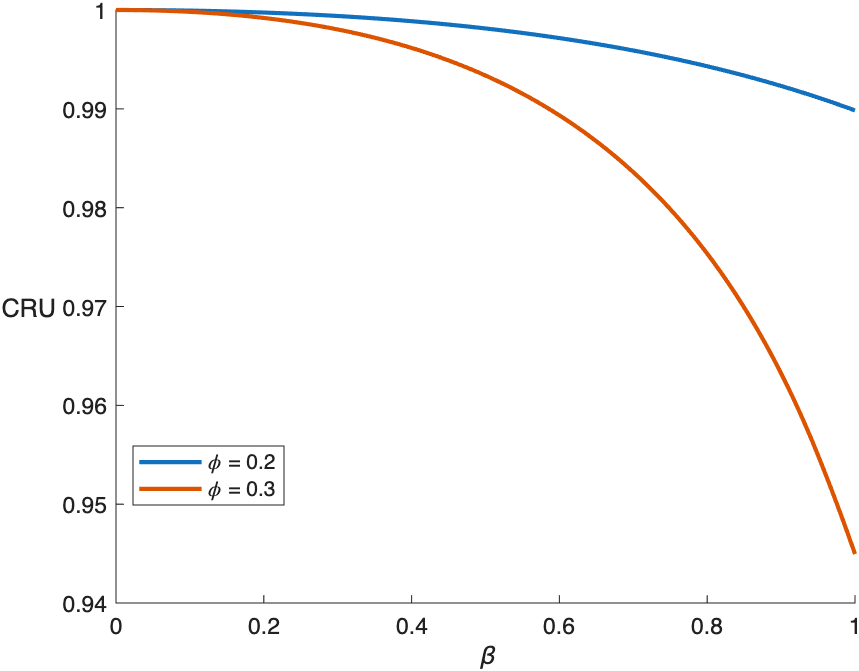}
        \caption{The values of $\mathcal{CRU}$}
        \label{fig.lessregularcru}
    \end{subfigure}
    \caption{Network regularity and equilibrium efficiency.}
    \label{fig.lessregular}
\end{figure}

To quantify the departure from regularity, we use the cosine similarity between the centrality vector of the second layer and the uniform vector, $\operatorname{corr}(\widetilde{\mathbf{b}}^2,\mathbf{1}) = 
\langle \widetilde{\mathbf{b}}^2,\mathbf{1} \rangle/(\lVert \widetilde{\mathbf{b}}^2 \rVert \cdot \lVert \mathbf{1} \rVert)$.
This measure equals one precisely when $\widetilde{\mathbf{b}}^2$ is parallel to $\mathbf{1}$, as is the case for a regular network. Figure~\ref{fig.lessregular} shows that the decline in regularity is accompanied by a decline in the CRU. At $\beta=0$, both layers are regular and hence $\widetilde{\mathbf{b}}^1 \parallelsum \widetilde{\mathbf{b}}^2 \parallelsum \mathbf{1}$, even though their network topologies are different. The equilibrium is efficient and, thus, $\mathcal{CRU}=1$. As $\beta$ increases, the second layer becomes progressively less regular, generating heterogeneity in consumers' relative network positions and breaking the centrality parallel condition. The corresponding decline in $\mathcal{CRU}$ illustrates that asymmetry in network positions can reduce efficiency when it creates disagreement across layers about consumers' relative network positions. Observe that $\phi$ can amplify this discrepancy when $\beta$ increases as shown in Figure~\ref{fig.lessregular}.

\paragraph*{Example VI}
Lastly, we study how equilibrium efficiency changes as initially identical network layers become structurally different. Consider an economy with four consumers and $m$ goods, where $m=2$ or $3$. Endowments and preferences are symmetric: each consumer is endowed with $0.25$ units of every good, and all goods receive equal preference weight, $\alpha^s\equiv1/m$. We set $\phi^s\equiv\phi$ for all goods, with $\phi=0.2$ or $0.3$.

Suppose that good 1 is associated with an unweighted symmetric star network, denoted by $\bG^1 = \mathbf{S}_4$. The network for good 2 is a convex combination of the star network and its complement: $\bG^2 = (1-\beta) \mathbf{S}_4+\beta \mathbf{S}^c_4$ with $\beta \in [0,1]$. The cosine similarity between $\bG^1$ and $\bG^2$ decreases with $\beta$, so $\beta$ captures the extent of structural divergence between the first two network layers: they coincide at $\beta=0$ and become progressively less similar as $\beta$ increases.  When $m=3$, we assume that good 3 has an empty network. 

\begin{figure}[ht]
    \begin{subfigure}{0.43\textwidth}
        \centering
        \includegraphics[width=1\textwidth]{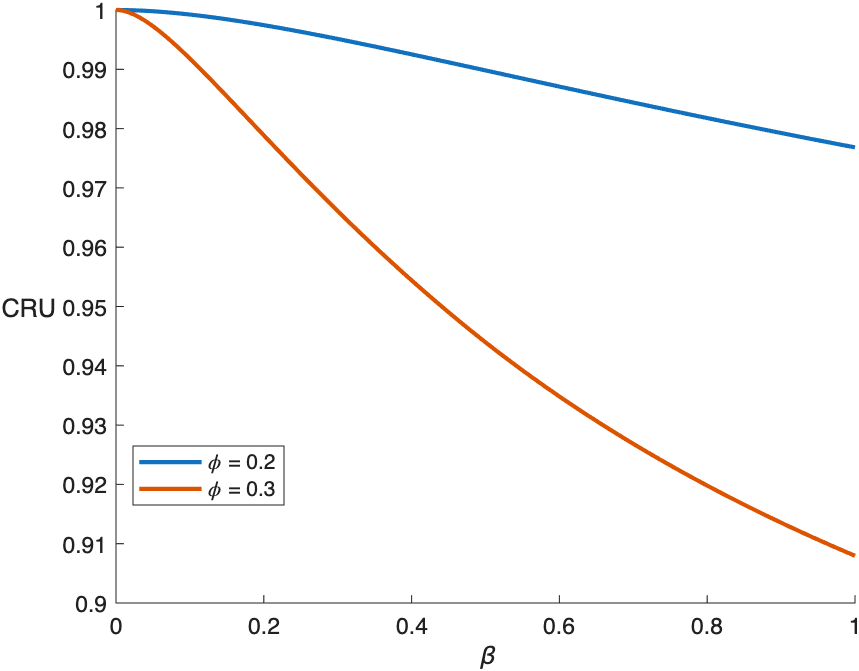}
        \caption{Two goods}
        \label{fig.efficiency2}
    \end{subfigure}
    \hfill
    \begin{subfigure}{0.43\textwidth}
        \centering
        \includegraphics[width=1\textwidth]{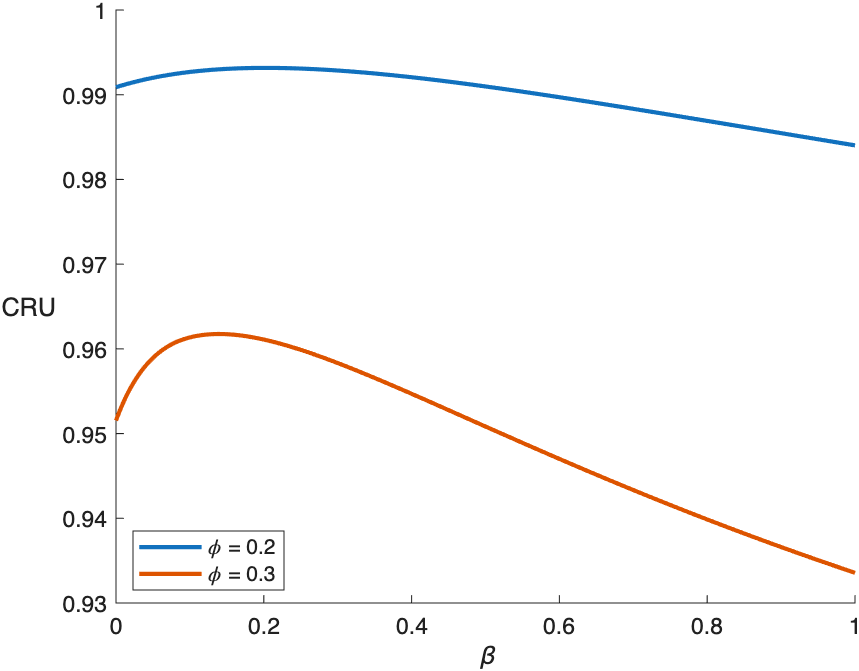}
        \caption{Three goods}
        \label{fig.efficiency3}
    \end{subfigure}
    \caption{The values of CRU for different values of $\beta$.}
    \label{fig.efficiency}
\end{figure}

Figure~\ref{fig.efficiency2} considers the two-good economy ($m=2$). At $\beta=0$, the two layers are identical, so the centrality parallel condition holds and the competitive equilibrium is efficient. As $\beta$ increases, the two layers become increasingly misaligned and $\mathcal{CRU}$ declines monotonically.

The three-good economy ($m=3$) plotted in Figure~\ref{fig.efficiency3}  exhibits a different pattern. Because the star network associated with good~$1$ is not aligned with the empty network associated with good~$3$, the centrality parallel condition fails for every value of $\beta$, and the equilibrium is never fully efficient. More importantly, efficiency is no longer monotone in the similarity between goods~$1$ and~$2$. Starting from $\beta=0$, making the second layer less similar to the first can initially raise the CRU because the induced centrality weights of good~$2$ may move closer to those of the third layer. Further increases in $\beta$, however, eventually enlarge the overall dispersion across the three layer-specific weights, causing efficiency to decline. This non-monotonicity illustrates an implication of Proposition~\ref{prop.efficiencyloss}: with more than two layers, the efficiency loss depends on the joint dispersion of the entire collection $\{\brho^s\}_{s\in\mS}$, rather than on the similarity between any particular pair of networks.


\subsection{General utility functions}\label{sec_gen_ut}





We generalize the Cobb--Douglas utility function \eqref{eq.utility}, which imposes linear aggregation of externalities, thereby ruling out cross-good spillovers and implying constant networks. We consider a general utility function that relaxes these restrictions and allows for income effects and, hence, distributional effects. We show that the link between our two efficiency measures, $\mathcal{CRU}$ and $L(\btheta)$, is preserved and derive an efficiency condition analogous to the centrality parallel condition. Under additive separability across goods, we further obtain a statistical-distance representation of $L(\btheta)$.

We now introduce the general utility specification and the assumptions under which these results are established.
We assume that consumer $i$'s utility $u_i(\bX)$ is smooth, jointly concave in $\bX$ and, for each $\bx_{-i}$, is strictly concave and strictly increasing in $\mathbf{x}_i$, with $\frac{\partial u_i(\bX)}{\partial x_i^s}>0$ for each good $s$. Define the normalized marginal externality of consumer $j$'s consumption of good $s$ on consumer $i$ by\footnote{ Clearly, this object is invariant under monotone transformations of $u_i$.}
\begin{equation*}
    g^s_{ij}(\bX) := \frac{\partial u_i(\bX)/\partial x^s_j}{\partial u_i(\bX)/\partial x^s_i}, \quad i \neq j,
\end{equation*}
where $\phi^s$ is included in  $g^s_{ij}$. We adopt the convention that $g^s_{ii} \equiv 0$. Unlike in the baseline model in \eqref{eq.utility}, the network $\bG^s(\bX)$ now depends on the allocation $\bX$. Despite this additional endogeneity, efficient allocations can still be characterized by the planner's problem in Lemma~\ref{lemma.efficienttoplanner}. Comparing the planner's and consumers' first-order conditions~\eqref{eq.plannerfoc} and~\eqref{eq.foceqm} at an interior equilibrium $\bX^*$ implies that there must exist a common positive vector $\mathbf b$ such that $[\mathbf{I}_n+(\bG^s(\bX^*))^\intercal] \mathbf{b} \parallelsum \mathbf{1}$ for all $s \in \mS$. 
Theorem~\ref{app:thm.efficiency.general} in Online Appendix~\ref{appsec.general} formalizes this condition, which is the counterpart of the centrality parallel condition in Theorem~\ref{thm.efficiency}. Unlike in the baseline model, however, it is no longer a global restriction on the primitive network structure, but an allocation-specific condition evaluated at the equilibrium. In this general utility function case, because the network is a function of the allocation, we can study income effects. In Remark~\ref{app:remark.income} in the Online Appendix, we illustrate these income effects using Example~3 in \cite{dufwenbergOtherRegardingPreferencesGeneral2011}, which considers a one-good exchange economy with rich and poor agents. Rich agents have other-regarding preferences, with a high marginal benefit from increasing the wealth of agents in poverty. We show that, in this setting, our generalized centrality parallel condition in Theorem~\ref{app:thm.efficiency.general} coincides with their condition for market efficiency. 
In Remark~\ref{app:remark.crossgood}, we study another channel accommodated by general utility functions, the cross-good spillovers. We show that asymmetries in cross-good spillovers can generate misalignment across the marginal-externality networks of different goods, thus violating the generalized centrality-parallel condition and leading to inefficiency.

Under a general utility function,   efficiency measures $\mathcal{CRU}$ and $L(\btheta)$ can be defined in a similar way, although explicit characterizations  are generally unavailable. Nevertheless, the link between $\mathcal{CRU}$ and $L(\btheta)$  established in Proposition \ref{prop.efficiencyloss} continues to hold, albeit in a different form. Specifically, for each $\gamma\in[0,1]$, define the maximal weighted welfare by
\begin{eqnarray*}
    V_{\gamma}(\btheta)
    :=& \max_{\bX \geq \mathbf{0}}
    \btheta^\intercal \mathbf{u}(\bX) \\
    \text{s.t.} &\ \sum_{i \in \mN} x_i^s
    \leq \gamma \bar{\omega}^s, \quad \forall s \in \mS,
\end{eqnarray*}
and the corresponding welfare gap by
\begin{equation*}
    L_\gamma(\btheta)
    :=
    V_{\gamma}(\btheta)-\btheta^\intercal \mathbf{u}(\bX^*).
\end{equation*}
Here, $\gamma$ scales aggregate resources: when $\gamma=1$, the economy has
the original aggregate endowments, while lower values of $\gamma$ correspond
to proportional reductions in available resources. As resources are reduced
($\gamma\downarrow$), the feasible set of allocations contracts, and hence
$V_{\gamma}(\btheta)$ decreases for a given $\btheta$. Consequently,
$L_\gamma(\btheta)$ also decreases, holding $\bX^*$ fixed.  At the critical resource level, the equilibrium utility profile lies on the boundary of the attainable utility set. Hence, $L_\gamma(\btheta)=0$ for at least one supporting Pareto weight $\btheta$, while $L_\gamma(\btheta)\geq0$ for all
$\btheta\in\Delta$.  Proposition~\ref{app:prop.cru} in Online Appendix~\ref{appsec.general} shows that this critical value of $\gamma$ is precisely the minimum fraction of resources needed to sustain the equilibrium utility profile:
\begin{equation*}
    \mathcal{CRU} =
    \inf
    \left\{
    \gamma \in [0,1]:
    \inf_{\btheta \in \Delta} L_\gamma(\btheta) \geq 0
    \right\}.
\end{equation*}





A sharper characterization is available under additional structure. In particular, suppose there are no cross-good spillovers and preferences are additively separable among different goods, i.e., $u_i(\bx) = \sum_{s \in \mS} u^s_i(\bx^s)$. Since the resource constraints are also separable across goods, the planner problem can be decomposed across layers. Let $V^s(\btheta)$ denote the maximal weighted utility generated by good $s$ under weight vector $\btheta$, defined in  equation~\eqref{app:eq.vs}. Proposition~\ref{app:prop.bregman} shows that there exists a layer-specific Pareto weight $\brho^s$ under which the equilibrium allocation $\bx^{s*}$ solves the corresponding layer-specific planner problem. Moreover, part~(i) of Proposition \ref{prop.efficiencyloss} extends to this setting. In particular,
\begin{equation*}
    L(\btheta) = \sum_{s \in \mS} D_{V^s}(\btheta,\brho^s) \geq 0,
\end{equation*}
where $D_{V^s}$ is the Bregman distance (\citealp{BREGMAN1967200}) generated by $V^s$.\footnote{For a differentiable convex function $F$, the Bregman divergence between $\mathbf y$ and $\mathbf z$ is
\begin{equation*}
    D_F(\mathbf{y},\mathbf{z}) :=
    F(\mathbf{y}) - F(\mathbf{z}) - \nabla F(\mathbf{z})^\intercal(\mathbf{y}-\mathbf{z}).
\end{equation*}
}
Thus, the efficiency loss can again be interpreted as the discrepancy between a common Pareto weight $\btheta$ and the layer-specific weights $\{\brho^s\}_{s\in\mS}$. Under CRRA or CARA sub-utilities with linear aggregations (as in \eqref{eq_eff_cons}), Proposition~\ref{app:prop.special.utilities}
provides explicit expressions for $L(\btheta)$, and
Corollary~\ref{app:cor.special} gives sharper characterizations of
$\mathcal{CRU}$. Finally, the baseline Cobb--Douglas result in
Proposition~\ref{prop.efficiencyloss} is obtained as the limiting case of
CRRA utility, in which the Bregman distance $D_{V^s}$ reduces to the
KL divergence.

Although the Cobb--Douglas utility specification in \eqref{eq.utility} is not
essential for the welfare results, its main advantage is that it yields
explicit equilibrium formulas for prices and allocations
(Theorem~\ref{thm.interior}). These closed-form expressions facilitate the
comparative-statics analysis developed in the next section.


\section{Comparative statics analysis}\label{sec_CS}

We now turn to comparative statics with a special focus on the effects of endowment changes.\footnote{Comparative statics with respect to other parameters, such as the layer weights $\alpha^s$, externality intensities $\phi^s$ or the network links $g_{ij}^s$, are carried out in Online Appendix \ref{appsec.cs}, but are omitted here for brevity.} Since the endowment shares $\{\boldsymbol{\eta}^{s}\}_{s \in \mS}$ enter the matrix $\overline{\bM}$ determining $\se$, it is convenient to fix the directions of the change in order to obtain a well-defined derivative. Formally, an endowment change is characterized by a collection of direction vectors $\{\btau^s\}_{s \in \mS}$, where each $ \btau^s \in \mathbb{R}^{n}$, and by a scalar magnitude $l \in \mathbb{R}$. The new endowment for each good $s$ is given by $\w^s(l) = \w^s+l\btau^s$.


Under this definition, all equilibrium objects are well-defined functions of $l$ in an open neighborhood of 0. Since the effective endowment vector is only determined up to a positive scalar multiple, we fix a reference solution $\bmu^{*0}$ at $l=0$;  hence, we also  fix a reference price $\bp^{*0}$. Consider a differentiable selection $\bmu(l)$ passing through $\bmu^{*0}$, for any equilibrium object $\mathbf z(l)$, define
\begin{equation*}
    \left. \dot{\mathbf{z}}  := \frac{\partial \mathbf{z}(l)}{\partial l} \right|_{l=0}.
\end{equation*}



The following proposition derives our main comparative statics results on relative prices and welfare.\footnote{In the proof, we also derive the effects on consumptions. We omit them here for brevity.}


\begin{proposition}[Comparative statics]\label{prop.comparative}
    If Assumption \ref{assump.interior} holds, an endowment change has the following effects.

    \begin{enumerate}
        \renewcommand{\labelenumi}{(\roman{enumi})}
        \item Price effect 
        \begin{equation}\label{eq.priceeffect}
            \frac{\dot{p}^{s*}}{p^{s*}} - \frac{\dot{p}^{t*}}{p^{t*}}
            =
            \underbrace{
            \frac{b^s(\dot{\bmu}^*)}{b^s(\bmu^{*0})}
            -
            \frac{b^t(\dot{\bmu}^*)}{b^t(\bmu^{*0})}
            }_{\text{redistribution effect}}
            \underbrace{
            -
            \frac{\mathbf{1}^\intercal\btau^s}{\bar{\omega}^s}
            +
            \frac{\mathbf{1}^\intercal\btau^t}{\bar{\omega}^t}
            }_{\text{aggregate effect}}.
        \end{equation}


        \item Welfare effect
        \begin{equation}\label{eq.welfareeffect}
            \dot{u}_i^*
            =
            \underbrace{
            \frac{\dot{\mu}^*_i}{\mu^{*0}_i}
            }_{ \text{effective endowment effect} }
            -
            \underbrace{
            \sum_{s \in \mS}\alpha^s
            \frac{b^s(\dot{\bmu}^*)}{b^s(\bmu^{*0})}
            }_{\text{price effect}}
            +
            \underbrace{
            \sum_{s \in \mS}
            \alpha^s\frac{\mathbf{1}^\intercal\btau^s}{\bar{\omega}^s}
            }_{\text{aggregate effect}},
        \end{equation}

    \end{enumerate}
    where $\dot{\bmu}^*$ solves
    \begin{equation}\label{eq.dotmu}
       \overline{\bM}\dot{\bmu}
        =
        \underbrace{
        \sum_{s \in \mS} p^{s*}(\btau^s-\boldsymbol{\eta}^s \mathbf{1}^\intercal\btau^s)
        }_{=:\tilde{\btau}}
        . 
    \end{equation}
\end{proposition}

Since all equilibrium objects are functions of $\bmu^{*0}$, it suffices to characterize $\dot{\bmu}^*$ to derive the comparative statics. Recall that $\bmu^{*0}$ solves equation~\eqref{eq.mudetermine}. Differentiating that equation yields
\begin{equation*}
    \overline{\bM}\dot{\bmu} + \dot{\overline{\bM}} \bmu^{*0} = \mathbf{0},
\end{equation*}
and a direct calculation of $\dot{\overline{\bM}}$ gives equation~\eqref{eq.dotmu}. This linear system has multiple solutions. In the proof, we show that expressions in~\eqref{eq.priceeffect} and~\eqref{eq.welfareeffect} are independent of the  selection of solution.

Both the price effect and the welfare effect admit a decomposition into a redistribution component and an aggregate component. For each good $s$, the term $\btau^s-\boldsymbol{\eta}^s \mathbf{1}^\intercal\btau^s$ represents the redistribution effect of the endowment change. The vector $\tilde{\btau}$ captures the implied net redistribution and determines $\dot{\bmu}^*$. The aggregate component, in contrast, depends only on the change in total endowment, represented by $\mathbf{1}^\intercal\btau^s$. Therefore, if all $\btau^s$ are pure redistributions, that is, $\mathbf{1}^\intercal\btau^s=0$ for all $s \in \mS$, the aggregate components in~\eqref{eq.priceeffect} and~\eqref{eq.welfareeffect}  vanish.



For the price effect, we consider a transfer from consumer $i$ to consumer $j$. Fix a good $t$ and suppose that the endowment change is given by $\btau^t = \btau$, where $\tau_j=-\tau_i >0$ and $\tau_k=0$ for all $k \neq i,j$, while $\btau^s = \mathbf{0}$ for all $s \neq t$.\footnote{The effect of any redistribution will be a linear combination of such bilateral transfers.} By the independence of solution selection, we can choose a specific solution to \eqref{eq.mudetermine}, denoted by $\bmu^{*0}:=\bH^{-1} \w^1$ where $\bH := \overline{\bM} + \alpha^1\boldsymbol{\eta}^1 \mathbf{1}^\intercal\bM^1$.  This selection always normalizes the price of good 1 to 1 and therefore implies that  $\dot{p}^{1*} = 0$. We can further show that $\dot{\bmu}^* = p^{t*}\bH^{-1} \btau$.  Hence, by equation \eqref{eq.priceeffect}, the induced (relative) price effect on good $s$ is
\begin{equation*}
    \frac{\dot{p}^{s*}}{p^{s*}} - \frac{\dot{p}^{1*}}{p^{1*}} = \frac{p^{t*}(\widetilde{\mathbf{b}}^s)^\intercal \bH^{-1}\btau}{b^s(\bmu^{*0})}.
\end{equation*}
This shows that the multiplex network structure plays a crucial role in determining the sign of the price effect. 

\begin{proposition}[Centrality and price effects]\label{prop.priceeffect}
    A transfer from consumer $i$ to consumer $j$ increases relative price $p^{s*}/p^{1*}$ if and only if $c_j^s > c_i^s$, where
    the generalized influence centrality is given by 
    \begin{equation*}
        \mathbf{c}^s := (\bH^\intercal)^{-1} \widetilde{\mathbf{b}}^s. 
    \end{equation*}
\end{proposition}

Proposition \ref{prop.priceeffect} shows that a transfer affecting one good may
propagate across layers and change the equilibrium prices of other
goods. This channel is absent in the benchmark without
externalities. With multiplex networks, prices reflect heterogeneous aggregate centralities across layers. As a result, the direction of the induced price effect is determined by the layer-specific generalized influence centrality $\mathbf{c}^s$.\footnote{If all layers have the same network structure, relative prices do not change after endowment transfers, the same as the benchmark.} This further shows that inequality in network centrality may reinforce inequality in endowments. When resources are transferred from a less central consumer to a more central one, the induced increase in equilibrium prices further disadvantages the consumer who loses endowment.

\medskip

\cite{ghiglinoKeepingNeighborsSocial2010} study a setting closely related to ours, but with only two goods: one private (corresponding to an empty network) and one network good subject to social comparison (i.e., $\phi < 0$). Their framework is therefore a special case of ours. In Online Appendix \ref{appsec.examplegg}, we formally show  how their results can be recovered using our $\bmu$-based approach. Moreover, in their setting, when externalities are small, the standard Katz--Bonacich centrality and the generalized influence centrality $\mathbf{c}^s$ induce the same ranking. Hence, the former is sufficient to determine the direction of price effects. By contrast, this may not hold under multiplexity with multiple network goods. The following example illustrates that multiplexity can even reverse the sign of the price effect.

\paragraph*{Example VII}  
Consider an economy with four consumers and three goods. Good~1 is private, while goods~2 and~3 generate network externalities with $\phi^2 = \phi^3 = -0.04$. The corresponding network structures are depicted in Figure~\ref{fig.network}. We assume symmetric endowments, $\omega_i^s = 0.25$ for all $s \in \mS$ and $i \in \mN$.
We fix the preference weight on good~1 at $\alpha^1 = 1/3$, and let $\alpha^2 = \sigma$ and $\alpha^3 = 2/3 - \sigma$, where $\sigma \in (0,\,2/3]$ varies to examine the role of multiplexity in equilibrium outcomes.

When $\sigma = 2/3$, we have $\alpha^3 = 0$,  which effectively eliminates the third good from the economy. The model therefore collapses to the framework of \cite{ghiglinoKeepingNeighborsSocial2010}. In this case, a transfer from consumer~3 to consumer~2 increases $p^{2*}/p^{1*}$, since consumer~2 is more K-B central than consumer~3 in network~2.

In the multiplex setting, i.e., when $0<\sigma < 2/3$, a decrease in $\sigma$ raises the weight on layer~3. Numerical computations show that, when $\sigma$ is sufficiently small, under multiplexity, consumer~3 becomes more $c$-central than consumer~2, i.e., $c^2_3>c^2_2$, reversing the comparative statics: $p^{2*}/p^{1*}$ decreases rather than increases. Table~\ref{table.priceeffect} reports the generalized influence centrality $\mathbf{c}^2$ for good 2, along with the corresponding directional changes of $p^{2*}/p^{1*}$ for different values of $\sigma$. This example illustrates that multiplex interactions can overturn in some cases the comparative statics derived from \cite{ghiglinoKeepingNeighborsSocial2010}.

    \begin{figure}[htbp]
        \centering
    
        \begin{minipage}[t]{0.45\linewidth}
            \vspace{0pt}
            \centering
            \resizebox{\linewidth}{!}{%
            \begin{tikzpicture}[scale=0.9]
                
                \coordinate (a1) at (0,0.5);
                \coordinate (b1) at (-1.5,1.5);
                \coordinate (c1) at (1,1.5);
                \coordinate (d1) at (2.5,0.5);
                
                \def\layergap{2.2}
                
                \coordinate (a2) at ($(a1)-(0,\layergap)$);
                \coordinate (b2) at ($(b1)-(0,\layergap)$);
                \coordinate (c2) at ($(c1)-(0,\layergap)$);
                \coordinate (d2) at ($(d1)-(0,\layergap)$);
    
                \coordinate (a3) at ($(a2)-(0,\layergap)$);
                \coordinate (b3) at ($(b2)-(0,\layergap)$);
                \coordinate (c3) at ($(c2)-(0,\layergap)$);
                \coordinate (d3) at ($(d2)-(0,\layergap)$);
                
                \path[fill=gray!12, draw=none]
                  (-3, 2) -- (1, 2) -- (4,0) -- (0,0) -- cycle;
                \path[fill=gray!12, draw=none]
                  ($(-3, 2)-(0,\layergap)$) -- ($(1, 2)-(0,\layergap)$)
                  -- ($(4, 0)-(0,\layergap)$) -- ($(0, 0)-(0,\layergap)$) -- cycle;
                \path[fill=gray!12, draw=none]
                  ($(-3, 2)-(0,2*\layergap)$) -- ($(1, 2)-(0,2*\layergap)$)
                  -- ($(4, 0)-(0,2*\layergap)$) -- ($(0, 0)-(0,2*\layergap)$) -- cycle;
                
                \draw[line width=1.4pt] (a2) -- (b2)
                  node[pos=0.4, above, yshift=1mm, font=\small] {$1.01$};
                \draw[line width=1.4pt] (a2) -- (c2)
                  node[midway, right, font=\small] {$1$};
                \draw[line width=1.4pt] (a2) -- (d2)
                  node[midway, below, font=\small] {$1$};
                \foreach \p/\lab in {a2/1,b2/2,c2/3,d2/4}{
                  \draw[fill=white, line width=0.9pt] (\p) circle (5pt);
                  \node at (\p) {\small\bfseries $\lab$};
                }
                \node[anchor=west] at (2.5,-\layergap+1.2) {good 2: star};
                
                \draw[line width=1.4pt] (a3) -- (b3)
                  node[midway, left, yshift=-1mm, font=\small] {$1$};
                \draw[line width=1.4pt] (b3) -- (c3)
                  node[midway, above, font=\small] {$1$};
                \draw[line width=1.4pt] (c3) -- (d3)
                  node[midway, right, yshift=1mm, font=\small] {$1$};
                \foreach \p/\lab in {a3/1,b3/2,c3/3,d3/4}{
                  \draw[fill=white, line width=0.9pt] (\p) circle (5pt);
                  \node at (\p) {\small\bfseries $\lab$};
                }
                \node[anchor=west] at (2.5,-2*\layergap+1.2) {good 3: line};
    
                \foreach \p/\lab in {a1/1,b1/2,c1/3,d1/4}{
                  \draw[fill=white, line width=0.9pt] (\p) circle (5pt);
                  \node at (\p) {\small\bfseries $\lab$};
                }
                \node[anchor=west] at (2.5,1.2) {good 1: private};
                
            \end{tikzpicture}
            }
            \caption{Network structures for the three goods. Numbers on the edges indicate the weights of the undirected links.}
            \label{fig.network}
        \end{minipage}\hfill
        \begin{minipage}[t]{0.4\linewidth}
            \vspace{2em}
            \centering
            \captionof{table}{Multiplex centrality and price change}
            \label{table.priceeffect}
            
            \renewcommand{\arraystretch}{1.3}
            \resizebox{\linewidth}{!}{%
            \begin{tabular}{lllc}
                \toprule
                \multicolumn{1}{c}{$\sigma$} 
                & \multicolumn{1}{c}{$c^2_2$} 
                & \multicolumn{1}{c}{$c^2_3$} 
                & \multicolumn{1}{c}{$p^{2*}/p^{1*}$ change} \\
                \midrule
                $\frac{2}{3}$ & $3.1729^*$ & $3.1724$   &   $+$\\
                $\frac{1}{2}$ & $3.1736$   & $3.1737^*$ &   $-$\\
                $\frac{1}{3}$ & $3.1743$   & $3.1749^*$ &   $-$\\
                \bottomrule
            \end{tabular}
            }
        \end{minipage}
    \end{figure}

We now turn to the welfare effect, which is determined in \eqref{eq.welfareeffect}. The redistribution component can be decomposed into an effective endowment effect and a price effect. The first term captures the change in consumer $i$'s effective endowment $\mu^*_i$. The second term is the weighted aggregation of the price effects across goods. It follows that consumer $i$ benefits from the endowment change if and only if the sum of her normalized effective endowment effect and the aggregate endowment effect dominates the price effect.

Moreover, differences in welfare effects across consumers are driven entirely by the effective endowment component, since the price effect and the aggregate endowment effect are common across consumers. In particular, for any two consumers $i$ and $j$,
\begin{equation*}
    \dot{u}_i^* - \dot{u}_j^*
    =
    \frac{\dot{\mu}^*_i}{\mu_i^{*0}}
    -
    \frac{\dot{\mu}^*_j}{\mu_j^{*0}}.
\end{equation*}
Hence, the ranking of consumers' welfare gains and losses is determined solely by their normalized changes of effective endowments.

Intuitively, a pure redistribution does not create additional resources. It only reallocates effective endowments across consumers while inducing a common price effect faced by all consumers. Hence, any gain generated through the effective endowment component for some consumers must be offset by losses for others. Redistribution can alter the distribution of welfare, but it cannot make all consumers better off.

\begin{corollary}[No Pareto-improving redistribution]\label{cor.infeasible}
    There exists a vector $\boldsymbol{\varpi}  \in \mathbb{R}^n_{++}$ such that, for any pure redistribution, $\boldsymbol{\varpi}^\intercal \dot{\mathbf{u}}^* = 0$. Therefore, any redistribution cannot generate a local Pareto improvement at the competitive equilibrium.
    
\end{corollary}

In the proof, we construct the positive vector
\begin{equation*}
    \boldsymbol{\varpi} = \sum_{s \in \mS} \alpha^s 
    \frac{\bmu^{*} \odot \widetilde{\mathbf{b}}^s}{(\bmu^{*})^\intercal \ \widetilde{\mathbf{b}}^s} \in \mathbb{R}^n_{++}
\end{equation*}
This vector can be interpreted as the welfare weights under which the gains and losses from any pure redistribution exactly balance out. Since all components of $\boldsymbol{\varpi}$ are strictly positive, the equality $\boldsymbol{\varpi}^\intercal \dot{\mathbf{u}}^* = 0$ precludes $\dot{\mathbf{u}}^* \geq \mathbf{0}$ with at least one strict inequality. Hence, no pure redistribution can be locally Pareto improving.
 
In particular, when the centrality parallel condition~\eqref{eq.parallel} holds, $\boldsymbol{\varpi} = \frac{\bmu^{*} \odot \mathbf{b}}{(\bmu^{*})^\intercal \ \mathbf{b}} \parallelsum \btheta^*$, the associated Pareto weight in Theorem~\ref{thm.efficiency} (i). Hence, the corollary is consistent with the First Welfare Theorem: under the centrality parallel condition, the competitive equilibrium utility vector already lies on the Pareto frontier, so no feasible redistribution can generate a Pareto improvement.

Finally, note that the key object behind both the price and welfare changes is $\dot{\bmu}^*$, which embeds within-layer network effects, cross-layer interactions, and the endowment distribution. This is precisely where the multiplex structure matters.

\section{Policy interventions: Lindahl equilibrium}\label{sec_policy}

When the centrality parallel condition fails, the equilibrium allocation is inefficient; hence, it can be improved through direct reallocation of consumption levels (Corollary~\ref{cor.notparallel}).
However, such centralized interventions may be impractical. In this section, we turn to Lindahl equilibrium \citep{foley70,BERGSTROM1970383,delmercatoSufficientConditionsSimple2023,bonnisseau2023existence} and examine whether it can deliver Pareto improvements.

Following \cite{arrow1969organization}, one approach to addressing inefficiencies arising from externalities is to complete the missing markets associated with them. In particular, each consumer faces personalized prices for the consumption of other consumers’ goods, thereby internalizing the external effects generated by others’ consumption choices. These  consumptions  enter only the individual’s own utility and do not affect the utilities of other consumers. Under such a market structure, the resulting competitive equilibrium restores efficiency.

Following this insight, we define the extended price system and the extended consumption bundles as follows. An extended price system is $\vec{\mathbbm{P}} = \{ \vec{\lp},(\vec{\lp}_{ij})_{i,j \in \mN}\} \in \mathbb{R}^{|\mS|}_+ \times \mathbb{R}^{|\mS| \times |\mN|^2}$, where $\vec{\lp}$ is the price vector for physical goods, $\vec{\lp}_{ii}$ is the effective price paid by consumer~$i$ for her own consumption, and $\vec{\lp}_{ij}$ is the personalized price paid by $i$ for the externality generated by $j$'s consumption. The price $\lp^s_{ij}$ can be negative, which means consumer $i$ is compensated for the externality generated by consumer $j$’s consumption of good $s$. 

The extended price system is called compatible if
\begin{equation*}
    \vec{\lp}_{ii} = \vec{\lp} - \sum_{j \neq i} \vec{\lp}_{ji}, \ \forall i \in \mN.
\end{equation*}
This condition requires the effective price paid by consumer $i$ for her own consumption to be equal to the physical price net of the payments made by other consumers for the externalities generated by $i$'s consumption.

The extended consumption bundle of consumer $i$ is $\vec{\mathbbm{X}}_i = (\vec{\mathbbm{x}}_{ij})_{j \in \mN} \in \mathbb{R}_+^{|\mS| \times |\mN|}$, where $\vec{\mathbbm{x}}_{ii}$ is consumer $i$'s own consumption, and $\vec{\mathbbm{x}}_{ij}$ denotes $i$'s demand for the consumption vector of $j$. Given an extended price system, the budget set of consumer $i$ is
\begin{equation*}
    B_i(\vec{\mathbbm{P}}) = \left\{\vec{\mathbbm{X}}_i \in \mathbb{R}_+^{|\mS| \times |\mN|}: \sum_{j \in \mN} \vec{\lp}_{ij} \cdot \vec{\mathbbm{x}}_{ij}  \leq \vec{\lp} \cdot \w_i \right\}
\end{equation*}

\begin{definition}
    A Lindahl equilibrium consists of extended consumption bundles $\{ \vec{\mathbbm{X}}_i^*\}_{i \in \mN}$ and a compatible price system $\vec{\mathbbm{P}}^*$, such that:
    \begin{enumerate}
        \item[(i)] Utility maximization: for each consumer $i \in \mN$,
        \[
        \vec{\mathbbm{X}}^*_i \in \arg\max \{u_i(\vec{\mathbbm{X}}_i): \vec{\mathbbm{X}}_i \in B_i({\vec{\mathbbm{P}}^*})\}.
        \]

        \item[(ii)] Market clearing and consistency: there exists an attainable allocation $\vec{\mathbbm{X}}^* = (\vec{\mathbbm{x}}_j^*)_{j\in\mN} \in \mathcal{A}$ such that for every $i,j\in\mN$, $\vec{\mathbbm{x}}_{ij}^* = \vec{\mathbbm{x}}_j^* = \vec{\mathbbm{x}}_{jj}^*$.
    \end{enumerate}
    
\end{definition}
The consistency condition requires all consumers' demands for consumer $j$'s consumption to coincide with consumer $j$'s actual consumption.

\begin{proposition}[Lindahl equilibrium]\label{prop.lindahl}
    If $|\phi^s|$ is small enough for all $s \in \mS$, a Lindahl equilibrium with interior allocation exists and can be characterized as follows.
    \begin{enumerate}
        \renewcommand{\labelenumi}{(\roman{enumi})}
        \item The Lindahl prices (up to scaling) are given by
        \begin{align*}
            \lp^{s*} = \frac{\alpha^{s}}{\bar{\omega}^{s}}, \qquad
            \lp^{s*}_{ii} = \widetilde{b}^s_i \lp^{s*}, \qquad
            \lp^{s*}_{ij} = \phi^s g^s_{ij}\lp^{s*}_{ii}, ~~ \mbox{ for } j\neq i.
        \end{align*}
        
        \item The consumption profile for good $s$ is given by
        \begin{equation}\label{eq.lindahlx}
            \vec{\mathbbm{x}}^{s*} = \bar{\omega}^{s} \bM^{s} 
            \left\{ 
            (\widetilde{\mathbf{b}}^s)^{-1}
            \odot 
            \left(\sum_{t \in \mS} \alpha^t \etab^t \right) 
            \right\}.
        \end{equation}
        where $(\widetilde{\mathbf{b}}^s)^{-1} = (1/\tilde{b}^s_1,\ldots,1/\tilde{b}^s_n)^\intercal$.
    \end{enumerate}
\end{proposition}

Proposition~\ref{prop.lindahl} characterizes the interior Lindahl equilibrium in the presence of network spillovers. First, note that the relative prices of physical goods under the Lindahl scheme coincide with those in the competitive equilibrium without spillovers. In particular, the ratio $\lp^{s*}/\lp^{t*}$ is identical to the no‐externality benchmark. Hence, these relative prices depend only on aggregate endowments and not on their distribution across consumers.

The Lindahl prices internalize externalities in two ways. First, the own price, $\lp^{s*}_{ii} = \widetilde{b}^s_i \lp^{s*}$, is adjusted by the consumer’s centrality, so that more central consumers face higher prices for their own consumption. Second, the price attached to spillovers generated by others, $\lp^{s*}_{ij} = \phi^s g^s_{ij}\lp^{s*}_{ii}$, reflects the intensity of the network link. When $\phi^s > 0$, individual~$i$ is charged for consuming positive spillovers generated by $j$; when $\phi^s<0$, she is subsidized to offset negative externalities. It is therefore natural to expect the Lindahl equilibrium to be efficient. The next theorem confirms this intuition.

\begin{theorem}[Lindahl equilibrium and efficiency]
\label{thm.lindahl} ~ 
\begin{enumerate}
    \renewcommand{\labelenumi}{(\roman{enumi})}
    
    \item The Lindahl equilibrium characterized in Proposition \ref{prop.lindahl} is Pareto efficient. Conversely, any interior Pareto-efficient allocation can be decentralized as a Lindahl equilibrium under an appropriate redistribution of initial endowments.
    
    \item Although Pareto efficient, the Lindahl equilibrium characterized in Proposition \ref{prop.lindahl} need not constitute a Pareto improvement relative to the competitive equilibrium.
    
\end{enumerate}
\end{theorem}

Since interior efficient allocations and Lindahl equilibrium allocations are characterized by \eqref{eq.efficientx} and \eqref{eq.lindahlx}, respectively, it is straightforward to verify the equivalence between these two allocations. In particular, one can choose suitable endowments so that $\sum_{s \in \mS}\alpha^s\etab^s$ coincides with a given Pareto weight vector $\btheta$, and conversely choose Pareto weights corresponding to any given vector $\sum_{s \in \mS}\alpha^s\etab^s$.


Although the Lindahl equilibrium restores efficiency, it need not be fair, in the sense that some consumers may be worse off than under the competitive equilibrium. The following example illustrates this point.

\paragraph*{Example VIII}
We revisit the three examples in Section~\ref{sec_Examples} under Lindahl pricing.\footnote{Throughout, we set $\phi = 0.4$, which ensures that the Lindahl equilibrium is interior.}
In Figure~\ref{fig.lbox} (Example~I), the centrality parallel condition does not hold. As a result, the equilibrium allocation (black dot) is inefficient: it lies on the purple line (the equilibrium locus) rather than on the green line (the contract curve). The Lindahl equilibrium (yellow dot) restores efficiency; moreover, it  makes consumer~2 better off while rendering consumer~1 worse off (Theorem \ref{thm.lindahl}(ii)). When the centrality parallel condition holds, the competitive equilibrium is Pareto efficient, as illustrated in Figures~\ref{fig.lbox2} (Example~II) and \ref{fig.lbox3} (Example~III). Although the Lindahl equilibrium is also Pareto efficient, it may yield an allocation that makes some consumers worse off compared to the competitive equilibrium outcome. For example, in Figures~\ref{fig.lbox2}  and \ref{fig.lbox3}, consumer 1 is worse off.\footnote{In Online Appendix \ref{App_Lindahl}, we show that the Lindahl equilibrium leads to agents being worse off in the \cite{ghiglinoKeepingNeighborsSocial2010}'s model with one private good and one conspicuous  network good ($\phi^s<0$). }
    
    \begin{figure}[ht]
        \begin{subfigure}{0.32\textwidth}
            \centering
            \includegraphics[width=1\textwidth]{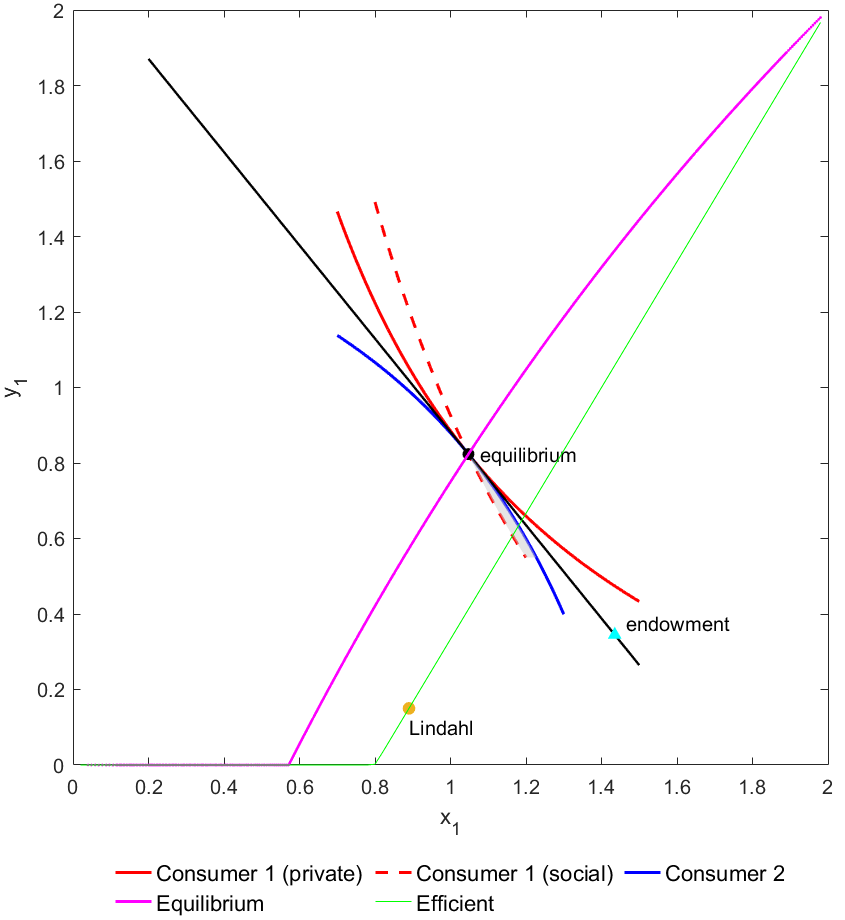}
            \caption{Asymmetric case}
            \label{fig.lbox}
        \end{subfigure}
        \hfill
        \begin{subfigure}{0.32\textwidth}
            \centering
            \includegraphics[width=1\textwidth]{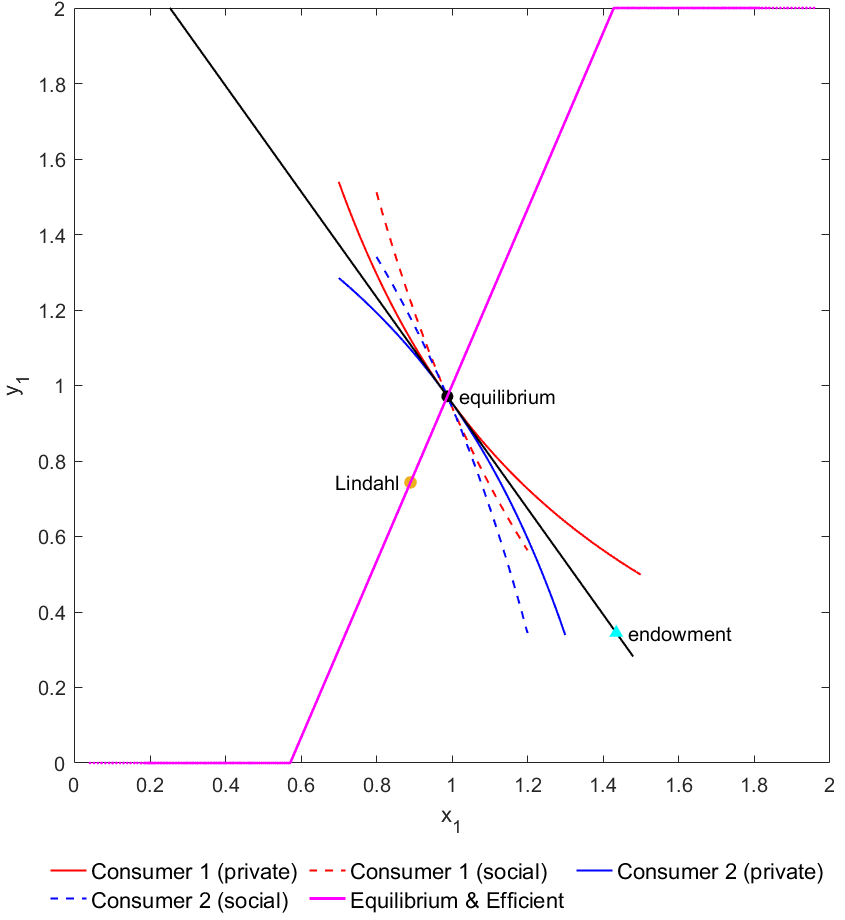}
            \caption{Symmetric consumers}
            \label{fig.lbox2}
        \end{subfigure}
        \hfill
        \begin{subfigure}{0.32\textwidth}
            \centering
            \includegraphics[width=1\textwidth]{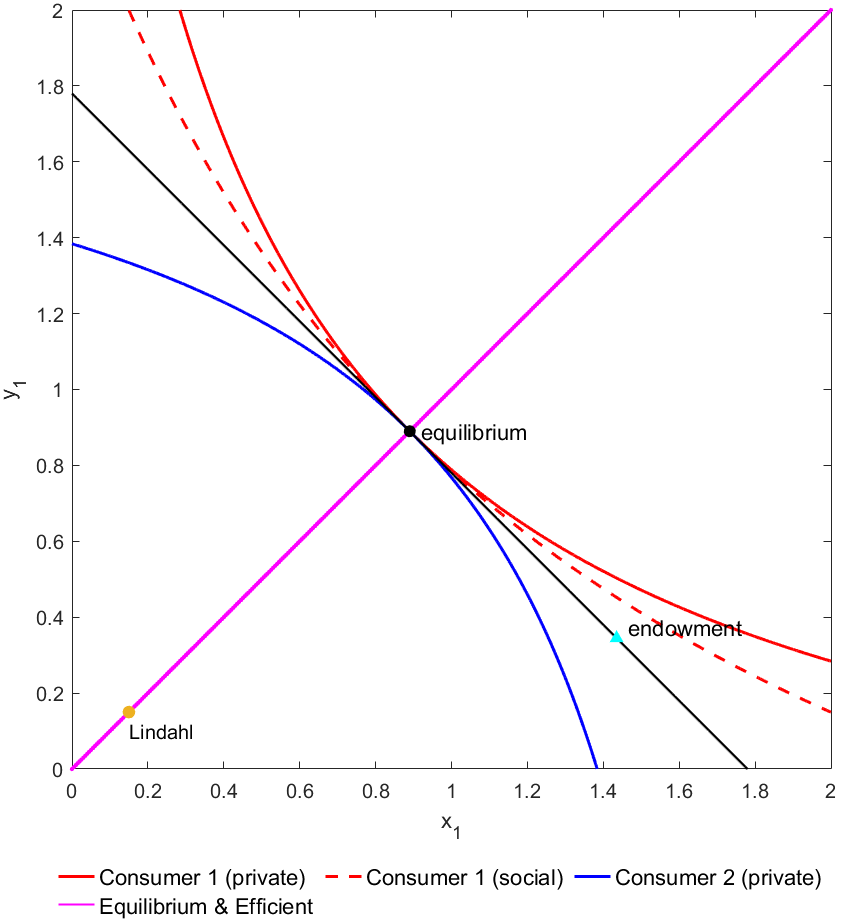}
            \caption{Symmetric goods}
            \label{fig.lbox3}
        \end{subfigure}
        \caption{Edgeworth boxes for Examples I-III with Lindahl equilibrium.}
        \label{fig.lindahl}
    \end{figure}

\medskip


\begin{remark}\label{remark.lindahl}
    \cite{elliottNetworkApproachPublic2019} study a one-dimensional public-good environment in which agents choose costly contribution levels that generate heterogeneous benefits for others. Our competitive-equilibrium framework differs from their cost-based setting in that agents trade physical goods in ordinary markets, while the externalities generated by consumption remain unpriced.

    The basic insight of Lindahl equilibrium is common to both settings: personalized prices for externalities should reflect the marginal benefits received by each consumer. The outcome, however, differs substantially. In their setting, Lindahl prices determine payments for agents' effort contributions. In our exchange economy, there are additional market prices for physical goods. This distinction leads to a different characterization of Lindahl allocations and prices.

    Our utility specification also allows us to study the distributional effect of the Lindahl equilibrium. In particular, although the Lindahl equilibrium restores efficiency, it need not constitute a Pareto improvement relative to the original competitive equilibrium: some consumers may gain, while others may be made worse off.
\end{remark}

\vspace{-.9em}

\section{Conclusion}\label{sec_conclusion}

In this paper, we incorporate multiplex network externalities into the Arrow--Debreu general equilibrium framework. Beyond establishing the existence and uniqueness of equilibrium, we characterize market outcomes and quantify the resulting efficiency losses using both a utility-based measure and a coefficient of resource utilization. This approach allows us to derive network-based conditions for the (in)efficiency of competitive equilibria.

Our analysis delivers three main insights. First, the interaction between multiplex network structures and the distribution of endowments gives rise to a novel statistic, the effective endowment, which plays a central role in determining equilibrium allocations and prices. Second, we quantify the departure of competitive equilibria from Pareto efficiency and characterize the multiplex network structures under which these losses vanish. In particular, efficiency is restored when the relevant network-centrality vectors are parallel. These welfare results extend beyond our benchmark specification to general utility functions that allow for nonlinear aggregation of externalities, cross-good spillovers,   and income effects. Third, completing the missing markets for network externalities through personalized Lindahl prices restores Pareto efficiency. This efficiency gain, however, need not constitute a Pareto improvement over the competitive equilibrium: consumers who benefit sufficiently from unpriced network externalities may be worse off under the Lindahl allocation.




\begin{center}
    {\bf\huge Appendix: Proofs}
\end{center}


\bigskip

\begin{proof}[Proof of Lemma \ref{lemma.exist}]
It is clear that each $\mathcal{X}_i$ and $\mathcal{P}$ are compact. Also, $u_i$ is continuous and quasi-concave in $\bx_i$. The correspondence $A_i$ is clearly convex and compact valued. Moreover, it has a closed graph since the mappings $(\bx_i,\bx_{-i},\bp) \mapsto \bp^\intercal (\bx_i-\w_i)$ and $(\bx_i,\bx_{-i},\bp) \mapsto x^{s}_i + \phi^s \sum_{j \in \mN} g^{s}_{ij} x^{s}_{j}$ are continuous, then the inequalities are preserved when taking limits. Next, we show $A_i$ is non-empty and continuous.\footnote{A correspondence $A(\mathbf{z})$ is continuous at $\mathbf{z}^{(0)}$ if for any $\hat{\mathbf{z}}^{(0)} \in A(\mathbf{z}^{(0)})$ and any sequence $(\mathbf{z}^{(k)})$ converging to $\mathbf{z}^{(0)}$, there exists a sequence $(\hat{\mathbf{z}}^{(k)})$ converging to $\hat{\mathbf{z}}^{(0)}$ such that for all $k$, $\hat{\mathbf{z}}^{(k)} \in A(\mathbf{z}^{(k)})$.} 

We first specify the value of $C$. If $\phi^s \geq 0$ for all $s \in \mS$, let $C=2$. Otherwise, let $C = [1+\min\{\frac{\min_{i \in \mN} \eta_i^s}{(n-1)|\phi^s|\bar{g}^s}: \phi^s <0\}]/2$, and by Assumption~\ref{assump.nonempty}, this is indeed larger than 1. For good $s$, if $\phi^s > 0$, then clearly $\w_i \in A_i(\bx_{-i},\bp)$ for any $(\bx_{-i},\bp)$. If $\phi^s <0$, we have 
\begin{equation}\label{appeq.nonempty}
    \begin{aligned}
        &\omega^{s}_i + \phi^s \sum_{j \in \mN} g^{s}_{ij} x^{s}_{j} 
        \geq
        \omega^{s}_i - |\phi^s|\bar{g}^s C(n-1)\sum_{k \in \mN}\omega_k^s\\
        =&  
        \left( \eta^s_i- |\phi^s|\bar{g}^s C(n-1) \right) \sum_{k \in \mN}\omega_k^s > 0,
    \end{aligned}
\end{equation}
where the first inequality follows from $x^s_j \leq C \sum_{k \in \mN} \omega^s_k$ and $g^s_{ij} \leq \bar{g}^s$. The last inequality is by the choice of $C$. Hence, we also have $\w_i \in A_i(\bx_{-i},\bp)$ for any $(\bx_{-i},\bp)$. Therefore, $A_i(\bx_{-i},\bp)$ is non-empty for any $(\bx_{-i},\bp)$.

Denote $f^s(\bx_{-i}) := \phi^s \sum_{j \neq i} g^s_{ij} x^s_j$. Take $\bx^{(0)}_i \in A_i(\bx^{(0)}_{-i},\bp^{(0)})$ and any sequence $(\bx^{(k)}_{-i}, \bp^{(k)})_{k \geq 1} \rightarrow (\bx^{(0)}_{-i},\bp^{(0)})$. According to equation (\ref{appeq.nonempty}), $\exists \delta>0$ such that $w_i^s-2\delta+f^s(\bx_{-i}) > 0$ for any $s$ and $\bx_{-i}$. Let $\bx^{(k)}_i := \lambda^{(k)} \bx^{(0)}_{i} + (1-\lambda^{(k)}) (\w_i-\delta\mathbf{1})$ with 
\begin{equation*}
    \lambda^{(k)} := \min
    \left\{1,
    \frac{\delta}{|(\bp^{(k)})^\intercal \bx^{(0)}_{i} - (\bp^{(k)})^\intercal \w_i + \delta|},
    \left( \frac{\delta}{\delta - x^{(0)s}_{i} - f^s(\bx^{(k)}_{-i}) } \right)_{s \in \hat{\mS}}
    \right\},
\end{equation*}
where $\hat{\mS} := \{s \in \mS: x^{(0)s}_i+f^s(\bx^{(0)}_{-i}) = 0 \}$. For $s \in \hat{\mS}$, $\lim_{k \rightarrow \infty} x^{(0)s}_i+f^s(\bx^{(k)}_{-i}) = x^{(0)s}_i+f^s(\bx^{(0)}_{-i}) = 0$. Then for sufficiently large $k$, we have $\lambda^{(k)} \in [0,1]$, and by convexity, $\bx^{(k)}_i \in \mathcal{X}_i$. Next, we show $\bx^{(k)}_i \in A_i(\bx^{(k)}_{-i},\bp^{(k)})$. 

The first inequality in the definition of $A_i$, $(\bp^{(k)})^\intercal \bx^{(k)}_i \leq (\bp^{(k)})^\intercal \w_i$, is equivalent to
\begin{equation}\label{appeq.budget}
    \lambda^{(k)} [(\bp^{(k)})^\intercal \bx^{(0)}_{i} - (\bp^{(k)})^\intercal \w_i + \delta] \leq \delta.
\end{equation}
which holds by the definition of $\lambda^{(k)}$. The second inequality is given by:
\begin{gather*}
    x^{(k)s}_i+ f^s(\bx^{(k)}_{-i}) 
    =
    \lambda^{(k)} (x^{(0)s}_{i} + f^s(\bx^{(k)}_{-i})) + (1-\lambda^{(k)}) (\omega^s_i-\delta + f^s(\bx^{(k)}_{-i}))
    \geq 0.
\end{gather*}
Since $\delta$ satisfies $\omega^s_i-\delta + f^s(\bx^{(k)}_{-i}) > \delta$, it suffices to have
\begin{equation}\label{appeq.nonnegative}
    \lambda^{(k)} (x^{(0)s}_{i} + f^s(\bx^{(k)}_{-i})) + (1-\lambda^{(k)}) \delta \geq 0
    \iff
    \delta \geq \lambda^{(k)}[\delta - x^{(0)s}_{i} - f^s(\bx^{(k)}_{-i})]
\end{equation}
For $s \in \hat{\mS}$, this inequality holds by the definition of $\lambda^{(k)}$. For $s \not\in \hat{\mS}$, $\lim_{k \rightarrow \infty} x^{(0)s}_i+f^s(\bx^{(k)}_{-i}) = x^{(0)s}_i+f^s(\bx^{(0)}_{-i}) > 0$. Then $x^{(0)s}_i+f^s(\bx^{(k)}_{-i}) > 0$ for sufficiently large $k$ and thus \eqref{appeq.nonnegative} holds. Therefore, we have $\bx^{(k)}_i \in A_i(\bx^{(k)}_{-i},\bp^{(k)})$. 

Moreover, since
\begin{align*}
    &\lim_{k \rightarrow \infty}(\bp^{(k)})^\intercal \bx^{(0)}_{i} - (\bp^{(k)})^\intercal \w_i + \delta = (\bp^{(0)})^\intercal \bx^{(0)}_{i} - (\bp^{(0)})^\intercal \w_i + \delta  \leq \delta,\\
    &\lim_{k \rightarrow \infty} \delta - x^{(0)s}_{i} - f^s(\bx^{(k)}_{-i}) 
    = \delta - x^{(0)s}_{i} - f^s(\bx^{(0)}_{-i})=\delta, \ s \in \hat{\mS}.
\end{align*}
Hence, we have $\lim_{k \rightarrow \infty} \lambda^{(k)} = 1$,\footnote{We can assume the first limit to be positive without loss of generality. Otherwise, \eqref{appeq.budget} holds automatically for large enough $k$ and we can drop the second term in the construction of $\lambda^{(k)}$.} that is, $\bx^{(k)}_i \rightarrow \bx^{(0)}_i$. Therefore, $A_i$ is continuous.

With all these conditions on $\mathcal{X}_i$, $\mathcal{P}$, $A_i$, $A_p$\footnote{$A_p$ is non-empty, convex, compact valued with a closed graph and continuous as it is constant.} and $u_i(\cdot)$, \cite{debreuSocialEquilibriumExistence1952} establishes the existence of an equilibrium for the abstract economy.
\end{proof}


\medskip

\begin{proof}[Proof of Proposition \ref{prop.exist}]
\noindent\textbf{Existence:}  As stated in Lemma~\ref{lemma.exist}, $\bx^*_i \in A_i(\bx^*_{-i},\bp^*)$ and is the maximizer of $u_i$, which is strongly increasing in $\bx_i$. Then we have $(\bp^*)^\intercal (\bx^*_i - \w_i) = 0$, which sums up to $(\bp^*)^\intercal (\sum_{i \in \mN} \bx^*_i - \sum_{i \in N} \w_i) = 0$. Moreover, by Lemma~\ref{lemma.exist}, the optimality of $\bp^*$ implies
\begin{equation}
    0 \leq p^{s*} \perp \left(\sum_{i \in \mN} x^{s*}_i - \sum_{i \in N} \omega^s_i \right) \leq 0.
    \label{appeq.priceplayer}
\end{equation}
Also, if $p^{s*} = 0$ for some $s$, then by utility maximization, $x^{s*}_i = C \sum_{i \in \mN} \omega^s_i$, which implies $\sum_{i \in \mN} x^{s*}_i = nC \sum_{i \in \mN} \omega^s_i >\sum_{i \in N} \omega^s_i $, a contradiction with \eqref{appeq.priceplayer}. Then we have $\sum_{i \in \mN} \bx^*_i - \sum_{i \in N} \w_i = 0$, i.e. $\bX^* \in \mathcal{A}$ is attainable. 

Next, suppose there exists $\tilde{\bx}_i \in R^{|\mS|}_+$ such that $\bp^\intercal \tilde{\bx}_i  \leq \bp^\intercal \w_i$ and $u_i(\tilde{\bx}_i, \bx^*_{-i}) > u_i(\bx^*_i,\bx^*_{-i})$, then by concavity $u_i(t\tilde{\bx}_i+(1-t)\bx^*_i,\bx^*_{-i}) > u_i(\bx^*_i, \bx^*_{-i})$ for $t \in (0,1)$. However, as $\bx^*_i < C \sum_{k \in \mN} \w_k$, then $t\tilde{\bx}_i+(1-t)\bx^*_i \in \mathcal{X}_i$ when $t$ is small enough. Moreover, $t\tilde{\bx}_i+(1-t)\bx^*_i \in A_i(\bx^*_{-i},\bp^*)$ since the set is convex. This contradicts with $\bx^*_i \in \arg\max_{\bx_i \in A_i(\bx^*_{-i},\bp^*)} u_i(\bx_i, \bx^*_{-i})$. Therefore, $\bx^*_i$ solves the consumer's optimization problem in $R^{|\mS|}_+$. In summary, $(\bX^*, \bp^*)$ is indeed a competitive equilibrium.  

\textbf{Interiority:}   First, for any equilibrium, the following inequality holds
\begin{equation}\label{appeq.opcon}
	\frac{\alpha^s}{x^{s}_i + \phi^s \sum_{j \in \mN} g^{s}_{ij} x^{s}_{j}} - \lambda_i p^{s} \leq 0 \ (\text{with equality if} \ x^s_i >0).
\end{equation} 
Suppose there is a non-interior equilibrium $(\bX, \bp)$ with $x^s_i = 0$, then $\phi^s > 0$. For good $s$, there exists a consumer $j$ such that $x^{s}_j \geq \bar{w}^s/n$ and thus $q^s_j \geq \bar{w}^s/n$. Furthermore, the equality in equation (\ref{appeq.opcon}) holds and we have $p^s = \frac{\alpha^s}{\lambda_j q^s_j}$.

For consumer $i$, $\exists s' \in \mS$ such that $x^{s'}_i \geq w^{s'}_i$. Then $\alpha^{s'}/q_i^{s'} = \lambda_i p^{s'}$. Using \eqref{appeq.opcon} for $j$ and $s'$, we have $\frac{\lambda_i}{\lambda_j} \leq \frac{q^{s'}_j}{q^{s'}_i} <
	\frac{\bar{w}^{s'}}{w^{s'}_i - |\phi^{s'}| \bar{g}^{s'} \bar{w}^{s'} }$, 
where the last inequality follows from Assumption~\ref{assump.interior}. Then we have
\begin{equation*}
    \frac{\alpha^s}{\phi^s\bar{g}^s\bar{w}^s} \leq
    \frac{\alpha^s}{q^s_i} \leq \lambda_i p^s =
    \frac{\lambda_i}{\lambda_j} \frac{\alpha^s}{q^s_j} <
    \frac{\bar{w}^{s'}}{w^{s'}_i - |\phi^{s'}| \bar{g}^{s'} \bar{w}^{s'} } \frac{n\alpha^s}{\bar{w}^s}
    =
    \frac{1}{\eta^{s'}_i - |\phi^{s'}| \bar{g}^{s'}} \frac{n\alpha^s}{\bar{w}^s}.
\end{equation*}
However, under Assumption~\ref{assump.interior}, we have
\begin{align*}
    \frac{\alpha^s}{\phi^s\bar{g}^s\bar{w}^s} > 
    \frac{\alpha^s(n+1) \bar{g}}{\ubar{\eta}\bar{g}^s\bar{w}^s} \geq \frac{\alpha^s(n+1)}{\ubar{\eta}\bar{w}^s} 
    = \frac{1}{\ubar{\eta} - \ubar{\eta}/(n+1) } \frac{n\alpha^s}{\bar{w}^s}
    >
    \frac{1}{\eta^{s'}_i - |\phi^{s'}| \bar{g}^{s'}} \frac{n\alpha^s}{\bar{w}^s},
\end{align*}
a contradiction. Therefore, the equilibrium must be interior. 

\textbf{Uniqueness:} The proof of uniqueness relies on the characterization in Theorem~\ref{thm.interior}. It is enough to show that $\operatorname{rank}(\overline{\bM})=n-1$. In that case, all solutions to \eqref{eq.mudetermine} are scalar multiples of one another. Since the induced consumption allocation is invariant to such rescalings, the interior equilibrium allocation is unique.

Let $\bM_0 := \mathbf I_n-(\sum_{s\in\mS}\alpha^s\boldsymbol{\eta}^s)\mathbf 1^\intercal$. We decompose $\overline{\bM} = \bM_0 + \bM_\Delta$, where $\bM_{\Delta} = \sum_{s\in\mS} \alpha^s (\mathbf I_n-\boldsymbol{\eta}^s\mathbf{1}^\intercal) (\bM^s-\mathbf I_n)$. Under Assumption~\ref{assump.interior}, $\| \phi^s \bG^s \|_{\infty} \leq n|\phi^s| \bar{g} < \frac{n\ubar{\eta}}{(n+1)} \leq \frac{1}{n+1}$. Recall that $\bM^s = \sum_{k=0}^{\infty} (-\phi^s \bG^s)^k$, then we obtain
\begin{equation*}
    \| \bM^s-\mathbf I_n \|_{\infty} \leq \sum_{k=1}^{\infty} \| \phi^s \bG^s \|_{\infty}^k 
    \leq \sum_{k=1}^{\infty} \frac{1}{(n+1)^k} 
    =\frac{1}{n}.
\end{equation*}
Since $\| \mathbf I_n-\boldsymbol{\eta}^s\mathbf{1}^\intercal  \|_{\infty} = \max_{i \in \mN}\{1+(n-2)\eta^s_i\} \leq n-1$, we have 
$\| \bM_\Delta \|_{\infty} \leq (n-1)/n$.

Consider the hyperplane $\mathcal{H} := \{\mathbf{z} \in \mathbb{R}^n: \mathbf{1}^\intercal \mathbf{z}=0\}$. Note that $\bM_0 \mathbf{z} = \mathbf{z}$, then for any $\mathbf{z} \in \operatorname{Ker}(\overline{\bM}) \cap \mathcal{H}$, $\mathbf{0} = \overline{\bM} \mathbf{z} = \mathbf{z} + \bM_\Delta \mathbf{z}$, which implies
\begin{equation*}
    0 = \| \mathbf{z} + \bM_\Delta \mathbf{z} \|_{\infty} \geq 
    \| \mathbf{z}  \|_{\infty} -\| \bM_\Delta \mathbf{z} \|_{\infty}
    \geq (1-\| \bM_\Delta \|_{\infty}) \| \mathbf{z}  \|_{\infty}
    \geq \| \mathbf{z}  \|_{\infty}/n.
\end{equation*}
This implies $\mathbf{z} = \mathbf{0}$ and thus $\operatorname{Ker}(\overline{\bM}) \cap \mathcal{H} = \{\mathbf{0}\}$. Since $\operatorname{dim}(\mathcal{H})=n-1$, then $\operatorname{rank}(\overline{\bM}) \geq n-1$. On the other hand, $\mathbf{1}^\intercal \overline{\bM} = \mathbf{0}^\intercal$ since $\mathbf{1}^\intercal \boldsymbol{\eta}^s = 1$. Hence, $\operatorname{rank}(\overline{\bM}) \leq n-1$. Therefore, $\operatorname{rank}(\overline{\bM}) = n-1$.
\end{proof}


\medskip

\begin{proof}[Proof of Theorem \ref{thm.interior}]
The derivation of $p^{s*}, \bx^{s*}$ and $\se$ is the same as in the main text. Moreover, the effective consumption equals $\mathbf{q}^{s*} = (\bM^s)^{-1} \bx^{s*} = \bar{\omega}^s \frac{\se}{b^s( \se)}$. Plug this into the utility function we obtain the equilibrium utility.
\end{proof}

\medskip

\begin{proof}[Proof of Lemma \ref{lemma.efficienttoplanner}]
Suppose that $\hat{\bX}$ solves \eqref{eq.pareto} for some
$\btheta\in\mathbb R^n_{++}$. If $\hat{\bX}$ is not Pareto efficient, then there exists $\bX \in \mathcal A$ such that $u_i(\bX)\geq u_i(\hat{\bX}), \forall i\in\mathcal N$, with strict inequality for at least one consumer. Since each entry of $\btheta$ is strictly positive, we have $\btheta^\intercal \mathbf{u}(\bX) > \btheta^\intercal \mathbf{u}(\hat{\bX})$, contradicting the optimality of $\hat{\bX}$ in \eqref{eq.pareto}. Hence
$\hat{\bX}$ must be Pareto efficient.

For the ``only if'' part, define set $\mathcal{U}:=\left\{\mathbf{v} \in\mathbb{R}^n: \exists \bX\in\mathcal A, v_i \leq u_i(\bX),\ \forall i\in\mathcal N \right\}$. Take any $\mathbf{v}^{(1)}, \mathbf{v}^{(2)} \in \mathcal{U}$, there exist $\bX^{(1)}, \bX^{(2)} \in\mathcal{A}$ such that $v^{(1)}_i \leq u_i(\bX^{(1)})$ and $v^{(2)}_i \leq u_i(\bX^{(2)})$ for all $i \in \mN$. For any $\lambda\in[0,1]$, convexity of $\mathcal A$ implies
$\lambda \bX^{(1)}+(1-\lambda)\bX^{(2)}\in\mathcal A$, and concavity of $u_i$ gives
\begin{equation*}
    u_i\big( \lambda \bX^{(1)}+(1-\lambda)\bX^{(2)} \big)
    \geq
    \lambda u_i(\bX^{(1)})+(1-\lambda)u_i(\bX^{(2)})
    \geq
    \lambda v^{(1)}_i+(1-\lambda)v^{(2)}_i .
\end{equation*}
Hence $\lambda \mathbf{v}^{(1)}+(1-\lambda)\mathbf{v}^{(2)} \in\mathcal U$, so $\mathcal U$ is convex.

If $\hat{\bX}$ is Pareto efficient, then $\mathbf{u}(\hat{\bX})$ is on the boundary of
$\mathcal U$. Therefore, by the supporting hyperplane theorem, there exists a nonzero vector $\btheta\in\mathbb R^n$ such that $\btheta^\intercal \mathbf{u} \leq \btheta^\intercal \mathbf{u}(\hat{\bX})$ for all $\mathbf{u} \in \mathcal U$. Moreover, the supporting vector must satisfy $\btheta\in\mathbb R^n_+$. Otherwise, if $\theta_j<0$ for some $j$, decreasing the $j$-th component of any feasible utility vector would increase $\btheta^\intercal \mathbf{u}$, contradicting the supporting inequality.

Now, for every $\bX\in\mathcal A$, the utility vector
$\big(u_i(\bX)\big)_{i\in\mathcal N}$ belongs to $\mathcal U$. Hence
\[
    \sum_{i\in\mathcal N}\theta_i u_i(\bX)
    \leq
    \sum_{i\in\mathcal N}\theta_i u_i(\hat{\bX}),
    \qquad \forall \bX\in\mathcal A.
\]
Therefore, $\hat{\bX}$ solves the weighted planner's problem \eqref{eq.pareto}.

Finally, we show that $\btheta\in\mathbb R^n_{++}$. Since $\hat{\bX}$ is an interior solution, then it must satisfy the first order condition \eqref{eq.plannerfoc}, which implies $\theta_i > 0$ as $\beta^s,\widetilde{b}^s_i>0$.

Since the planner’s problem \eqref{eq.pareto} is a concave maximization problem with affine equality constraints, and the feasible set is nonempty, strong duality holds. Hence, KKT conditions are necessary and sufficient for optimality. In particular, for an interior solution, the KKT conditions reduce to the first-order conditions. Write \eqref{eq.plannerfoc} in matrix form, we obtain
\begin{equation*}
    \mathbf{q}^s = \frac{\alpha^s}{\beta^s} (\widetilde{\mathbf{b}^s})^{-1} \odot \btheta
    \implies
    \bx^s = \frac{\alpha^s}{\beta^s} \bM^s \left\{(\widetilde{\mathbf{b}^s})^{-1} \odot \btheta \right\}.
\end{equation*}
Using the market clearing condition $\mathbf{1}^\intercal \bx^s = \bar{\omega}^s$, we obtain $\beta^s = \alpha^s \mathbf{1}^\intercal \btheta/\bar{\omega}^s$ and \eqref{eq.efficientx}.
\end{proof}


\medskip

\begin{proof}[Proof of Proposition \ref{prop.efficiencyloss}]
(i) By Theorem~\ref{thm.interior} and Lemma~\ref{lemma.efficienttoplanner}, we can calculate
$
u_i(\hat{\bX}(\btheta)) - u^*_i 
= \sum_{s \in \mS} \alpha^s  
\left(
\ln \theta_i - \ln \rho^s_i
\right)
$.
Taking the $\boldsymbol{\theta}$-weighted sum across consumers yields
\begin{equation*}
    L(\btheta) = \sum_{i \in \mN} \theta_i \sum_{s \in \mS}  \alpha^s  
    \left(
    \ln \theta_i - \ln \rho^s_i
    \right)
    =
    \sum_{s \in \mS}  \alpha^s  
    D_{KL} \left( \btheta \,\|\, \boldsymbol{\rho}^s \right),
\end{equation*}

(ii) We first show that every feasible $(\bX,\gamma)$ of \eqref{eq.efficiencycoef} satisfies $\gamma \geq  \sum_{i \in \mN} \prod_{s \in \mS}  (\rho^s_i)^{\alpha^s}$. Using \eqref{eq.equiu} and the definition of $\brho^s$, the constraint $u_i(\bX) \geq u_i(\bX^*)$ implies
\begin{equation}\label{appeq.feasibleu}
    \sum_{s \in \mS} \alpha^s \ln\left( \frac{\tilde{b}^s_i q^s_i}{\bar{\omega}^s} \right)
    \geq 
    \sum_{s \in \mS} \alpha^s \ln \rho^s_i
    \implies
    \prod_{s \in \mS} \left( \frac{\tilde{b}^s_i q^s_i}{\bar{\omega}^s} \right)^{\alpha^s }
    \geq 
    \prod_{s \in \mS}  (\rho^s_i)^{\alpha^s}.
\end{equation}
Summing over $i$ and applying the weighted AM-GM inequality gives
\begin{equation*}
    \sum_{i \in \mN} \prod_{s \in \mS}  (\rho^s_i)^{\alpha^s}
    \leq
    \sum_{i \in \mN} \prod_{s \in \mS} \left( \frac{\tilde{b}^s_i q^s_i}{\bar{\omega}^s} \right)^{\alpha^s }
    \leq
    \sum_{i \in \mN} \sum_{s \in \mS} \alpha^s \frac{\tilde{b}^s_i q^s_i}{\bar{\omega}^s} 
    \leq
    \sum_{s \in \mS} \alpha^s \gamma = \gamma,
\end{equation*}
where the last inequality follows from the feasibility $\sum_{i \in \mN} x^s_i = \sum_{i \in \mN} \tilde{b}^s_i q^s_i \leq \gamma \bar{\omega}^s$.

Next, we show that this lower bound is attainable. Let
\begin{equation*}
    q^s_i = \frac{\prod_{t \in \mS}  (\rho^t_i)^{\alpha^t}}{\tilde{b}^s_i} \bar{\omega}^s.
\end{equation*}
We assume externalities are sufficiently small such that $\bx^s = \bM^s \mathbf{q}^s \geq \mathbf{0}$. Moreover, $q^s_i$ satisfies \eqref{appeq.feasibleu} and $\sum_{i \in \mN} x^s_i = \sum_{i \in \mN} \tilde{b}^s_i q^s_i = \bar{\omega}^s \sum_{i \in \mN} \prod_{t \in \mS}  (\rho^t_i)^{\alpha^t}$. Therefore, such construction is indeed feasible and the lower bound is attained.

Finally, $\sum_{s \in \mS} \alpha^s D_{KL} \left( \btheta \,\|\, \brho^s \right)$ is convex in $\btheta$, the minimizer is given by KKT conditions:
\begin{equation*}
    \ln \theta_i +1 - \sum_{s \in \mS} \alpha^s \ln \rho^s_i + \lambda=0, 
    \quad \mathbf{1}^\intercal \btheta = 1
    \implies
    \theta^{min}_i = \frac{\prod_{s \in \mS}  (\rho^s_i)^{\alpha^s}}{\sum_{k \in \mN} \prod_{s \in \mS}  (\rho^s_k)^{\alpha^s}}.
\end{equation*}
Plugging this back we get $\min_{\btheta \in \Delta}\sum_{s \in \mS} \alpha^s D_{KL} \left( \btheta \,\|\, \brho^s \right) = -\ln \gamma^*$.

Furthermore, for $\btheta$ satisfying the condition in part (i), we can write
\begin{align*}
    L(\btheta) =& \sum_{i \in \mN} \theta_i 
    \left( \ln \theta_i-\sum_{s\in \mS} \alpha^s \ln \rho^s_i \right)\\
    = & \sum_{i \in \mN} \theta_i 
    \left( \ln \theta_i- \ln (\theta_i^{min} \times\mathcal{CRU}) \right)
    = D_{KL}(\btheta \, \lVert \, \btheta^{\min})-\ln \mathcal{CRU}.
\end{align*}

The standard result \citep{tsybakov2008introduction} gives $D_{KL} \left( \btheta \,\|\, \boldsymbol{\rho}^s \right) \geq H(\btheta,\brho^s)^2$. Denote $\sqrt{\btheta}$ and $\sqrt{\brho^s}$ the entry-wise square root vector. Then we have 
\begin{align*}
    &\sum_{s \in \mS}  \alpha^s  
    D_{KL} \left( \btheta \,\|\, \boldsymbol{\rho}^s \right)
    \geq \sum_{s\in \mS} \alpha^s H(\btheta,\brho^s)^2 
    =
    \sum_{s\in \mS} \alpha^s \| \sqrt{\btheta} - \sqrt{\brho^s} \|_2^2\\
    =&
    \| \sqrt{\btheta} - \sum_{s\in \mS} \alpha^s \sqrt{\brho^s} \|_2^2 + 
    \frac{1}{2} \sum_{s,t \in \mS} \alpha^s \alpha^t \left\| \sqrt{\brho^s} - \sqrt{\brho^t} \right\|_2^2
    \geq
    \frac{1}{2} \sum_{s,t \in \mS} \alpha^s \alpha^t H(\brho^s,\brho^t)^2 \geq 0.
\end{align*}
On the other hand, note that $\brho^1 \in \Delta$,
\begin{equation*}
    \min_{\btheta \in \Delta} \sum_{s \in \mS}  \alpha^s  
    D_{KL} \left( \btheta \,\|\, \boldsymbol{\rho}^s \right)
    \leq 
    \sum_{s \in \mS}  \alpha^s  
    D_{KL} \left( \brho^1 \,\|\, \boldsymbol{\rho}^s \right)
    \leq
    \max_{s,t \in \mS} D_{KL} \left( \brho^s \,\|\, \brho^t \right).
\end{equation*}
Use $\min_{\btheta \in \Delta}\sum_{s \in \mS} \alpha^s D_{KL} \left( \btheta \,\|\, \brho^s \right) = -\ln \gamma^*$ gives the bounds for $\gamma^*$.
\end{proof}


\medskip

\begin{proof}[Proof of Theorem \ref{thm.efficiency}]
(i1) Recall that $\btheta^* = \mathbf{b} \odot \se$, where $\mathbf{b} \parallelsum \widetilde{\mathbf{b}^s}$. Therefore, $\frac{\btheta^*}{\mathbf{1}^\intercal\btheta^*} = \frac{\widetilde{\mathbf{b}}^s \odot \se}{(\widetilde{\mathbf{b}}^s)^\intercal \se}$. By Lemma \ref{lemma.efficienttoplanner}, the efficient allocation \eqref{eq.efficientx} becomes
\begin{equation*}
    \hat{\bx}^s = \bar{\omega}^s \frac{\bM^s \se}{(\widetilde{\mathbf{b}}^s)^\intercal \se}.
\end{equation*}
This is the equilibrium consumption \eqref{eq.equix}. Therefore, the equilibrium is efficient.

(i2) Consider an interior efficient allocation $\bX^P = \hat{\bX}$, it satisfies equation~\eqref{eq.plannerfoc}, which can be written as
\begin{equation*}
    \theta_i \frac{\alpha^s}{\hat{q}^s_i} = \beta^s \widetilde{b}^s_i = \beta^sb_i \widetilde{b}^s,
\end{equation*}
where the last equality follows from the parallel condition. Therefore, we can simply let $\hat{\w} = \hat{\bX}$, $\mu_i = \theta_i/b_i$, $p^s = \beta^s\widetilde{b}^s$, then $\hat{\bX}$ satisfies the equilibrium first-order condition \eqref{eq.foceqm} and the budget constraint. Then $\hat{\bX}$ solves the utility maximization problem for each individual $i \in \mN$. Moreover, it is feasible since it solves the planner's constrained problem. Therefore, $\hat{\bX}$ is indeed a competitive equilibrium under the previous endowment and price selection.

(ii) If the centrality parallel condition fails, suppose the First Welfare Theorem still holds. Then there is an allocation $\bX$ being both equilibrium and efficient allocation. Therefore, it satisfies planner FOC \eqref{eq.plannerfoc} and equilibrium FOC \eqref{eq.foceqm}. This implies
\begin{equation*}
    \frac{\widetilde{b}^s_i}{\widetilde{b}^t_i} 
    = \frac{p^s \beta^t}{p^t\beta^s}
    = \frac{\widetilde{b}^s_j}{\widetilde{b}^t_j},
    \quad \forall i,j \in \mN, s,t \in \mS.
\end{equation*}
However, this is exactly the parallel condition, a contradiction. Therefore, the First Welfare Theorem fails. Then the Second Welfare Theorem naturally fails since no interior competitive equilibrium can be efficient.
\end{proof}


\medskip

\begin{proof}[Proof of Corollary \ref{cor.notparallel}]
We show the existence of the Pareto improving allocation by construction. Suppose $\widetilde{\mathbf{b}}^s$ is not proportional to $\widetilde{\mathbf{b}}^t$, then we construct a redistribution of the consumption for good $s$ and good $t$. The consumption for all other goods remains the same. Then the new allocation $\bX = (\ldots,\bx^{s*}+l\btau^s,\ldots,\bx^{t*}+l\btau^t,\ldots)$ is attainable. Moreover, after the reallocation, 
\begin{align*}
    q^s_i = q^{s*}_i + l 
    \underbrace{
    \left(\tau^s_i+\phi^s \sum_{j \in \mN} g^s_{ij} \tau^s_j
    \right)
    }_{=: \hat{\tau}^s_i}, \qquad
    q^t_i = q^{t*}_i + l
    \underbrace{
    \left(\tau^t_i+\phi^t \sum_{j \in \mN} g^t_{ij} \tau^t_j
    \right)
    }_{=: \hat{\tau}^t_i}.
\end{align*}
Next, we show that it Pareto dominates $\bX^*$ for sufficiently small $l$ by showing
\begin{equation*}
    \left. \frac{\partial u_i(\bX)}{\partial l} \right|_{l=0} 
    =
    \frac{\alpha^t q^{s*}_i \hat{\tau}^t_i+\alpha^sq^{t*}_i \hat{\tau}^s_i}{q^{t*}_iq^{s*}_i}
    \geq 1,
\end{equation*}
which is equivalent to the following systems
\begin{alignat*}{2}
    &\text{Original system}
    \qquad\qquad
    &&\text{Dual system}
    \\[0.6em]
    &\alpha^t \mathbf{Q}^{s*} \hat{\btau}^t
    + \alpha^s \mathbf{Q}^{t*} \hat{\btau}^s
    \geq \mathbf{q}^{s*} \odot \mathbf{q}^{t*},
    \qquad\qquad
    &&
    \alpha^t \mathbf{Q}^{s*} \mathbf{y}
    = m_1(\bM^t)^\intercal \mathbf{1}
    = m_1 \widetilde{\mathbf{b}}^t,
    \\
    &\mathbf{1}^\intercal \bM^t \hat{\btau}^t = 0,
    &&
    \alpha^s \mathbf{Q}^{t*} \mathbf{y}
    = m_2(\bM^s)^\intercal \mathbf{1}
    = m_2 \widetilde{\mathbf{b}}^s,
    \\
    &\mathbf{1}^\intercal \bM^s \hat{\btau}^s = 0,
    &&
    \mathbf{y}^\intercal \{\mathbf{q}^{s*} \odot \mathbf{q}^{t*}\} > 0.
\end{alignat*}
where $\mathbf{Q}^{s*} = diag(q^{s*}_1,\ldots, q^{s*}_n)$, $\mathbf{Q}^{t*} = diag(q^{t*}_1,\ldots, q^{t*}_n)$, $\mathbf{y} \in \mathbb{R}^n_+$ and $m_1, m_2 \in \mathbb{R}$. If the dual system holds, then
\begin{equation*}
    \frac{m_1 \widetilde{b}^t_i}{m_2 \widetilde{b}^s_i} = \frac{\alpha^t q^{s*}_i}{\alpha^s q^{t*}_i} = 
    \frac{p^{t*}}{p^{s*}}
\end{equation*}
where the last equality follows from equations (\ref{eq.equix}) and (\ref{eq.equip}). Note that this equation holds for all $i,j \in \mN$, which implies that $\widetilde{\mathbf{b}}^s \parallelsum \widetilde{\mathbf{b}}^t$, a contradiction. Hence, the dual system cannot hold. By Farkas’ lemma, the original system must therefore hold, and Pareto improving reallocations exist.
\end{proof}


\medskip

\begin{proof}[Proof of Proposition \ref{prop.comparative}]
By the definition of endowment change, we have
\begin{equation*}
    \etab^s(l)
    =
    \frac{\w^s+l\btau^s}{\bar{\omega}^s+l\mathbf{1}^\intercal\btau^s}
    \implies
    \dot{\etab^s} = 
    \left. \frac{\partial \etab^s(l)}{\partial l} \right|_{l=0}
    =
    \frac{\btau^s-\etab^s\mathbf{1}^\intercal\btau^s}{\bar{\omega}^s}.
\end{equation*}

By the definition of $\overline{\bM}$, we have
\[
\dot{\overline{\bM}}
=
\sum_{s \in \mS}
-\alpha^s
\dot{\etab}^s
\mathbf{1}^\intercal\bM^s
=
\sum_{s \in \mS}
-\frac{\alpha^s}{\bar{\omega}^s}
(\btau^s-\etab^s\mathbf{1}^\intercal\btau^s)\mathbf{1}^\intercal\bM^s.
\]
As we discussed in the main text, $\se(l)$ is the differentiable solution path of equation $\overline{\bM}(l)\se(l)=\mathbf{0}$, fixing $\se(0) \equiv \bmu^{*0}$. Differentiating at $l=0$ yields
\begin{equation}\label{eq.app-linearizedmu}
    \overline{\bM}\dot{\bmu}^*
    =
    -\dot{\overline{\bM}}\bmu^{*0}
    =
    \sum_{s \in \mS}
    \frac{\alpha^s}{\bar{\omega}^s}
    (\btau^s-\etab^s\mathbf{1}^\intercal\btau^s)\mathbf{1}^\intercal\bM^s \bmu^{*0}
    =
    \sum_{s \in \mS}
    p^{s*}(\btau^s-\etab^s\mathbf{1}^\intercal\btau^s).
\end{equation}
Since $\mathbf{1}^\intercal(\btau^s-\etab^s\mathbf{1}^\intercal\btau^s)=0$ for all $s \in \mS$ and $\mathbf{1}^\intercal\overline{\bM}=\mathbf{0}^\intercal$, the solution to \eqref{eq.app-linearizedmu} exists. 

We can now derive the comparative statics of relative prices. Using equation~\eqref{eq.equip},
\[
\ln\frac{p^{s*}}{p^{t*}}
=
\ln\alpha^s-\ln\alpha^t
-\ln\bar{\omega}^s+\ln\bar{\omega}^t
+\ln(b^s(\se))-\ln(b^t(\se)).
\]
Differentiating with respect to $l$, note that $b^s(\mathbf{z}) = \mathbf{1}^\intercal \bM^s \mathbf{z}$ is linear in $\mathbf{z}$, we obtain
\begin{equation}\label{eq.app-relativeprice}
    \begin{aligned}
        \dot{\left( \ln\frac{p^{s*}}{p^{t*}}\right)}
        =
        \frac{\dot{p}^{s*}}{p^{s*}}-\frac{\dot{p}^{t*}}{p^{t*}}
        =
        \frac{b^s(\dot{\bmu}^*)}{b^s(\bmu^{*0})}
        -
        \frac{b^t(\dot{\bmu}^*)}{b^t(\bmu^{*0})}
        -
        \frac{\mathbf{1}^\intercal\btau^s}{\bar{\omega}^s}
        +
        \frac{\mathbf{1}^\intercal\btau^t}{\bar{\omega}^t}.
    \end{aligned}
\end{equation}
Turning to welfare, differentiating equation~\eqref{eq.equiu} with respect to $l$ gives
\begin{equation}\label{eq.app-welfare}
    \begin{aligned}
        \dot{u}_i^*
        =&
        \frac{\dot{\mu}^*_i}{\mu^{*0}_i}
        -
        \sum_{s \in \mS}\alpha^s
        \frac{b^s(\dot{\bmu}^*)}{b^s(\bmu^{*0})}
        +
        \sum_{s \in \mS}
        \alpha^s\frac{\mathbf{1}^\intercal\btau^s}{\bar{\omega}^s}
    \end{aligned}
\end{equation}
The consumption change can also be derived as
\begin{equation}\label{eq.app-consumption}
    \dot{(\ln x^{s*}_i)} = \frac{b^s_i(\dot{\bmu}^*)}{b^s_i(\bmu^{*0})} 
    - \frac{b^s(\dot{\bmu}^*)}{b^s(\bmu^{*0})}
    +\frac{\mathbf{1}^\intercal\btau^s}{\bar{\omega}^s}
\end{equation}

The values of equations~\eqref{eq.app-relativeprice}-\eqref{eq.app-consumption} do not depend on the choice of $(\bmu, \dot\bmu)$. To see this, let $\bmu'$ be any solution of \eqref{eq.mudetermine}, and $\dot{\bmu}'$ be any solution to equation~(\ref{eq.app-linearizedmu}) given $\bmu'$. Recall that Assumption \ref{assump.interior} implies $\operatorname{rank}(\overline{\bM})=n-1$, $\bmu' = y\bmu^{*0} $ and $\dot{\bmu}'=y\dot{\bmu}^*+z\bmu^{*0}$ for some scalar $y,z \in \mathbb{R}$. Therefore
\[
\frac{\mathbf{1}^\intercal\bM^s\dot{\bmu}'}{\mathbf{1}^\intercal\bM^s\bmu'}
=
\frac{y\mathbf{1}^\intercal\bM^s\dot{\bmu}^*+z\mathbf{1}^\intercal\bM^s\bmu^{*0}}
{y\mathbf{1}^\intercal\bM^s\bmu^{*0}}
=
\frac{\mathbf{1}^\intercal\bM^s\dot{\bmu}^*}{\mathbf{1}^\intercal\bM^s\bmu^{*0}}
+
\frac{z}{y}.
\]
Hence the additive term $z/y$ cancels from~\eqref{eq.app-relativeprice}-\eqref{eq.app-consumption}. Therefore, the comparative statics results are independent of the solution selection of~$(\bmu, \dot\bmu)$.
\end{proof}


\medskip

\begin{proof}[Proof of Proposition \ref{prop.priceeffect}]
Recall that $\bH = \overline{\bM} + \alpha^1\boldsymbol{\eta}^1 \mathbf{1}^\intercal\bM^1$, then $\mathbf{1}^\intercal\bH = \alpha^1\mathbf{1}^\intercal\bM^1$. We first show that $\bH$ is invertible. Otherwise, there exists 
$\mathbf{z} \neq \mathbf{0}$ such that $\bH\mathbf{z}=\mathbf{0}$. Then we have $\alpha^1\mathbf{1}^\intercal\bM^1\mathbf{z} = \mathbf{1}^\intercal\bH\mathbf{z} = 0$ and $\mathbf{0} =\bH \mathbf{z} = \overline{\bM}\mathbf{z} + \alpha^1 \boldsymbol{\eta}^1 \mathbf{1}^{\intercal}\bM^1 \mathbf{z} =\overline{\bM}\mathbf{z}$. By Assumption~\ref{assump.interior}, $\operatorname{rank}(\overline{\bM})=n-1$, then $\operatorname{Ker}(\overline{\bM})$ is 
one-dimensional. Hence, there exists $c \neq 0$ such that $\mathbf{z}=c\se$. However, the equilibrium price $p^{1*} = \alpha^1\mathbf{1}^{\intercal} \bM^1 \se/\bar{\omega}^1 = \alpha^1\mathbf{1}^{\intercal} \bM^1 \mathbf{z}/(c\bar{\omega}^1) > 0$, a contradiction to $\alpha^1\mathbf{1}^\intercal\bM^1\mathbf{z} = 0$. Therefore, $\bH$ is invertible. Let $\bmu^{*0}=\bH^{-1}\w^1$, we have 
\begin{equation*}
    \overline{\bM} \bmu^{*0} = (\bH-\alpha^1\boldsymbol{\eta}^1\mathbf{1}^\intercal\bM^1)\bH^{-1} \w^1
    =\w^1 - \etab^1 \mathbf{1}^\intercal \bH \bH^{-1} \w^1 
    = \w^1-\w^1=\mathbf{0}.
\end{equation*}
Therefore, $\bmu^{*0}$ solves \eqref{eq.mudetermine} and is indeed the effective endowment. We next show that a particular solution to equation~\eqref{eq.dotmu} is $\dot{\bmu}^*=p^{t*}\bH^{-1}\btau$.
For any vector $\mathbf z$ satisfying $\mathbf{1}^\intercal\mathbf z=0$:
\[
\overline{\bM}\bH^{-1}\mathbf z
=
(\bH-\alpha^1\boldsymbol{\eta}^1\mathbf{1}^\intercal\bM^1)\bH^{-1}\mathbf z
=
\mathbf z-\boldsymbol{\eta}^1 \mathbf{1}^\intercal \bH\bH^{-1}\mathbf z
= \mathbf{z}.
\]
Therefore, it holds for $\btau$, which proves that $\dot{\bmu}^*=p^{t*}\bH^{-1}\btau$ is indeed a particular solution to~\eqref{eq.dotmu}. Substituting this into~\eqref{eq.priceeffect} and by $\alpha^1\mathbf{1}^\intercal\bM^1\bH^{-1}\btau= \mathbf{1}^\intercal\bH\bH^{-1}\btau =0$,
\begin{equation*}
    \frac{\dot{p}^{s*}}{p^{s*}} - \frac{\dot{p}^{1*}}{p^{1*}}
    =
    \frac{p^{t*}\mathbf{1}^\intercal\bM^s\bH^{-1}\btau}{\mathbf{1}^\intercal\bM^s\bmu^{*0}}   
    =  \frac{p^{t*}(\mathbf{c}^s)^\intercal\btau}{\mathbf{1}^\intercal\bM^s\bmu^{*0}}   
    =  \frac{p^{t*}\tau_j(c^s_j-c^s_i)}{\mathbf{1}^\intercal\bM^s\bmu^{*0}}.   
\end{equation*}
where the last equality follows from the definition of transfer on good $t$. The sign of relative price change is the same as $c^s_j-c^s_i$.
\end{proof}


\medskip

\begin{proof}[Proof of Corollary \ref{cor.infeasible}]
Denote $(\bmu^{*0})^{-1} = (1/\mu^{*0}_1,\ldots,1/\mu^{*0}_n)^\intercal$. By definition of $\boldsymbol{\varpi}$,
\begin{equation*}
    \boldsymbol{\varpi}^\intercal \{ (\bmu^{*0})^{-1} \odot\dot{\bmu}^*  \}
    =
    \sum_{s \in \mS}\alpha^s
    \frac{ (\dot{\bmu}^*)^\intercal \widetilde{\mathbf{b}}^s }{(\bmu^{*0})^\intercal \widetilde{\mathbf{b}}^s}.
\end{equation*}
For any pure redistribution, rewrite equation \eqref{eq.welfareeffect} in matrix form for $i \in \mN$, we obtain
\begin{equation*}
    \dot{\mathbf{u}}^* 
    =
    (\bmu^{*0})^{-1} \odot\dot{\bmu}^*
    -
    \mathbf{1} \sum_{s \in \mS}\alpha^s
    \frac{ (\dot{\bmu}^*)^\intercal \widetilde{\mathbf{b}}^s }{(\bmu^{*0})^\intercal \widetilde{\mathbf{b}}^s}
    =
    \left[ 
    \mathbf{I}_n- \mathbf{1}\boldsymbol{\varpi}^\intercal
    \right]
    \left\{
    (\bmu^{*0})^{-1} \odot \dot{\bmu}^*
    \right\}.
\end{equation*}
Since $\mathbf{1}^\intercal \boldsymbol{\varpi}=1$, we have $\boldsymbol{\varpi}^\intercal \dot{\mathbf{u}}^* = 0$.
\end{proof}


\medskip

\begin{proof}[Proof of Proposition \ref{prop.lindahl}] The first order conditions for consumer $i$ give
\begin{align*}
    \frac{\alpha^s}{\mathbbm{q}^s_i} - \lambda_i \lp^s_{ii} = 0,
    \quad
    \frac{\alpha^s \phi^s g^s_{ij}}{\mathbbm{q}^s_i} - \lambda_i \lp^s_{ij} = 0,
\end{align*}
which implies $\lp^{s}_{ij} = \phi^s g^s_{ij}\lp^{s}_{ii}$. The price compatibility condition then implies
\begin{equation*}
    \lp^s = \sum_{j \in \mN} \lp^s_{ji} = \lp^s_{ii}+ \sum_{j \neq i} \phi^s g^s_{ji} \lp^s_{jj} \implies
    [\mathbf{I}_n + \phi^s (\bG^s)^\intercal]
    (\lp^s_{11},\ldots,\lp^s_{nn})^\intercal
    =
    \lp^s \mathbf{1}.
\end{equation*}
This gives $\lp^{s}_{ii} = \widetilde{b}^s_i \lp^s$. Substituting this into the first FOC and denote $\gamma_i = 1/\lambda_i$, we obtain $\alpha^s \gamma_i = \lp^s \widetilde{b}^s_i  \mathbbm{q}^s_i$. Summing over $i$, we have
\begin{equation}\label{appeq.lindahlp}
    \alpha^s \mathbf{1}^\intercal \boldsymbol{\gamma} 
    = \lp^s \sum_{i \in \mN} \widetilde{b}^s_i  \mathbbm{q}^s_i
    = \lp^s \mathbf{1}^\intercal \bM^s \vec{\mathbbm{q}}^s
    = \lp^s \mathbf{1}^\intercal \vec{\mathbbm{x}}^s 
    = \lp^s \bar{\omega}^s.
\end{equation}
Substituting $\lp^{s}_{ii} = \widetilde{b}^s_i \lp^s$ and (\ref{appeq.lindahlp}) into the first FOC, we obtain
\begin{equation*}
    \mathbbm{q}^s_i = \frac{\alpha^s}{\widetilde{b}^s_i \lp^s} \gamma_i
    \implies
    [\mathbf{I}_n + \phi^s\bG^s] \vec{\mathbbm{x}}^s = \frac{\alpha^s}{\lp^s} (\widetilde{\mathbf{b}}^s)^{-1} \odot \boldsymbol{\gamma} 
    \implies
    \vec{\mathbbm{x}}^s = \bar{\omega}^s  \bM^s 
    \left\{(\widetilde{\mathbf{b}}^s)^{-1} \odot \frac{\boldsymbol{\gamma}}{\mathbf{1}^\intercal \boldsymbol{\gamma}} \right\},
\end{equation*}
where $(\widetilde{\mathbf{b}}^s)^{-1} = (1/\tilde{b}^s_1,\ldots,1/\tilde{b}^s_n)^\intercal$.

Finally, using the relationship between prices and the first FOC, we have
\begin{equation*}
    \sum_{j \in \mN}\lp^s_{ij} \mathbbm{x}^s_{ij} = \lp^s_{ii} \left(\mathbbm{x}^{s}_{ii} + \phi^s \sum_{j\neq i} g^s_{ij}\mathbbm{x}^s_{ij}
    \right) = \lp^s_{ii} \mathbbm{q}^s_i = \alpha^s \gamma_i.
\end{equation*}
By the binding budget and \eqref{appeq.lindahlp}, we obtain
\begin{equation*}
    \gamma_i = \sum_{s \in \mS}\sum_{j \in \mN}\lp^s_{ij} \mathbbm{x}^s_{ij}
    = \sum_{s \in \mS}\lp^s \omega^s_i = \mathbf{1}^\intercal \boldsymbol{\gamma} \sum_{s \in \mS} \alpha^s \eta^s_i
    \implies
    \frac{\boldsymbol{\gamma}}{\mathbf{1}^\intercal \boldsymbol{\gamma}} = \sum_{s \in \mS} \alpha^s \etab^s.
\end{equation*}
This pins down $\boldsymbol{\gamma}$ and proves Proposition~\ref{prop.lindahl}.
\end{proof}


\medskip

\begin{proof}[Proof of Theorem~\ref{thm.lindahl}]
\noindent(i) By Lemma~\ref{lemma.efficienttoplanner} and Proposition~\ref{prop.lindahl}, the Lindahl equilibrium allocation solves planner's problem with $\btheta = \sum_{s \in \mS} \alpha^s \etab^s$. Therefore, it is Pareto efficient.

For any interior efficient allocation $\hat{\bX}$, by Lemma~\ref{lemma.efficienttoplanner}, there exists a Pareto weight vector $\btheta \in \mathbb{R}^n_{++}$ such that $\hat{\bX}$ can be expressed as in equation \eqref{eq.efficientx}. We can choose the initial endowment such that $\etab^s \parallelsum \btheta$. Then $\sum_{s \in \mS} \alpha^s \etab^s \parallelsum \btheta$. According to Proposition~\ref{prop.lindahl}, the Lindahl equilibrium consumption $\vec{\mathbbm{x}}^{s*}$ is parallel to $\hat{\bx}^s$. Moreover, by the feasibility constraint, $\mathbf{1}^\intercal \vec{\mathbbm{x}}^{s*} = \mathbf{1}^\intercal \hat{\bx}^s = \bar{\omega}^s$, we must have $\vec{\mathbbm{x}}^{s*} = \hat{\bx}^s$. The efficient allocation is indeed supported by a Lindahl equilibrium.

(ii) This point is illustrated by Example VIII.
\end{proof}




  \bigskip 
\setlength{\bibsep}{1.3pt}
\bibliographystyle{apalike}
\bibliography{ref.bib}

\newpage 

\appendix

\renewcommand{\thetheorem}{\thesection.\arabic{theorem}}
\renewcommand{\theproposition}{\thesection.\arabic{proposition}}
\renewcommand{\thelemma}{\thesection.\arabic{lemma}}
\renewcommand{\thecorollary}{\thesection.\arabic{corollary}}
\renewcommand{\theclaim}{\thesection.\arabic{claim}}
\renewcommand{\theassumption}{\thesection.\arabic{assumption}}
\renewcommand{\thedefinition}{\thesection.\arabic{definition}}
\renewcommand{\theremark}{\thesection.\arabic{remark}}
\renewcommand{\theequation}{\thesection.\arabic{equation}}
\renewcommand{\thefigure}{\thesection.\arabic{figure}}
\counterwithin{theorem}{section}
\counterwithin{proposition}{section}
\counterwithin{lemma}{section}
\counterwithin{corollary}{section}
\counterwithin{claim}{section}
\counterwithin{assumption}{section}
\counterwithin{definition}{section}
\counterwithin{remark}{section}
\counterwithin{equation}{section}
\counterwithin{figure}{section}

\renewcommand{\thepage}{A\arabic{page}}
\setcounter{page}{1}

\begin{center}
    {\bf\huge Online Appendix}
\end{center}

\section{Examples}
In this section, we provide detailed descriptions of the examples presented in the main text.

\subsection{Details for illustrative examples in Sections \ref{sec_Examples} and \ref{sec_policy}} 

We first provide analytical characterizations of the equilibrium allocations and the contract curves, allowing for corner solutions. The interior Lindahl equilibrium follows directly from Proposition~\ref{prop.lindahl}. In the calculations below, we denote consumer 1's consumption bundle by $(x,y)$. Consumer 2’s consumption bundle is then pinned down by market clearing, namely $(2-x,2-y)$.

\paragraph*{Example I}

For an interior equilibrium, the private marginal rates of substitution are
\begin{equation*}
    MRS_1^{xy} =
    \frac{\alpha(y+\phi(2-y))}{(1-\alpha) x},
    \qquad
    MRS_2^{xy} =
    \frac{\alpha(2-y)}{(1-\alpha)(2-x)}.
\end{equation*}
Thus, an interior equilibrium satisfies $MRS_1^{xy} = MRS_2^{xy}$, which gives
\begin{equation*}
    y =
    \frac{2\big[(1+\phi)x-2\phi\big]}
    {2(1-\phi)+\phi x}.
\end{equation*}
One can verify that a corner solution arises when the right-hand side is non-positive, that is, when $(1+\phi)x-2\phi \leq 0$. Therefore, the equilibrium allocation curve is
\begin{equation*}
    \mathcal E_I:
    \qquad
    y=
    \begin{cases}
    0,
    &
    0\leq x\leq \dfrac{2\phi}{1+\phi},
    \\[1.2em]
    \dfrac{2\big[(1+\phi)x-2\phi\big]}
    {2(1-\phi)+\phi x},
    &
    \dfrac{2\phi}{1+\phi}<x \leq 2.
    \end{cases}
\end{equation*}

For efficient allocation, using feasibility, consumer 1's reduced utility is
\begin{equation*}
    \widetilde u_1(x,y)
    =
    \alpha\ln x+(1-\alpha)\ln\big(y+\phi(2-y)\big)
    =
    \alpha\ln x+(1-\alpha)\ln\big(2\phi+(1-\phi)y\big),
\end{equation*}
while
\begin{equation*}
    \widetilde u_2(x,y)
    =
    \alpha\ln(2-x)+(1-\alpha)\ln(2-y).
\end{equation*}
The planner solves
\begin{equation*}
    \max_{0\leq x,y\leq 2}
    \theta \widetilde u_1(x,y)+(1-\theta)\widetilde u_2(x,y).
\end{equation*}
For an interior solution, the first-order conditions give
\begin{equation*}
    \frac{\theta}{x}
    =
    \frac{1-\theta}{2-x},
    \quad
    \frac{\theta(1-\phi)}{2\phi+(1-\phi)y}
    =
    \frac{1-\theta}{2-y}
    \implies
    x^P(\theta)=2\theta, 
    \quad
    y^P(\theta)
    =
    \frac{2(\theta-\phi)}{1-\phi}.
\end{equation*}
The non-negativity constraint binds when $\theta\leq\phi$. Therefore, the contract curve is
\begin{equation*}
    CC_I:
    \qquad
    y=
    \begin{cases}
    0,
    &
    0\leq x\leq 2\phi,
    \\[0.8em]
    \dfrac{x-2\phi}{1-\phi},
    &
    2\phi<x\leq 2.
    \end{cases}
\end{equation*}

Thus, in Example I,
\[
    \mathcal E_I\neq CC_I.
\]
Interior competitive equilibria are inefficient. However, a corner equilibrium with
\(y_1^*=0\) lies on the bottom segment of the contract curve and is therefore efficient.

\paragraph*{Example II}

By the same approach, one can verify that 
\begin{equation*}
    \mathcal E_{II} = CC_{II}:
    \qquad
    y=
    \begin{cases}
    0,
    &
    0\leq x\leq \dfrac{2\phi}{1+\phi},
    \\[1.2em]
    \dfrac{(1+\phi)x-2\phi}{1-\phi},
    &
    \dfrac{2\phi}{1+\phi}<x<\dfrac{2}{1+\phi},
    \\[1.2em]
    2,
    &
    \dfrac{2}{1+\phi}\leq x\leq 2.
    \end{cases}
\end{equation*}
Hence, every competitive equilibrium is Pareto efficient.

\paragraph*{Example III}

Similarly, when the two goods are symmetric, the equilibrium allocation curve coincides with the contract curve and is given by the diagonal of the Edgeworth box:
\begin{equation*}
    \mathcal E_{III} = CC_{III}:
    \qquad
    y= x.
\end{equation*}
This is the same as in the benchmark case without externalities.

In Figures \ref{fig.inefficientexample}, \ref{fig.efficientexample} and \ref{fig.lindahl}, we set the initial endowment for each consumer at $(w_1,v_1)=(1.435,0.345)$ and $(w_2,v_2) = (0.565,1.655)$. Since $\phi$ takes different values across the two sections, we report the numerical results only for $\phi=0.4$ for brevity. The results are qualitatively identical to those obtained for $\phi=0.7$ in Section \ref{sec_Examples}. 

Table \ref{table.examples} reports each consumer’s consumption and utility in the competitive and Lindahl equilibrium, with the initial utilities as a benchmark. The table shows that consumer 2 is better off in competitive equilibrium in all four cases, whereas consumer 1 is worse off in the asymmetric and symmetric consumer examples. Moreover, although the Lindahl equilibrium restores efficiency, consumer 1 is worse off than in the competitive equilibrium in Examples I-III.

\begin{table}[htbp]
    \centering
    \caption{Numerical results for illustrative examples}
    \label{table.examples}
    \small
    \setlength{\tabcolsep}{5pt}
    \renewcommand{\arraystretch}{1.15}
    \resizebox{\textwidth}{!}{%
        \begin{tabular}{ccccccccccc}
        \toprule
        & \multicolumn{5}{c}{Consumer 1} & \multicolumn{5}{c}{Consumer 2} \\
        \cmidrule(lr){2-6} \cmidrule(lr){7-11}
        Example 
        & $(x^*_1,y^*_1)$ 
        & $(x^{\text{L}}_1,y^{\text{L}}_1)$ 
        & $u_1^{\text{initial}}$ 
        & $u_1^{\text{eqm.}}$
        & $u_1^{\text{L}}$
        & $(x^*_2,y^*_2)$ 
        & $(x^{\text{L}}_2,y^{\text{L}}_2)$ 
        & $u_2^{\text{initial}}$ 
        & $u_2^{\text{eqm.}}$
        & $u_2^{\text{L}}$ \\
        \midrule
        Benchmark 
        & $(0.89,\ 0.89)$  
        & $(0.89,\ 0.89)$ 
        & $-0.35$ 
        & $-0.12$  
        & $-0.12$
        & $(1.11,\ 1.11)$
        & $(1.11,\ 1.11)$
        & $-0.03$ 
        & $0.10$
        & $0.10$\\
        
        I 
        & $(1.05,\ 0.82)$ 
        & $(0.89,\ 0.15)$
        & $0.18$ 
        & $0.15$
        & $-0.12$
        & $(0.95,\ 1.18)$ 
        & $(1.11,\ 1.85)$
        & $-0.03$ 
        & $0.06$ 
        & $0.36$\\
        
        II 
        & $(0.99,\ 0.97)$ 
        & $(0.89,\ 0.74)$
        & $0.18$ 
        & $0.16$
        & $0.05$
        & $(1.01,\ 1.03)$
        & $(1.11,\ 1.26)$
        & $0.01$ 
        & $0.18$ 
        & $0.27$\\
        
        III 
        & $(0.89,\ 0.89)$  
        & $(0.15,\ 0.15)$
        & $0.26$ 
        & $0.29$  
        & $-0.12$
        & $(1.11,\ 1.11)$
        & $(1.85,\ 1.85)$
        & $-0.03$ 
        & $0.10$
        & $0.62$\\
        \bottomrule
        \end{tabular}%
    }
\end{table}


\subsection{Regular multiplex networks with endogenous prices}\label{appsec.regular}


Consider the case in which all networks are regular and undirected, such that $(\bG^s)^\intercal \mathbf{1} = d^s \mathbf{1}$ for some $d^s \geq 0$. Then we have $\mathbf{1}^\intercal \bM^s = \mathbf{1}^\intercal/(1+\phi^s d^s)$. Assume further that consumers are homogeneous, in the sense that all consumers hold identical endowments of each good. Then $\etab^s = \mathbf{1}/n$ for every $s \in \mS$. Equation~\eqref{eq.mudetermine} implies
\begin{equation*}
    \frac{\se}{\mathbf{1}^\intercal \se} = 
    \left( \sum_{s \in \mathcal{S}} \alpha^s \bM^s \right)^{-1} 
    \left( \sum_{s \in \mathcal{S}} \frac{\alpha^s}{1+\phi^sd^s} 
    \underbrace{\frac{1}{n} \mathbf{1}}_{\etab^s} \right)
    = \frac{1}{n} \mathbf{1}.
\end{equation*} 
By Theorem~\ref{thm.interior}, the equilibrium allocation and prices are given by
\begin{equation*}
    x^{s*}_i = \frac{\bar{\omega}^s}{n}, \qquad
    p^{s*} = \frac{\alpha^s}{\bar{\omega}^s}
    \frac{\mathbf{1}^\intercal \se}{1+\phi^s d^s}.
\end{equation*}
Thus, the equilibrium consumption allocation is independent of the network structure, where each consumer simply consumes her initial endowment of each good. The heterogeneity across layers is instead fully absorbed by equilibrium prices.

For consumer $i$'s relative expenditure on goods $s$ and $t$, we have
\begin{equation*}
    \frac{p^{s*}x^{s*}_i}{p^{t*}x^{t*}_i}
    =
    \frac{\alpha^s/(1+\phi^s d^s)}
    {\alpha^t/(1+\phi^t d^t)}.
\end{equation*}
This expression coincides with the regular-network and homogeneous-consumer example in \cite{zenouGamesMultiplexNetworks2025}.

The welfare implications, however, differ. In our general equilibrium setting, the competitive equilibrium is efficient because the centrality parallel condition holds. By contrast, the equilibrium in \cite{zenouGamesMultiplexNetworks2025} is inefficient. The reason is that, in our model, consumers’ budgets are endogenized through market-clearing prices. These prices adjust in response to the distortions generated by network externalities. In the regular-network case, this endogenous price adjustment is sufficient to restore efficiency. With exogenous prices, however, this market mechanism is absent, and the equilibrium is not efficient.


\section{General utility functions}\label{appsec.general}


Suppose the utility function of each consumer $i$, $u_i$, is continuously differentiable, concave and $\frac{\partial u_i(\bX)}{\partial x_i^s}>0$. We further assume that utilities are scale monotone in each good: for every allocation $\bX$, every good $s$, and every $t\geq1$, $u_i(\bx^1,\ldots,t\bx^s,\ldots,\bx^{|\mS|}) \geq u_i(\bX), \forall i\in\mN$. Thus, proportionally expanding any good's consumption profile across all consumers cannot reduce any consumer's utility.\footnote{In the benchmark, this requires that negative externalities are not too large.}

For $i\neq j$, define the normalized marginal externality from $j$'s consumption of good $s$ to consumer $i$ by
\begin{equation}\label{appeq.generalnetwork}
    g_{ij}^s(\bX)
    :=
    \frac{\partial u_i(\bX)/\partial x_j^s}
         {\partial u_i(\bX)/\partial x_i^s},
    \qquad
    g_{ii}^s(\bX):=0.
\end{equation}
Let $\bG^s(\bX)=(g_{ij}^s(\bX))_{i,j\in\mN}$ and
\begin{equation*}
    \bB^s(\bX):=\mathbf I_n+[\bG^s(\bX)]^\intercal.
\end{equation*}
Unlike the adjacency matrices in the main text, $\bG^s(\bX)$ generally depends on the allocation and may contain positive or negative entries. Define the set of centrality vectors of layer $s$ as
\begin{equation*}
    \mathcal{B}^s(\bX)
    :=
    \left\{
      \mathbf{z} \in \mathbb{R}^n_{++}:
      \bB^s(\bX)\mathbf{z} \parallelsum \mathbf{1}
    \right\}.
\end{equation*}

\begin{theorem}\label{app:thm.efficiency.general}
Let $(\bX^*,\bp^*)$ be an interior competitive equilibrium. 

\begin{enumerate}[label=(\roman*)]
    \item If there exists a common centrality vector, that is
    \begin{equation}\label{appeq.general.alignment}
        \bigcap_{s \in \mS}\,
        \mathcal{B}^s(\bX^*) 
        \neq
        \emptyset,
    \end{equation}
    then $\bX^*$ is Pareto efficient.  

    \item Conversely, if $\bX^*$ is Pareto efficient and admits a supporting Pareto weight in $\Delta^o$, then~\eqref{appeq.general.alignment} holds.

\end{enumerate}
\end{theorem}

\begin{proof}
    At an interior competitive equilibrium, consumer $i$'s first-order condition is
    \begin{equation}\label{appeq.general.ce.foc}
        \frac{\partial u_i}{\partial x^s_i}
        =\lambda_ip^{s*}.
    \end{equation}
    Fix a positive welfare vector $\btheta$ and consider the planner problem
    \[
        \max_{\bX\in\mathcal{A}}\sum_{j\in\mN}\theta_j u_j(\bX).
    \]
    The first-order condition with respect to $x_i^s$ is
    \begin{equation*}
        \beta^s=
        \theta_i \frac{\partial u_i}{\partial x^s_i}
        +\sum_{j\neq i}\theta_j
          \frac{\partial u_j}{\partial x_i^s}
        =
        \theta_i \frac{\partial u_i}{\partial x^s_i}
        +\sum_{j\neq i}\theta_jg_{ji}^s \frac{\partial u_j}{\partial x^s_j}.
    \end{equation*}
    Using~\eqref{appeq.general.ce.foc}, evaluate the right-hand side at $\bX^*$ and write it in vector form,
    \begin{equation*}
        \bB^s(\bX^*)
        \operatorname{diag} \left(\frac{\partial u_1}{\partial x^s_1},\ldots,\frac{\partial u_n}{\partial x^s_n} \right)\btheta
        = p^{s*} \bB^s(\bX^*)\, ( \btheta \odot \boldsymbol{\lambda})
    \end{equation*}
    
    For part (i), pick $\mathbf{b} \in \bigcap_{s \in \mS}\, \mathcal{B}^s(\bX^*)$ such that $\bB^s(\bX^*) \mathbf{b} = \kappa^s \mathbf{1}$ for all $s \in \mS$. Let $\beta^s := \kappa^s p^s$ and $\theta_i := b_i/\lambda_i$, then we have 
    \[
    p^{s*} \bB^s(\bX^*)\, ( \btheta \odot \boldsymbol{\lambda})
    =
    p^{s*} \bB^s(\bX^*) \mathbf{b} = \kappa^s p^{s*} \mathbf{1} = \beta^s \mathbf{1}.
    \]
    Thus $\bX^*$ satisfies the KKT conditions of the weighted planner problem. Concavity makes these conditions sufficient. Since $\btheta \gg 0$ by definition, any Pareto improvement would strictly raise the planner objective, which contradicts the optimality of $\bX^*$. Therefore, $\bX^*$ is Pareto efficient.
    
    For part (ii), denote the associated Pareto weight by $\btheta \in \Delta^o$, for which $\bX^*$ solves the weighted planner problem.  Because $\bX^*$ is interior, its KKT conditions reduce to
    \[
    p^{s*} \bB^s(\bX^*)\, ( \btheta \odot \boldsymbol{\lambda}) = \beta^s \mathbf{1}.
    \]
    Setting $\mathbf{b}=\btheta \odot \boldsymbol{\lambda} \gg 0$ gives~\eqref{appeq.general.alignment}.
\end{proof}

Theorem~\ref{app:thm.efficiency.general} shows that the centrality parallel condition extends beyond the linear effective-consumption specification. For general preferences, the relevant network object is the normalized marginal externality matrix evaluated at the equilibrium allocation. Efficiency requires that these allocation-specific marginal externality structures across goods admit a common centrality vector.

In the baseline model, this matrix is constant, so the condition becomes a restriction on the multiplex network itself. For general preferences, the same economic mechanism remains, but the condition becomes equilibrium-specific since marginal externalities depend on the allocation.

\begin{remark}\label{app:remark.alpha}
    In the baseline model, we assume a common preference weight $\alpha^s$ across agents for each good $s$. Allowing for agent-specific weights $\alpha_i^s$ generally makes the characterization of equilibrium more involved. Nevertheless, as can be seen from~\eqref{appeq.generalnetwork}, these preference weights cancel out when computing the normalized marginal externalities. Consequently, the effective network remains constant and is still given by $\phi^s \bG^s$, exactly as in the baseline model. Hence, the centrality parallel condition continues to characterize equilibrium efficiency even in the presence of heterogeneous preference weights.
\end{remark}

\begin{remark}\label{app:remark.income}
    To illustrate the income effect, we apply Theorem~\ref{app:thm.efficiency.general} to Example 3 in \cite{dufwenbergOtherRegardingPreferencesGeneral2011}. Consider a one-good exchange economy with two groups of size $n \geq 2$ each, rich (R) and poor (P) agents. For each consumer $i$, the utility function takes the form:\footnote{Their utility function  depends on the budget sets of all agents. Since this is a one-good economy, we use the following utility where the results are equivalent to theirs.}
    \begin{equation*}
        u_i (\bX) = 
        \begin{cases}
            x_i^\alpha \quad & i \in P\\
            x_i^\alpha+c \sum_{j\neq i}  x_j^\alpha \quad & i\in R
        \end{cases}
    \end{equation*}
    where $c=\beta/(2n-1)$ and $\alpha \in (0,1)$. The initial endowment is given by $\omega_i = 1$ for the rich agents and $\omega_i = \eta <1$ for the poor agents. It is clear that the unique competitive equilibrium is the initial endowments with price $p^*=1$.

    One can easily check that, by our definition, for $i\neq j$
    \begin{equation*}
        g_{ij} (\bx) = 
        \begin{cases}
            c(x_j/x_i)^{\alpha-1} \quad & i \in R\\
            0 \quad & i \in P
        \end{cases}
    \end{equation*}
    which indeed depends on the allocation. Evaluating at the equilibrium gives
    \[
        \bG(\bX^*) = 
        \begin{pmatrix}
            c [\mathbf{1\mathbf{1}^\intercal-\mathbf{I}}] & c \eta^{\alpha-1} \mathbf{1}\mathbf{1}^\intercal\\
            \mathbf{O} & \mathbf{O}
        \end{pmatrix}.
    \]
    Then condition~\eqref{appeq.general.alignment} is equivalent to the one such that  there exists $\mathbf{z} \in \mathbb{R}^{2n}_{++}$ such that
    \begin{align*}
        z_i + c \sum_{j \in R,j\neq i} z_j = 1 \quad i \in R,
        \qquad
        z_i+ c \eta^{\alpha-1} \sum_{j\in R}z_j = 1 \quad i \in P.
    \end{align*}
    Denote $S_R = \sum_{j \in R} z_j$, then the solution to this linear system is $S_R/n$ for rich agents and $1-c\eta^{\alpha-1} S_R$ for poor agents, where $S_R = \frac{n}{1+c(n-1)}$. To ensure $z_i > 0$, we must have
    \begin{equation*}
        1-c\eta^{\alpha-1} S_R>0
        \iff
        c n\eta^{\alpha-1} < 1+c(n-1).        
    \end{equation*}
    This inequality is violated when $\eta$ is small enough and thus the equilibrium is inefficient. In fact, this is exactly the threshold derived in their example.
\end{remark}

\begin{remark}\label{app:remark.crossgood}
    General utility functions also allow for cross-good spillovers, so that $\bG^s$ may depend on the allocation of goods other than $s$. To illustrate this point, consider an economy with two goods, $x$ and $y$, and two consumers with other-regarding preferences:
    \begin{equation*}
        u_1(\bx,\mathbf{y})=(x_1y_1)^{1/4}+\frac{1}{2} (x_2)^{(1-\delta)/4}(y_2)^{(1+\delta)/4}, \quad
        u_2(\bx,\mathbf{y})=(x_2y_2)^{1/4}+\frac{1}{2} (x_1)^{(1+\delta)/4}(y_1)^{(1-\delta)/4},
    \end{equation*}
    where $\delta\in[0,1)$. When $\delta>0$, consumer 1 places relatively more weight on consumer 2's consumption of good $y$, whereas consumer 2 places relatively more weight on consumer 1's consumption of good $x$. The case $\delta=0$ reduces to symmetric well-being externalities.
    
    According to~\eqref{appeq.generalnetwork}:
    \begin{equation*}
        g^x_{12} = \frac{1-\delta}{2} \left( \frac{x_2}{x_1} \right)^{-3/4} \left( \frac{y_2}{y_1} \right)^{1/4} \left( \frac{y_2}{x_2} \right)^{\delta/4},
        \quad
        g^x_{21} = \frac{1+\delta}{2} \left( \frac{x_1}{x_2} \right)^{-3/4} \left( \frac{y_1}{y_2} \right)^{1/4} \left( \frac{y_1}{x_1} \right)^{-\delta/4}.
    \end{equation*}
    Thus, the marginal-externality network associated with good $x$ depends not only on the allocation of $x$, but also on that of $y$. Analogously, the marginal-externality network associated with good $y$ depends on the allocation of $x$.
   
    Suppose both consumers have endowment one for both goods. Since consumer $i$ takes the other consumer's allocation as given, her private choice is equivalent to maximizing $(x_iy_i)^{1/4}$, which is equivalent to Cobb-Douglas utility. Normalizing the price of good $y$ to one, the unique interior competitive equilibrium is given by $x_i=y_i = 1$ with $p^*=1$. This implies
    \begin{equation*}
        \bG^{x*} = (\bG^{y*})^\intercal = 
        \begin{pmatrix}
            0 & (1-\delta)/2
            \\
            (1+\delta)/2 & 0
        \end{pmatrix}.
    \end{equation*}
    Then one can check that the generalized centrality parallel condition is equivalent to $\delta=0$.

    Under general utilities, cross-good spillovers make each marginal-externality network depend on the entire consumption bundle. When the cross-good spillovers are symmetric ($\delta=0$), the competitive equilibrium remains efficient. In contrast, when $\delta>0$, the two consumers place relatively greater weight on different goods in the other consumer's bundle. This asymmetry generates opposite distortions across the two layers, so that the competitive equilibrium becomes inefficient.
\end{remark}

\paragraph*{Efficiency measures.} For general utility functions, the two measures of efficiency can still be defined. Although the explicit forms are generally unavailable, the link between the two measures can still be established. For each $\gamma \in [0,1]$, define the $\gamma$-attainable set by
\begin{equation*}
    \mathcal{A}(\gamma) := 
    \left\{
    \bX \in \mathbb{R}^{|\mS| \times |\mN|}_+: 
    \sum_{i \in \mN} x^s_i \leq \gamma \bar{\omega}^s, \forall s \in \mS
    \right\}.
\end{equation*}
Given an equilibrium allocation $\bX^*$, recall that the $\mathcal{CRU}$ can be defined as
\begin{equation}\label{appeq.cru}
    \begin{gathered}
        \mathcal{CRU} := \inf_{\bX \in \mathcal{A}(\gamma),\gamma \in [0,1]} \gamma\\
        s.t. \ 
        u_i(\bX) \geq u_i(\bX^*), \forall i \in \mN.
    \end{gathered}
\end{equation}
For any Pareto weight $\btheta \in \Delta$, define the $\gamma$-utility loss 
\begin{equation*}
    L_\gamma(\btheta) := V_{\gamma}(\btheta)-\btheta^\intercal \mathbf{u}(\bX^*),
    \quad
    \text{where} ~~
    V_{\gamma}(\btheta) :=\max_{\bX \in \mathcal{A}(\gamma)} 
    \btheta^\intercal \mathbf{u}(\bX).
\end{equation*}
In particular, the efficiency loss $L(\btheta)$ defined in the main text is exactly $L_1(\btheta)$.

\begin{proposition}\label{app:prop.cru}
    The two measures are linked by
    \begin{equation*}
        \mathcal{CRU} =
        \inf
        \left\{
        \gamma \in [0,1]:
        \inf_{\btheta \in \Delta} L_\gamma(\btheta) \geq 0
        \right\}.
    \end{equation*}
\end{proposition}

\begin{proof}
If $\gamma$ is feasible in problem~\eqref{appeq.cru}, then there exists $\bX\in\mathcal{A}(\gamma)$ such that $\mathbf{u}(\bX)\geq \mathbf{u}(\bX^*)$. Hence, for every $\btheta\in\Delta$,
\begin{equation*}
    V_{\gamma}(\btheta)
    \geq
    \btheta^\intercal \mathbf{u}(\bX)
    \geq
    \btheta^\intercal \mathbf{u}(\bX^*),
\end{equation*}
which implies $L_\gamma(\btheta)\geq 0$ for all $\btheta\in\Delta$. Therefore,
$\inf_{\btheta\in\Delta}L_\gamma(\btheta)\geq 0$.

Define the set $\mathcal{U}(\gamma):=\left\{\mathbf{v} \in\mathbb{R}^n: \exists \bX\in\mathcal A(\gamma), v_i \leq u_i(\bX),\ \forall i\in\mathcal N \right\}$. By the same argument used in the proof of Lemma 2, $\mathcal{U}(\gamma)$ is closed and convex. Suppose that $\gamma$ is not feasible in problem~\eqref{appeq.cru}. Then 
$\mathbf{u}(\bX^*) \notin \mathcal{U}(\gamma)$. The strict separating hyperplane theorem implies that there exists a nonzero vector $\btheta$ such that
\begin{equation*}
    \btheta^\intercal \mathbf{u}(\bX^*)
    >
    \sup_{\mathbf{v} \in\mathcal{U}(\gamma)}
    \btheta^\intercal \mathbf{v}.
\end{equation*}
Moreover, the supporting vector must satisfy $\btheta\in\mathbb R^n_+$. Otherwise, if $\theta_j<0$ for some $j$, decreasing $v_j$ would increase $\btheta^\intercal \mathbf{v}$ unboundedly, contradicting the supporting inequality. Normalizing $\btheta$ so that $\mathbf{1}^\intercal\btheta=1$, we obtain $\btheta\in\Delta$. Moreover, by the definition of $\mathcal{U}(\gamma)$, we have $\sup_{\mathbf{v} \in\mathcal{U}(\gamma)} \btheta^\intercal \mathbf{v} = V_{\gamma}(\btheta)$. Therefore,
\begin{equation*}
    L_\gamma(\btheta)
    =
    V_{\gamma}(\btheta)
    -
    \btheta^\intercal \mathbf{u}(\bX^*)
    <0,
\end{equation*}
which implies $\inf_{\btheta\in\Delta}L_\gamma(\btheta)<0$. This proves the desired equivalence.
\end{proof}

From the definition, the $\gamma$-utility loss can be decomposed as
\begin{equation*}
    L_\gamma(\btheta) = L(\btheta)+ \left\{ V_\gamma(\btheta)-V(\btheta) \right\}.
\end{equation*}
The first term  $L(\btheta)$ is the efficiency loss at the competitive equilibrium, which is non-negative. The second term $\left\{ V_\gamma(\btheta)-V(\btheta) \right\}$ is the total welfare decrease due to resource decrease, which is non-positive. Hence, $\mathcal{CRU}$ is the smallest resource fraction for which the full-resource efficiency loss is sufficient to compensate for the welfare loss induced by resource contraction for every Pareto weight $\btheta\in\Delta$.

\paragraph*{Divergence representation of $L(\btheta)$.}
The loss function cannot generally be represented by a divergence form. However, if externality channels are separated across goods and preferences satisfy additive separability,
\begin{equation}\label{appeq.additive.layers}
    u_i(\bX)=\sum_{s\in\mS}u_i^s(\bx^s),
\end{equation}
the efficiency loss can be decomposed into layer-specific components. For each layer $s$, define the planner value function
\begin{equation}\label{app:eq.vs}
\begin{gathered}
    V^s(\btheta)
    :=
    \max_{\bx^s\geq0}
    \sum_{i\in\mN}\theta_i u_i^s(\bx^s)\\
    s.t. \ \mathbf{1}^{\intercal}\bx^s \leq \bar\omega^s.
\end{gathered}
\end{equation}
Although the equilibrium allocation need not be efficient in the aggregate economy, each layer allocation can be rationalized as the solution of a layer-specific planner problem with an implied welfare weight. The following proposition shows that the efficiency loss admits an exact Bregman divergence representation among these weights.

\begin{proposition}\label{app:prop.bregman}
    Suppose~\eqref{appeq.additive.layers} holds and each
    $u_i^s$ is concave and continuously differentiable. 
    Assume that $V^s$ is continuously differentiable and $\mathcal{B}^s(\bX^*)\neq\emptyset$. Then there exists $\brho^s \in \Delta^o$ for every $s \in \mS$ such that
    \[
    \bx^{s*} \in
    \arg\max_{\bx^s\geq0:
    \mathbf{1}^\intercal\bx^s\leq \bar\omega^s}
    \sum_{i\in\mN}
    \rho_i^s u_i^s(\bx^s).
    \]
    Moreover,
    \[
    L(\btheta)
    =
    \sum_{s\in\mS}
    D_{V^s}
    (\btheta,\brho^s)
    \geq0,
    \]
    where
    \[
    D_{V^s}
    (\btheta,\brho^s)
    =
    V^s(\btheta)
    -
    V^s(\brho^s)
    -
    \nabla V^s(\brho^s)^\intercal
    (\btheta-\brho^s)
    \]
    is the Bregman divergence.
\end{proposition}

\begin{proof}
    Fix any layer $s\in\mS$, since $\mathcal{B}^s(\bx^*) \neq \emptyset$, there exists $\kappa^s \geq 0$ and $\brho^s \in \Delta^o$ such that
    \[
    \bB^s(\bX^*)\,
    (\brho^s \odot \boldsymbol{\lambda})
    =
    \kappa^s \mathbf{1}
    \]
    Following the same argument as in the proof of Theorem~\ref{app:thm.efficiency.general}, the equilibrium allocation $\bx^{s*}$ satisfies the KKT conditions of the layer-specific planner problem with welfare weight~$\brho^s$. By concavity:
    \[
    \bx^{s*}
    \in
    \arg\max_{\bx^s\geq0:
    \mathbf{1}^\intercal\bx^s \leq \bar\omega^s}
    \sum_{i\in\mN}
    \rho_i^s u_i^s(\bx^s).
    \]
    Hence,
    \[
    V^s(\brho^s) =
    \sum_i\rho_i^s u_i^s(\bx^{s*}) =
    (\brho^s)^\intercal\mathbf u^{s*}.
    \]
    Moreover, since $V^s$ is continuously differentiable,
    the envelope theorem gives
    \[
    \nabla V^s(\brho^s) =
    \mathbf u^{s*}.
    \]

    By additive separability,
    the planner problem decomposes across layers $V(\btheta) = \sum_{s\in\mS}V^s(\btheta)$, and the efficiency loss is therefore
    \[
    L(\btheta)
    =
    \sum_{s\in\mS}
    \left(
    V^s(\btheta)
    -
    \btheta^\intercal\mathbf u^{s*}
    \right).
    \]
    For each layer,
    substituting the two identities above yields
    \[
    V^s(\btheta) - \btheta^\intercal\mathbf u^{s*}
    =
    V^s(\btheta) - V^s(\brho^s) - \nabla V^s(\brho^s)^\intercal (\btheta-\brho^s)
    =
    D_{V^s} (\btheta,\brho^s).
    \]
    Summing over all layers gives $L(\btheta) = \sum_{s\in\mS}
    D_{V^s} (\btheta,\brho^s)$.
 
    Finally, $V^s$ is convex because it is the pointwise maximum of linear functions of $\btheta$. Therefore, the Bregman divergence is nonnegative and thus $L(\btheta)\geq0$.
\end{proof}

\paragraph{Special utilities.}
We next specialize Proposition~\ref{app:prop.bregman} to two commonly used utility specifications. Both admit closed-form planner value functions and hence explicit Bregman representations.

Suppose consumptions are aggregated by linear effective consumptions $q_i^s = x_i^s+\phi^s\sum_{j\neq i}g_{ij}^sx_j^s$ and externalities are sufficiently small. The utilities are additively separable and the sub-utilities for each good share a common functional form:
\begin{equation*}
    u_i(\bx) = \sum_{s \in \mS} \alpha^s \tilde{u}(q^s_i).
\end{equation*}

\begin{proposition}\label{app:prop.special.utilities}

Suppose $\btheta \in \Delta^o$ and supports an interior efficient allocation.

\begin{enumerate}
    \item[(i)] Under the CRRA sub-utility:
    \begin{equation*}
        \tilde{u}(q) = \frac{q^r-1}{r}, \qquad 0<r<1.
    \end{equation*}
    We have
    \begin{equation*}
        L(\btheta) = \sum_{s \in \mS} 
        \frac{\alpha^s(\bar\omega^s)^r}{r}
        \bigl(\Pi^s(\btheta)\bigr)^{1-r}\times
        \left[
            1-
            \exp\left\{
                -(1-r)
                D_r
                \left(
                    \boldsymbol\zeta^s
                    \Vert
                    \boldsymbol\pi^s(\btheta)
                \right)
            \right\}
        \right],
    \end{equation*}
    where
    \[
        \Pi^s(\btheta) :=
        \sum_{i\in\mN}
        \theta_i^{\frac{1}{1-r}}
        (\tilde b_i^s)^{-\frac{r}{1-r}},
        \quad
        \zeta_i^s :=
        \frac{
            (\rho_i^s)^{\frac{1}{1-r}}
            (\tilde b_i^s)^{-\frac{r}{1-r}}
        }{
            \Pi^s(\brho^s)
        },
        \quad
        \pi_i^s(\btheta) :=
        \frac{
            \theta_i^{\frac{1}{1-r}}
            (\tilde b_i^s)^{-\frac{r}{1-r}}
        }{
            \Pi^s(\btheta)
        }.
    \]
    and $D_r(\mathbf y\Vert\mathbf z) := \frac{1}{r-1} \ln \left( \sum_i y_i^rz_i^{1-r} \right)$ is the R\'enyi divergence of order $r$.
    
    \item[(ii)] Under the CARA sub-utility:
    \[
        \tilde{u}(q) = \frac{1-e^{-a q}}{a},
        \qquad
        a>0,
    \]
    We have
    \begin{equation*}
        L(\btheta) = \sum_{s \in \mS}
        \frac{\alpha^s}{a} e^{-\kappa^s}
        \prod_{i\in\mN}
        \left(
            \frac{\rho^s_i}{\nu_i^s}
        \right)^{\nu_i^s}
        \times
        \left[
            \sum_{i\in\mN}
            \nu_i^s\frac{\theta_i}{\rho_i^s}
            -
            \prod_{i\in\mN}
            \left(
                \frac{\theta_i}{\rho_i^s}
            \right)^{\nu_i^s}
        \right],
    \end{equation*}
    where
    \[
        \nu_i^s := \frac{\tilde b_i^s}{\mathbf1^\intercal\widetilde{\mathbf b}^s},
        \qquad
        \kappa^s :=
        \frac{a\bar\omega^s}{\mathbf1^\intercal\widetilde{\mathbf b}^s},
    \]
\end{enumerate}

\end{proposition}

\begin{proof}
By Proposition~\ref{app:prop.bregman}, it suffices to calculate the Bregman divergence for each sub-utility. 

(i) For CRRA sub-utility, fixing a layer $s$, the first-order conditions of the
planner problem are
\[
    \bB^s \left(\theta_i(q_i^s)^{r-1}\right)_{i\in\mN}
    = \eta^s\mathbf1
    \implies
    q_i^s =
    \frac{\bar\omega^s}{\Pi^s(\btheta)}
    \left(
        \frac{\theta_i}{\tilde b_i^s}
    \right)^{\frac{1}{1-r}},
\]
where the coefficient $\eta^s$ follows from $(\widetilde{\mathbf b}^s)^\intercal\mathbf q^s=\bar\omega^s$. Substituting into the planner objective gives
\[
    V^s(\btheta) =
    \frac{\alpha^s}{r}
    \left[
        (\bar\omega^s)^r
        \bigl(\Pi^s(\btheta)\bigr)^{1-r}
        -1
    \right],
\]
and thus
\[
    \nabla_i V^s(\btheta)
    =
    \frac{\alpha^s}{r}
    (\bar\omega^s)^r
    \bigl(\Pi^s(\btheta)\bigr)^{-r}
    \theta_i^{\frac{r}{1-r}}
    (\tilde b_i^s)^{-\frac{r}{1-r}}.
\]
It follows that
\begin{equation*}
    D_{V^s}(\btheta,\brho^s) =
    \frac{\alpha^s(\bar\omega^s)^r}{r}
    \Bigg[
    \bigl(\Pi^s(\btheta)\bigr)^{1-r}
    -
    \bigl(\Pi^s(\brho^s)\bigr)^{-r}
    \sum_i
    \theta_i
    (\rho_i^s)^{\frac{r}{1-r}}
    (\tilde b_i^s)^{-\frac{r}{1-r}}
    \Bigg].
\end{equation*}

By definition of the R\'enyi divergence,
\begin{equation*}
    \exp\left\{-(1-r) D_r \left( \boldsymbol\zeta^s \Vert\boldsymbol\pi^s(\btheta)\right) \right\}
    =
    \sum_i
    (\zeta_i^s)^r
    \bigl(\pi_i^s(\btheta)\bigr)^{1-r}
    =
    \frac{
        \sum_i
        \theta_i
        (\rho_i^s)^{\frac{r}{1-r}}
        (\tilde b_i^s)^{-\frac{r}{1-r}}
    }{
        \bigl(\Pi^s(\brho^s)\bigr)^r
        \bigl(\Pi^s(\btheta)\bigr)^{1-r}
    }.
\end{equation*}
Then we obtain
\[
    D_{V^s}(\btheta,\brho^s)
    =
    \frac{\alpha^s(\bar\omega^s)^r}{r}
    \bigl(\Pi^s(\btheta)\bigr)^{1-r}
    \left[
        1-
        \exp\left\{
            -(1-r)
            D_r
            \left(
                \boldsymbol\zeta^s
                \Vert
                \boldsymbol\pi^s(\btheta)
            \right)
        \right\}
    \right].
\]
Summing over $s$ proves part~(i).

(ii) For CARA sub-utility, the first-order conditions are
\[
    \bB^s
    \left(
        \theta_i e^{-a q_i^s}
    \right)_{i\in\mN}
    =
    \eta^s\mathbf1
    \implies
    \theta_i e^{-a q_i^s} = \eta^s\tilde b_i^s.
\]
Using
$(\widetilde{\mathbf b}^s)^\intercal\mathbf q^s=\bar\omega^s$
and the definitions of $\nu_i^s$ and $\kappa^s$ gives
\[
    \eta^s
    \mathbf1^\intercal\widetilde{\mathbf b}^s
    =
    e^{-\kappa^s}
    \prod_{i\in\mN}
    \left(
        \frac{\theta_i}{\nu_i^s}
    \right)^{\nu_i^s}.
\]
Therefore,
\[
    V^s(\btheta)  =
    \frac{\alpha^s}{a}
    \left[
        1-
        e^{-\kappa^s}
        \prod_{i\in\mN}
        \left(
            \frac{\theta_i}{\nu_i^s}
        \right)^{\nu_i^s}
    \right].
\]
Differentiating gives
\[
    \nabla_i V^s(\btheta) =
    -
    \frac{\alpha^s}{a}
    e^{-\kappa^s}
    \prod_{j\in\mN}
    \left(
        \frac{\theta_j}{\nu_j^s}
    \right)^{\nu_j^s}
    \frac{\nu_i^s}{\theta_i}.
\]
Substituting this expression into the Bregman divergence yields
\[
    D_{V^s}(\btheta,\brho^s) =
    \frac{\alpha^s}{a}
    e^{-\kappa^s}
    \prod_{i\in\mN} \left( \frac{\rho_i^s}{\nu_i^s} \right)^{\nu_i^s}
    \times
    \left[
    \sum_{i\in\mN} \nu_i^s\frac{\theta_i}{\rho_i^s}
    -
    \prod_{i\in\mN}  \left(\frac{\theta_i}{\rho_i^s} \right)^{\nu_i^s} \right].
    \]
Summing over $s$ proves part~(ii).
\end{proof}

The CRRA and CARA specifications also yield tractable scaling properties with respect to aggregate resources. Combined with Proposition~\ref{app:prop.cru}, these properties provide an explicit link between the $\mathcal{CRU}$ and $L(\btheta)$.

\begin{corollary}\label{app:cor.special}
    Suppose the conditions in Proposition~\ref{app:prop.special.utilities} hold,
    \begin{enumerate}
        \item[(i)] Under the CRRA sub-utility:
        \begin{equation}\label{appeq.ces-cru}
            \mathcal{CRU} =
            \left[
            1- \min_{\btheta\in\Delta}\frac{rL(\btheta)}{1+rV(\btheta)}
            \right]^{1/r}.
        \end{equation}

        \item[(ii)] Under the CARA sub-utility, suppose additionally that the resource-minimization problem~\eqref{appeq.cru} admits an interior solution with strictly positive effective consumption.
        \begin{equation*}
            \mathcal{CRU} = 
            \inf \left\{ \gamma\in[0,1]: \min_{\btheta\in\Delta} 
            \left[ L(\btheta) - 
            \sum_{s\in\mS} \left( e^{(1-\gamma)\kappa^s}-1 \right) 
            \left( \frac{\alpha^s}{a}-V^s(\btheta) \right) 
            \right] \geq0 \right\}. 
        \end{equation*} 
        If, in addition, $\kappa^s=\kappa$, for all $s\in\mS$, then \begin{equation}\label{appeq.cara-cru.common} 
            \mathcal{CRU} = 1- 
            \frac{1}{\kappa} 
            \ln \left[ 1+ \min_{\btheta\in\Delta} \frac{L(\btheta)} {1/a-V(\btheta)} \right]. 
        \end{equation}
    \end{enumerate}
\end{corollary}
\begin{proof}
    (i) Under CRRA sub-utility, by linearity of effective consumption and the scaling of the feasible set, for any $\gamma\in[0,1]$: 
    \begin{equation*}
        V_\gamma^s(\btheta) =
        \gamma^r V^s(\btheta) - \frac{\alpha^s(1-\gamma^r)}{r}.
    \end{equation*}
    Since $\sum_s\alpha^s=1$, aggregating across layers gives
    \[
        V_\gamma(\btheta) = \gamma^r V(\btheta) - \frac{1-\gamma^r}{r},
    \]
    and hence
    \begin{equation*}
        L_\gamma(\btheta) = L(\btheta) - (1-\gamma^r) \left(V(\btheta)+\frac{1}{r} \right).
    \end{equation*}
    Then equation~\eqref{appeq.ces-cru} follows from Proposition~\ref{app:prop.cru}.
    
    (ii) Under CARA sub-utility, 
    \begin{equation*}
    \frac{\alpha^s}{a}
    -
    V_\gamma^s(\btheta)
    \geq
    e^{(1-\gamma)\kappa^s}
    \left[
    \frac{\alpha^s}{a}
    -
    V^s(\btheta)
    \right],
    \end{equation*}
    where the inequality comes from allowing for corner solutions. Summing across goods gives
    \begin{equation*}
    L_\gamma(\btheta)
    \leq
    L(\btheta)
    -
    \sum_{s\in\mS}
    \left(
    e^{(1-\gamma)\kappa^s}-1
    \right)
    \left[
    \frac{\alpha^s}{a}
    -
    V^s(\btheta)
    \right]
    =: F_\gamma(\btheta).
    \end{equation*}
    By the assumed interiority of the resource-minimizing allocation, there exists a supporting Pareto weight $\hat{\btheta}\in\Delta^o$ such that inequalities above hold with equality at $\gamma=\mathcal{CRU}$ and $\btheta=\hat{\btheta}$. Proposition~\ref{app:prop.cru} then implies $\min_{\btheta\in\Delta}F_{\mathcal{CRU}}(\btheta)=0$. Since $F_\gamma(\hat{\btheta})$ is strictly increasing in $\gamma$, it is negative for every $\gamma<\mathcal{CRU}$, proving the first expression in part (ii).

    If, in addition, $\kappa^s=\kappa$ for all $s\in\mS$, the scaling is common across layers. Since $\sum_s\alpha^s=1$,
    \[
        F_\gamma(\btheta) =
        L(\btheta)
        -
        \left( e^{(1-\gamma)\kappa}-1 \right)
        \left[
        \frac{1}{a} - V(\btheta)
        \right].
    \]
    The first expression in part (ii) then yields~\eqref{appeq.cara-cru.common}.
\end{proof}


\paragraph*{Cobb-Douglas utility.}
Finally, for the Cobb-Douglas utility in the main text, we can consider the limiting case in the CRRA-sub-utility form as $r\downarrow0$. Note that 
\begin{align*}
    \sum_i (\zeta_i^s)^r (\pi_i^s(\btheta))^{1-r}=
    \sum_i \pi_i^s(\btheta) \exp \left( r \ln \frac{\zeta_i^s}{\pi_i^s(\btheta)} \right)
    =
    1-rD_{\mathrm{KL}}(\boldsymbol\pi^s(\btheta)\Vert\boldsymbol{\zeta}^s)
    +o(r).
\end{align*}
Hence,
\[
    \exp
    \left\{-(1-r)D_r(\boldsymbol{\zeta}^s \Vert \boldsymbol\pi^s(\btheta))\right\}
    =
    1-r D_{\mathrm{KL}}(\boldsymbol\pi^s(\btheta)\Vert\boldsymbol{\zeta}^s)
    +o(r).
\]
Combining with the fact that $\boldsymbol{\pi}^s(\btheta) \to \btheta$, $\boldsymbol{\zeta}^s \to \brho^s$ and $\Pi^s(\btheta) \to 1$, we obtain
\begin{equation*}
    \lim_{r \downarrow 0} D_{V^s}(\btheta,\brho^s)
    = \alpha^s D_{\mathrm{KL}}(\btheta\Vert\brho^s)
    \implies
    L(\btheta) = \sum_{s\in \mS} \alpha^s D_{KL}(\btheta \Vert \brho^s).
\end{equation*}
In the proof of Proposition~\ref{app:prop.bregman}, we have $(\boldsymbol{\rho^s} \odot \boldsymbol{\lambda}) \propto \widetilde{\mathbf{b}}^s$. Therefore, 
\begin{equation*}
    \brho^s = \frac{\se \odot \widetilde{\mathbf{b}}^s}{(\se)^\intercal \widetilde{\mathbf{b}}^s}
\end{equation*}
is exactly the layer-specific vector defined in the main text.

Furthermore, taking the limit $r \downarrow 0 $ in equation~\eqref{appeq.ces-cru} gives
\begin{equation*}
    \ln(\mathcal{CRU}) = -\min_{\btheta \in \Delta} \sum_{s\in \mS} \alpha^s D_{KL}(\btheta \Vert \brho^s).
\end{equation*}
These are the results in Proposition~\ref{prop.efficiencyloss}.

\section{The case of one private good and one conspicuous network good ($\phi^s<0$)}\label{appsec.examplegg}

In \cite{ghiglinoKeepingNeighborsSocial2010}, there is one private good and one good exhibiting social comparison, and the utility function can be written as
\begin{align*}
    u_i 
    =&
    \sigma \ln(x_i^1) + (1-\sigma) \ln \left(x_i^2 +\alpha S(n_i) 
    \left[x^2_i- \frac{1}{n_i} \sum_{j \in N_i} x_j^2 \right]\right)\\
    =&
    \sigma \ln(x_i^1) + (1-\sigma) \ln \left(x_i^2 - \frac{\alpha S(n_i)}{n_i(1+\alpha S(n_i))} \sum_{j \in N_i} x_j^2 \right) + constant
\end{align*}
where $N_i$ denotes the set of direct neighbors of consumer $i$ and $n_i = |N_i|$. In their setting $S(n_i)$ can be $n_i$ or $1$. In the first case, the number of neighbors has linear effect. In the second case, only the average consumption matters. 

Define $\phi = -\alpha$ and $g_{ij} := \frac{S(n_i)}{n_i(1+\alpha S(n_i))} \mathds{1}_{\{j \in N_i \}}$. Then the utility function is equivalent to that in our main text. Moreover, the network $\bG$ is the same as their definitions for the two specifications of $S(n_i)$. Therefore, we derive the equilibrium using our $\bmu$-approach in the following section.

\subsection{Equilibrium}\label{appsec.ggeqm}

Denote $\bM = [\mathbf{I}_n-\alpha \bG]^{-1}$. Recall that $\bH = \overline{\bM} + \alpha^1\boldsymbol{\eta}^1 \mathbf{1}^\intercal\bM^1$, in this case, we have
\begin{align*}
    \se =& \bH^{-1} \w^1 = 
    \left[\sigma\mathbf{I}_n + (1-\sigma) (\mathbf{I}_n - \etab^2 \mathbf{1}^\intercal ) \bM \right]^{-1} \w^1\\
    =& (\bM)^{-1} \left[\mathbf{I}_n - \sigma \alpha \bG - (1-\sigma) \etab^2 \mathbf{1}^\intercal  \right]^{-1} \w^1\\
    =& (\bM)^{-1} \left[ (\mathbf{I}_n - \sigma \alpha \bG)^{-1} + \frac{(\mathbf{I}_n - \sigma \alpha \bG)^{-1}(1-\sigma) \etab^2 \mathbf{1}^\intercal (\mathbf{I}_n - \sigma \alpha \bG)^{-1}}{1-(1-\sigma)\mathbf{1}^\intercal (\mathbf{I}_n - \sigma \alpha \bG)^{-1}\etab^2 }  \right] \w^1.
\end{align*}
where the last equality follows from the Sherman–Morrison formula.

Using Theorem~\ref{thm.interior}, one can verify that
\begin{align*}
    & p^2 
    = \frac{(1-\sigma) \mathbf{1}^\intercal (\mathbf{I}_n - \sigma \alpha \bG)^{-1} \w^1}{\bar{\omega}^{2}-(1-\sigma) \mathbf{1}^\intercal (\mathbf{I}_n - \sigma \alpha \bG)^{-1}\w^2},\\
    &\bx^{2*} 
    =  \frac{\bar{\omega}^{2}-(1-\sigma) \mathbf{1}^\intercal (\mathbf{I}_n - \sigma \alpha \bG)^{-1}\w^2}{\mathbf{1}^\intercal (\mathbf{I}_n - \sigma \alpha \bG)^{-1} \w^1} (\mathbf{I}_n - \sigma \alpha \bG)^{-1}\w^1
    + (1-\sigma)(\mathbf{I}_n - \sigma \alpha \bG)^{-1}\w^2,
\end{align*}
which is exactly the same as the characterization in \cite{ghiglinoKeepingNeighborsSocial2010}. Therefore, our framework nests the setting studied in \cite{ghiglinoKeepingNeighborsSocial2010}.

As stated in Example VII, multiplexity can reverse the price effects. To see this, the direction of the price effect in Proposition~\ref{prop.priceeffect} is given by
\begin{align*}
    \mathbf{1}^\intercal \bM \bH^{-1} \btau = &
    \mathbf{1}^\intercal \left[ (\mathbf{I}_n - \sigma \alpha \bG)^{-1} + \frac{(\mathbf{I}_n - \sigma \alpha \bG)^{-1}(1-\sigma) \etab^2 \mathbf{1}^\intercal (\mathbf{I}_n - \sigma \alpha \bG)^{-1}}{1-(1-\sigma)\mathbf{1}^\intercal (\mathbf{I}_n - \sigma \alpha \bG)^{-1}\etab^2 }  \right] \btau\\
    \propto & \mathbf{1}^\intercal(\mathbf{I}_n - \sigma \alpha \bG)^{-1} \btau.
\end{align*}
This implies that $\mathbf{c} \propto [\mathbf{I}_n - \sigma \alpha (\bG)^\intercal]^{-1} \mathbf{1}$, which coincides with the results in \cite{ghiglinoKeepingNeighborsSocial2010}. Moreover, when $\alpha = |\phi|$ is small enough, this has the same ranking as $\widetilde{\mathbf{b}}$, the single layer K-B centrality.\footnote{This is indeed the case in Example VII. } However, as shown in Example VII, $\mathbf{c}$ can have different ranking when there are multiple networked goods.

\subsection{Efficiency}

In this setting, the network associated with the private good is empty and hence regular. Therefore, the centrality parallel condition holds if and only if $\bG^\intercal$ is also regular, in which case the equilibrium is efficient. In particular, $\bG^\intercal$ is regular if and only if
\begin{equation*}
    \sum_{k \in \mN} 
    \frac{S(n_k)}{n_k(1+\alpha S(n_k))} 
    \mathds{1}_{\{i \in N_k\}}
    =
    \sum_{k \in \mN} 
    \frac{S(n_k)}{n_k(1+\alpha S(n_k))} 
    \mathds{1}_{\{j \in N_k\}},
    \quad \forall i,j \in \mN .
\end{equation*}
The two sides measure the aggregate influence of consumers $i$ and $j$, respectively. In particular, this condition is satisfied when the original connection network is undirected, unweighted and regular, so that $\sum_{k \in \mN} \mathds{1}_{\{i \in N_k\}} = n_i= n_k$ for any $i,k \in \mN$.

\subsection{Lindahl equilibrium}\label{App_Lindahl}

For the Lindahl equilibrium, we consider the case where the allocation of the initial endowments is identical for these two goods, that is, $\etab^1 = \etab^2 \equiv \etab$. Moreover, all consumers only compare with consumer 1, that is $g_{j1} = 1$ for any $j \neq 1$. As calculated in Section~\ref{appsec.ggeqm},
\begin{align*}
    \se =& 
    (\bM)^{-1} \left[ (\mathbf{I}_n - \sigma \alpha \bG)^{-1} + \frac{(\mathbf{I}_n - \sigma \alpha \bG)^{-1}(1-\sigma) \etab \mathbf{1}^\intercal (\mathbf{I}_n - \sigma \alpha \bG)^{-1}}{1-(1-\sigma)\mathbf{1}^\intercal (\mathbf{I}_n - \sigma \alpha \bG)^{-1}\etab }  \right] \etab \bar{\omega}^1
    \\
    \propto & (\mathbf{I}_n-\alpha \bG) (\mathbf{I}_n - \sigma \alpha \bG)^{-1} \etab
    = (\mathbf{I}_n-\alpha \bG) (\mathbf{I}_n + \sigma \alpha \bG) \etab
    \\
    = &[\mathbf{I}_n - \alpha(1-\sigma)\bG] \etab
\end{align*}
where the last two equalities follow from the fact that $\bG^m = \mathbf{O}$ for $m \geq2$. This property of $\bG$ also implies $\widetilde{\mathbf{b}} = (\mathbf{I}_n-\alpha \bG^\intercal)^{-1} \mathbf{1} = (\mathbf{I}_n+\alpha \bG^\intercal) \mathbf{1}$. Then we have
\begin{equation}
    \frac{\widetilde{b}_1 \mu^*_1}{\widetilde{b}_k\mu^*_k} > \frac{\mu^*_1}{\mu^*_k} > \frac{\eta_1}{\eta_k} = \frac{\mu_1}{\mu_k}, \ \forall k \neq 1
\end{equation}
where $\bmu = \etab$ is the weight in Lindahl equilibrium by Proposition~\ref{prop.lindahl}. This implies
\begin{equation}\label{eq.examplelindahl}
    \mu_1 = \frac{\mu_1}{\sum_{j\in \mathcal{N}} \mu_j} 
    < \frac{\mu^*_1}{\sum_{j\in \mathcal{N}} \mu^*_j}
    < \frac{\widetilde{b}_1\mu^*_1}{\sum_{j\in \mathcal{N}} \widetilde{b}_j\mu^*_j}.
\end{equation}

Denote the utility in the Lindahl equilibrium as $u^L_i$. Using Theorem~\ref{thm.interior} and Proposition~\ref{prop.lindahl}, the utility difference of consumer 1 is equal to
\begin{equation*}
    u^L_1-u^*_1 = \sigma \ln \left( \frac{\mu_1}{\mu^*_1/(\sum_{j\in \mathcal{N}} \mu^*_j)} \right) + (1-\sigma) \ln \left( \frac{\mu_1}{\widetilde{b}_1\mu^*_1/(\sum_{j\in \mathcal{N}} \widetilde{b}_j\mu^*_j)} \right).
\end{equation*}
By the inequalities in (\ref{eq.examplelindahl}), this is negative. Therefore, consumer 1 is strictly worse off under the Lindahl equilibrium than under the competitive equilibrium.

This illustrates that the consumer who imposes substantial negative externalities to others can be worse off under the Lindahl equilibrium. Intuitively, this is because the introduction of personalized prices raises the cost of consuming these externalities. Consequently, although the Lindahl equilibrium restores efficiency, it does not necessarily Pareto dominate the competitive equilibrium.


\section{Comparative statics of other parameters}\label{appsec.cs}
In this section, we derive comparative statics with respect to preference weights $\alpha^s$, externality intensities $\phi^s$, and networks $\bG^s$.

Since preference weights are normalized, we consider a perturbation direction $\btau$ satisfying $\mathbf{1}^\intercal \btau = 0$, and a scalar $l$. The perturbed preference vector is then given by $\boldsymbol{\alpha}(l) = \boldsymbol{\alpha} + l\btau$. Similarly, for the network, we consider a local perturbation $\boldsymbol{\Gamma}^s$ for each layer~$s$ and a common scalar $l$, so that the perturbed network is $\bG^s(l) = \bG^s+l\boldsymbol{\Gamma}^s$. In particular, the change in the externality intensity $\phi^s$ can be viewed as a special case of perturbing layer $s$ with direction $\boldsymbol{\Gamma} = \bG^s/\phi^s$.

We continue to use the notation
\begin{equation*}
    \left. \dot{\mathbf{z}}  := \frac{\partial \mathbf{z}(l)}{\partial l} \right|_{l=0}.
\end{equation*}
Whether $\dot{\mathbf{z}}$ refers to the effect of a change in preference weights or the networks will be clear from the context.

\begin{proposition}\label{appprop.cs}
~
\begin{enumerate}
    \item[(i)] For a perturbation in preference weights:
    \begin{align*}
        &\frac{\dot{p}^{s*}}{p^{s*}}-\frac{\dot{p}^{t*}}{p^{t*}}
        =
        \frac{b^s(\dot{\bmu}^*)}{b^s(\bmu^{*0})}
        -
        \frac{b^t(\dot{\bmu}^*)}{b^t(\bmu^{*0})}
        +
        \frac{\tau^s}{\alpha^s}
        -
        \frac{\tau^t}{\alpha^t},\\
        &\dot{u}_i^*
        =
        \frac{\dot{\mu}^*_i}{\mu^{*0}_i}
        -
        \sum_{s \in \mS}\alpha^s
        \frac{b^s(\dot{\bmu}^*)}{b^s(\bmu^{*0})}
        +
        \sum_{s \in \mS}
        \tau^s\ln\left( \frac{\bar{\omega}^s \mu^{*0}_i}{b^s(\bmu^{*0})} \right),\\
        &\dot{(\ln x^{s*}_i)} = \frac{b^s_i(\dot{\bmu}^*)}{b^s_i(\bmu^{*0})} 
        - \frac{b^s(\dot{\bmu}^*)}{b^s(\bmu^{*0})},
    \end{align*}
    where $\dot{\bmu}^*$ solves
    \begin{equation*}
        \overline{\bM} \dot{\bmu}^* 
        = -\sum_{s \in \mS} \tau^s (\mathbf{I}_n - \etab^s \mathbf{1}^\intercal)\bM^s \bmu^{*0}.
    \end{equation*}

    \item[(ii)] For perturbations in networks:
    \begin{align*}
        &\frac{\dot{p}^{s*}}{p^{s*}}-\frac{\dot{p}^{t*}}{p^{t*}}
        =
        \frac{b^s(\dot{\bmu}^*)+\mathbf{1}^\intercal \dot{\bM}^s \bmu^{*0}}{b^s(\bmu^{*0})}
        -
        \frac{b^t(\dot{\bmu}^*)+\mathbf{1}^\intercal \dot{\bM}^t \bmu^{*0}}{b^t(\bmu^{*0})},\\
        &\dot{u}_i^*
        =
        \frac{\dot{\mu}^*_i}{\mu^{*0}_i}
        -
        \sum_{s \in \mS}\alpha^s
        \frac{b^s(\dot{\bmu}^*)+\mathbf{1}^\intercal \dot{\bM}^s \bmu^{*0}}{b^s(\bmu^{*0})},\\
        &\dot{(\ln x^{s*}_i)} 
        = 
        \frac{b^s_i(\dot{\bmu}^*)+(\dot{\bM}^s \bmu^{*0})_i}{b^s_i(\bmu^{*0})} 
        - \frac{b^s(\dot{\bmu}^*)+\mathbf{1}^\intercal \dot{\bM}^s \bmu^{*0}}{b^s(\bmu^{*0})},
    \end{align*}
    where $\dot{\bM}^s = -\phi^s\bM^s \boldsymbol{\Gamma}^s \bM^s$ and $\dot{\bmu}^*$ solves
    \begin{equation*}
        \overline{\bM} \dot{\bmu}^* 
        = \sum_{s \in \mS} \alpha^s\phi^s (\mathbf{I}_n - \etab^s \mathbf{1}^\intercal) \bM^s \boldsymbol{\Gamma}^s \bM^s\bmu^{*0}.
    \end{equation*}
\end{enumerate}
\end{proposition}

\begin{proof}
(i) For a perturbation in preference weights, directly differentiating equations~\eqref{eq.equix}-\eqref{eq.equiu} with respect to $l$ gives the expressions in the proposition. To pin down $\dot{\bmu}^*$, differentiating $\overline{\bM}(l) \bmu(l) = 0$ gives
\begin{equation*}
    \overline{\bM} \dot{\bmu}^* = - \dot{\overline{\bM}} \bmu^{*0}
    = -\sum_{s \in \mS} \tau^s (\mathbf{I}_n - \etab^s \mathbf{1}^\intercal)\bM^s \bmu^{*0}.
\end{equation*}

(ii) For perturbations in networks, we have
\begin{align*}
    &\dot{\bM}^s = -\phi^s\bM^s \boldsymbol{\Gamma}^s \bM^s,\\
    &\dot{\{b^s(\se)\}} = b^s(\dot{\bmu}^*) + \mathbf{1}^\intercal \dot{\bM}^s \bmu^{*0},\\
    &\dot{\{\mathbf{b}^s(\se)\}} = \mathbf{b}^s(\dot{\bmu}^*) + \dot{\bM}^s \bmu^{*0}.
\end{align*}
Directly differentiating equations~\eqref{eq.equix}-\eqref{eq.equiu} with respect to $l$ and using these equations gives the expressions in the proposition. $\dot{\bmu}^*$ is determined similarly by
\begin{equation*}
    \overline{\bM} \dot{\bmu}^* 
    = - \dot{\overline{\bM}} \bmu^{*0}
    =  - \sum_{s \in \mS} \alpha^s (\mathbf{I}_n - \etab^s \mathbf{1}^\intercal)\dot{\bM}^s \bmu^{*0}
    = \sum_{s \in \mS} \alpha^s\phi^s (\mathbf{I}_n - \etab^s \mathbf{1}^\intercal) \bM^s \boldsymbol{\Gamma}^s \bM^s\bmu^{*0}.
\end{equation*} 
This proves Proposition \ref{appprop.cs}.
\end{proof}

\end{document}